\documentclass[letterpaper, 11pt]{article}
\usepackage[top=1in, bottom=1in,margin=1in]{geometry}

\usepackage{amsmath,amsthm,amsfonts, amssymb}
\theoremstyle{plain}
\newtheorem{thm}{Theorem}
\newtheorem{lem}{Lemma}

\theoremstyle{definition}
\newtheorem{defn}{Definition}

\newtheorem{remark}{Remark}

\usepackage{longtable}

\usepackage{graphicx}
\graphicspath{{./fig/}}

\usepackage[labelformat=simple]{subcaption}

\usepackage{authblk}

\usepackage[toc,page]{appendix}

\usepackage{hyperref}
\usepackage[usenames, dvipsnames]{color}

\title{Learning the LMP-Load Coupling From Data:\\ A Support Vector Machine Based Approach}
\author[1]{Xinbo Geng\thanks{gengxbtamu@tamu.edu}}
\author[1]{Le Xie\thanks{le.xie@tamu.edu}}
\affil[1]{Department of Electrical and Computer Engineering, Texas A\&M University}

\begin{document}
\maketitle

\begin{abstract}
This paper investigates the fundamental coupling between loads and locational marginal prices (LMPs) in security-constrained economic dispatch (SCED). 
Theoretical analysis based on multi-parametric programming theory points out the unique one-to-one mapping between load and LMP vectors. 
Such one-to-one mapping is depicted by the concept of system pattern region (SPR) and identifying SPRs is the key to understanding the LMP-load coupling.
Built upon the characteristics of SPRs, the SPR identification problem is modeled as a classification problem from a market participant's viewpoint, and a Support Vector Machine based data-driven approach is proposed. 
It is shown that even without the knowledge of system topology and parameters, the SPRs can be estimated by learning from historical load and price data.
Visualization and illustration of the proposed data-driven approach are performed on a 3-bus system as well as the IEEE 118-bus system. 
\end{abstract}

\section{Introduction} % (fold)
\label{sec:introduction}
A fundamental issue with electricity market operation is to understand the impact of operating conditions (e.g. load levels at each bus) on the locational marginal prices (LMPs).  
This paper examines this key issue of the relationship between nodal load levels and LMPs. 
This issue is further compounded by the increasing levels of demand response and variable resources in the grid.

In the power systems literature, reference \cite{Conejo2005} is among the pioneering works that uses perturbation techniques to compute the sensitivities of the dual variables in SCED (e.g. LMPs) with respect to parameters (e.g. the nodal load levels).
% To understand how the LMPs change as parameters in the SCED change, \cite{Conejo2005} proposed a perturbation technique that calculates the sensitivities of dual variables (e.g. LMPs) with respect to parameters in SCED.
This sensitivity calculation method is widely used in subsequent researches. However, this approach is valid only for small changes and the marginal generator stays the same. Reference \cite{Li2007} observed the ``step changes'' of LMPs with respect to increasing system load level and discovered that new binding constraints (transmission or generation) are the reason of the ``step changes''.
This is followed by further analysis on identifying the critical load levels (CLLs) that trigger such step changes of LMPs \cite{Li2009}, \cite{Bo2009}, \cite{Bo2011}. 
% Their following papers \cite{Li2009, Bo2009, Bo2011} attempted to identify critical load levels (CLLs) where the ``step changes'' happen. 
This line of work assumes that the system load change is distributed to each bus proportional to the base case load, which, in many instances, do not necessarily represent the real-world situations. 
Reference \cite{Zhou2011} analyzed this problem using quadratic-linear programming (QLP) and the concepts of \emph{system patterns} and \emph{system pattern regions (SPRs)} were first introduced. The SPRs depict the relationship between loads and LMPs in the whole load space, which is not confined in a small neighborhood of an operating point or constrained by a specific load distribution pattern. 
This paper is inspired by \cite{Zhou2011} but focuses on the case of piecewise linear generation costs, instead of the quadratic cost case in \cite{Zhou2011}.
The reason that we study the piecewise linear cost case is that piecewise linear cost curves are often quite representative of the market practice in the real world. In addition, some new theoretical results based on piecewise linear cost curves are derived, and are generalizable towards quadratic cost cases.

Characterizing the SPRs would provide important insights to both system operators and market participants. Reference \cite{Ji} advances the theory of SPR from system operator's perspective where the knowledge of system topology and parameters is available. For market participants, such knowledge is not necessarily available. Our previous work \cite{Geng2015} examines the issue from market participant's viewpoint and applies the geometric features of SPRs to identify them. 

This paper significantly advances our previous work by (1) completing the theoretical characterization of SPRs as a function of nodal load levels; (2) proposing a computational algorithm to identify SPRs using historical data; (3) introducing the posterior probabilities of SPRs with the presence of uncertain system parameters such as transmission limits; and (4) extending the algorithm to consider practical factors such as partial load information and loss component of LMPs.
 
% \begin{itemize}
%   \item We establish the mappings among loads, system status and LMPs and more characteristics of SPRs are revealed and discussed in the Monte-Carlo simulation. 
%   \item Based on the theoretical analysis,  A data-driven approach is proposed, and it does not require confidential system information. Compared with other methods, the data-driven approach reduces the data requirements and exploits power system background to improve performance.
% \end{itemize}

The rest of the paper is organized as follows. 
Section \ref{sec:theoretical_analysis} provides the analysis of LMP-load coupling in SCED problem from the viewpoint of MLP theory, with an illustrative example. 
Section \ref{sec:sprs_with_varying_parameters} illustrates the changes of SPRs given changes of system parameters such as transmission limits. 
Based on the theoretical analysis, a data-driven algorithm for market participants to identify SPRs is described in Section \ref{sec:extended_data_driven_approach}. 
Section \ref{sec:case_study} illustrates the performance of the algorithm on the IEEE 118-bus system.
Section explores the impact of nodal load information, and Section \ref{sec:discussions} provides critical assessment
of the proposed method. Concluding remarks and future works
are presented in Section \ref{sec:conclusion}.

% ased on the theoretical results, the extended data-driven approach is depicted in  Section \ref{sec:extended_data_driven_approach}. Section \ref{sec:case_study} gives an illustrative example of the data-driven approach as well as the case study on the IEEE 118 Bus System. The concluding remarks and future work are in Section \ref{sec:conclusion}.

% Section \ref{sec:discussion} discusses the relationship of this paper with others and compares the data-driven approach with state-of-the-art forecast techniques. 

% section introduction (end)

\section{Theoretical Analysis} % (fold)
\label{sec:theoretical_analysis}
\subsection{Notations} % (fold)
\label{sub:notations}
The notations of this paper are summarized below: mathematical symbols in hollowed-out shapes (e.g. $\mathbb{R}$) represent spaces and symbols in Calligra font (e.g. $\mathcal{S_{\pi}}$) stand for sets. The superscript ``$^*$'' indicates the variable is optimal, ``$\hat{\quad}$'' denotes estimated values (e.g. $\hat{\lambda}$). Variables with ``$\bar{\quad}$'' are expectations or average values (e.g. $\bar{\lambda}$). ``$^\intercal $'' denotes the transpose of a vector or matrix (e.g. $\mathbf{1}_{n}^\intercal $). The subscript ``$_i$'' represents the $i$th element of the vector (e.g. $P_{G_i}$), and the superscript ``$^{(i)}$'' represents the $i$th element in a set (e.g. $P_D^{(i)} $).
The vector of $n\times 1$ ones, matrix of $m\times n$ zeros and the $n\times n$ identity matrix are denoted by $\mathbf{1}_{n}$ and $\mathbf{0}_{m\times n}$ and $\mathbf{I}_n$ respectively.

% subsection notations (end)

\subsection{Security Constrained Economic Dispatch} % (fold)
\label{par:reivew_sced}
In real-time energy market operations, the LMPs are the results from the security-constrained economic dispatch (SCED), which is formulated as follows: 
% In this paper, the lossless-DCOPF based security constrained economic dispatch (SCED)  problem is discussed. Its formulation is summarized in Eqn. (\ref{eqn:ED_primal}). 
\begin{subequations}
\label{eqn:ED_primal}
\begin{align}
& \underset{P_G^{(k)}}{\min} && \sum_{i=1}^{n_b} c_i(P_{G_i}^{(k)}) & \label{eqn:static_sced_obj} \\
& \text{s.t.} && \sum_{i = 1}^{n_b}{P_{G_i}^{(k)}} = \sum_{j=1}^{n_b}{P_{D_j}^{(k)}} & & :\lambda_1 \label{eqn:static_sced_balance} \\
& && -F^+ \le H(P_G^{(k)}-P_D^{(k)}) \le F^+ & & :\mu^+, \mu^- \label{eqn:static_sced_transmission}\\
% & && -H\cdot (P_G-P_D) \le F & & :\mu^-\\
& && P_G^- \le P_G^{(k)} \le P_G^+ & & : \eta^+, \eta^- \label{eqn:static_sced_generation}
% & && -P_G \le -P_G^{min} & & : \eta^- 
\end{align}
\end{subequations}
where $P_G^{(k)}$ is the generation vector at time $k$, and $P_D^{(k)}$ is the load vector at time $k$. We assume there are both generation and load at each bus. Let $n_b$ denote the number of buses and $n_l$ denote the number of transmission lines, then $P_G^{(k)}, P_D^{(k)} \in \mathbb{R}^{n_b}$. 
$H \in \mathbb{R}^{n_l\times n_b} $ is the shift factor matrix.

This formulation considers each snapshot independently, therefore it is called \emph{static SCED} in this paper. For simplicity, we write $P_G^{(k)}$ and $P_D^{(k)}$ as $P_G$ and $P_D$ when discussing the static SCED.
% with an all-zero column corresponding to the slack bus.
% In Eqn. (\ref{eqn:ED_primal}), bus 1 is assumed to be the slack bus. 
 
% However, the analysis here could be extended to formulations with losses considered \footnote{ \textbf{putting this sentence in the analysis part is better? ``the analysis here'' is not clear to me}}. 
 % In Eqn. (\ref{eqn:ED_primal}), bus\#1 is assumed to be the slack bus.

The objective of SCED is to minimize the total generation cost and satisfy the transmission and generation capacity constraints while keeping the real-time balance between supply and demand. The generation cost function $c_i(P_{G_i}^{(k)})$ of generator $i$ is increasing and convex, and it is usually regarded as a  quadratic function or approximated by a piecewise linear function. 
To better reflect the current practice in electricity markets, this paper studies the SCED problem with piecewise linear generator bidding functions. And for the consideration of simplicity, the simplest form, i.e. $\sum_{i=1}^{n_b} c_i(P_{G_i}) = c^\intercal  P_G$ is being considered in this paper. 

A fundamental concept in electricity markets is the \emph{Locational Marginal Price. The LMP $\lambda_i$ at bus $i$ is defined as the change of total system cost if the demand at node $i$ is increased by 1 unit \cite{Kirschen2005}.} According to \cite{Wu1996}, the LMP vector $\lambda$ can be calculated by the following equation:
\begin{equation}
\label{eqn:LMP}
\lambda = \lambda_1 \mathbf{1}_{n_b} + H^\intercal (\mu^+ - \mu^-)
\end{equation}
% \begin{equation}
%   \lambda_i = \frac{\partial f^*}{\partial P_{D_i}}. \text{ where: } f^* \text{ :total system cost}, P_{D_i}: \text{load at bus $i$}
% \end{equation}
% The associated Lagrangian Multipliers $[\lambda_1; \mu^+; \mu^-; \eta^+; \eta^-] $ play a pivotal role in market pricing, and the vector of LMPs $\lambda$ can be calculated as follows :

For better understanding, we start with the simplest case of
static SCED. More elaborated SCED formulations are in Section
\ref{sub:ramp_constraints}. Since the line losses are not explicitly modeled the
SCED formulation, the LMPs in this paper do not contain the
loss components. Further discussions on the loss component
are in Section \ref{sub:lmps_with_loss_components}.

\subsection{SCED Analysis via MLP} % (fold)
\label{sub:SCED_analysis_via_MLP}
In real-world market operations, the parameters associated with the SCED above are typically time-varying. Therefore, it is essential to 
% In the reality, the parameters of optimization problems are either estimated or stochastic by their nature. It is essential to
understand the effects of parameters on the optimality of the problem. Multi-parametric Programming (MP) problem aims at exploring the characteristics of an optimization problem which depending on a \emph{vector of parameters} \cite{Borrelli2003}. Multi-parametric Linear Programming (MLP) theory, which is the foundation of this paper, pays special attention to Linear Programming (LP) problems.

In this paper, we would like to understand the impact of parameters (i.e., load levels, line capacities, etc) on the outcome of SCED (namely, the prices). We pose the problem in view of MLP, and analyze the theoretical properties.

In the reality, the LMP vector depends upon a number of factors, including: (1) the loads in the system; (2) line flow limits; (3) ramp constraints; (4) generation offer prices; (5) topology of the system; (6) unit commitment results. We first focus on the relationship between loads and LMPs assuming the other five factors remain unchanged; then Section \ref{sec:sprs_with_varying_parameters} takes the line flow limits and ramp constraints into account; the influence of generation offer prices is explored in Section \ref{sub:on_generation_offer_prices}. Future work will investigate the impacts of unit commitment results and the system topology changes on the prices.

Consider the static SCED in the standard MLP form\footnote{In other references (e.g. \cite{Adler1992a, Gal1972}), the primal form of the MLP problem is different. For the consideration of convenience of analyzing SCED problem, we follow the formulations in \cite{Borrelli2003}. Those two forms are interchangeable.}:
\begin{subequations}
\label{eqn:SCED_MLP_Primal}
\begin{align}
  \text{Primal:} & \min\{c^\intercal P_G: AP_G + s = b + WP_D, s\ge 0\} \\
  \label{eqn:SCED_MLP_Dual}
  \text{Dual:} &\max\{-(b+WP_D)^\intercal y: A^\intercal y=-c, y\ge 0\} 
\end{align}
\end{subequations}
where:
\begin{eqnarray}
\label{eqn:details_of_the_A_W_b}
  A = 
  \begin{bmatrix}
     \mathbf{1}_{n_b}^\intercal  \\ -\mathbf{1}_{n_b}^\intercal  \\
    H\\ - H\\
    \mathbf{I}_{n_b} \\
    -\mathbf{I}_{n_b}
  \end{bmatrix},
      b = 
  \begin{bmatrix}
    0 \\0  \\ F^+ \\ -F^+ \\ P_G^+ \\ -P_G^-
  \end{bmatrix},
  W = 
  \begin{bmatrix}
    \mathbf{1}_{n_b}^\intercal \\
    -\mathbf{1}_{n_b}^\intercal\\ 
    H \\ -H \\
    \mathbf{0}_{n_b \times n_b} \\ \mathbf{0}_{n_b \times n_b}
  \end{bmatrix}
\end{eqnarray}
The load vector $P_D$ is the vector of parameters $\theta$, and the load space $\mathbb{D}$ is the parameter space $\Theta$. Since not every $P_D$ in the load space leads to a feasible SCED problem, $\mathcal{D} \in \mathbb{D}$  denotes the set of all feasible vectors of loads. \cite{Gal1972} shows that $\mathcal{D}$ is a convex polyhedron in $\mathbb{D}$.

\begin{defn}[Optimal Partition/System Pattern]
\label{defn:optimal_partition}
For a load vector $P_D \in \mathcal{D}$, we could find a finite optimal solution $P_G^*$ and $s^*$. Let $\mathcal{J} = \{1,2,\cdots, n_c\}$ denote the index set of constraints where $n_c = 2+2n_l+2n_g$ for Eqn. (\ref{eqn:SCED_MLP_Primal}). The \emph{optimal partition} $\pi = (\mathcal{B},\mathcal{N})$ of the set $\mathcal{J}$ is defined as follows:
\begin{subequations}
\begin{align}
  \mathcal{B}(P_D) &:= \{i: s_i^* = 0 \text{ for } P_D \in \mathcal{D}  \} \\
  \mathcal{N}(P_D) &:= \{j: s_j^* > 0 \text{ for } P_D \in \mathcal{D} \} 
\end{align}
\end{subequations}
Or in the dual form:
\begin{subequations}
\begin{align}
  \mathcal{B}(P_D) &:= \{i: y_i^* > 0 \text{ for } P_D \in \mathcal{D}  \} \\
  \mathcal{N}(P_D) &:= \{j: y_j^* = 0 \text{ for } P_D \in \mathcal{D} \} 
\end{align}
\end{subequations}
Obviously, $\mathcal{B}\cap \mathcal{N} = \emptyset$ and $\mathcal{B}\cup \mathcal{N} = \mathcal{J}$. The optimal partition $\pi = (\mathcal{B}, \mathcal{N})$ divides the index set into two parts: binding constraints $\mathcal{B}$ and non-binding constraints $\mathcal{N}$.
In SCED, the \emph{optimal partition} represents the status of the system (e.g. congested lines, marginal generators), and is called \emph{system pattern}.
% \footnote{The name ``system pattern'' is selected because a similar concept was defined in \cite{Zhou2011}.}accordingly in this paper.
\end{defn}

\begin{defn}[Critical Region/System Pattern Region]
The concept \emph{critical region} refers to the set of vectors of parameters which lead to the same optimal partition (system pattern) $\pi = (\mathcal{B}_{\pi}, \mathcal{N}_{\pi})$:
\begin{equation}
\mathcal{S}_{\pi} := \{P_D \in \mathcal{D}: \mathcal{B}(P_D) = \mathcal{B}_{\pi}\}
\end{equation}
For the consideration of consistency, the \emph{critical region} is called \emph{system pattern region} (SPR) in this paper.
\end{defn}
According to the definitions, each SPR is one-to-one mapped to a system pattern, the SPRs are therefore disjoint and the union of all the SPRs is the feasible set of vectors of loads: $\cup_i \mathcal{S}_{\pi_i} = \mathcal{D}$.
All the SPRs together represent a specific partition of the load space. The features of SPRs, which directly inherit from critical regions in MLP theory, are summarized as follows:
\begin{thm}
\label{thm:convexSPR}
The load space could be decomposed into many SPRs. Each SPR is a convex polytope. The relative interiors of SPRs are disjoint convex sets and each corresponds to a unique system pattern \cite{Zhou2011}. 
There exists a separating hyperplane between any two SPRs \cite{Geng2015}.
\end{thm}

\begin{lem}[Complementary Slackness]
\label{lem:complementary_slackness}
According to complementary slackness:
\begin{subequations}
\label{eqn:complementary_slackness}
\begin{align}
A_{\mathcal{B}} P_G^* &= (b+WP_D)_{\mathcal{B}} \label{eqn:B_primal}\\
A_{\mathcal{N}}P_G^* &< (b+WP_D)_{\mathcal{N}} \label{eqn:N_primal}\\
A_{\mathcal{B}}^\intercal y_{\mathcal{B}} &= -c, y_{\mathcal{B}} > 0 \label{eqn:determine_y}\\
y_{\mathcal{N}} &= 0  
\end{align}
\end{subequations}
where the $(\cdot)_{\mathcal{B}}$ is the sub-matrix or the sub-vector whose row indices are in set $\mathcal{B}$, same meaning applies for $(\cdot)_{\mathcal{N}}$.
\end{lem}

\begin{remark}
The supply-demand balance equality constraint in is rewritten as two inequalities in Eqn. (\ref{eqn:SCED_MLP_Primal}). These two inequalities will always be binding and appear in the binding constraint set $\mathcal{B}$ at the same time. One of them is redundant and therefore eliminated from the set $\mathcal{B}$. \emph{In the remaining part of the paper, set $\mathcal{B}$ denotes the set after elimination. }
\end{remark}
\begin{remark}
If the problem is not degenerate, the cardinality of binding constraint set $\mathcal{B}$ is equal to the number of decision variables (i.e. number of generators $n_g$) \footnote{This is consistent with the statement that the number of marginal generators equals to the number of congested lines plus one.}.
The matrix $A_{\mathcal{B}}$ is invertible and $P_G^*$ is uniquely determined by $A_{\mathcal{B}}^{-1}(b+WP_D)_{\mathcal{B}}$.
\end{remark}

\begin{remark}
SCED problems with different generation costs will have different SPRs.
For a system pattern $\pi = (\mathcal{B},\mathcal{N})$, its SPR would remain the same as
long as the generation cost vector $c$ satisfies Eqn. (\ref{eqn:determine_y}).
\end{remark}

\begin{lem}
Within each SPR, the vector of LMPs is unique \cite{Ji}\cite{Geng2015}.
\end{lem}
The proof of this lemma follows Eqn. (\ref{eqn:determine_y}) (dual form of system pattern definition).
% \footnote{$y_{\mathcal{B}} = A_{\mathcal{B}}^{-1}c$, \textbf{and this paragraph is not clear}.}
Since the system pattern $\pi$ is unique within an SPR $S_{\pi}$, therefore the solution $y^*$ is unique for any $P_D \in \mathcal{S}_{\pi}$
And the vector of LMPs can be calculated using Eqn. (\ref{eqn:LMP}). This lemma also illustrates that the LMP vectors are discrete by nature in the case of linear costs.

\begin{thm}
\label{thm:diff_LMPs}
If the SCED problem is not degenerate, then different SPRs have different LMP vectors.
\end{thm}
The proof of Theorem \ref{thm:diff_LMPs} turns out to be non-trivial, and is described as follows. 
If two SPRs have the same LMP: $\lambda^{(i)} = \lambda^{(j)}$, their energy components are the same because of the entry-wise equality, then Eqn. (\ref{eqn:LMP}) suggests that the congestion components should also be the same: $H^\intercal (\mu^{(i)} - \mu^{(j)}) = 0$. Given the fact that the null space of $H^\intercal$ is always non-empty\footnote{$\text{dim}(N(H^\intercal)) = n_l - n_b +1 \ge 0$. The equality holds if and only if the topology of the system is a tree, where $n_l = n_b -1$.}, a critical question arises: ``is it possible that $\mu^{(i)} - \mu^{(j)}$ belongs to the null space of $H^\intercal$?'' Or equivalently, ``is it possible that different congestion patterns have the same LMP vector?''
We show that the answer is ``no''. A complete proof of the theorem is provided in Appendix \ref{sec:proof_diff_LMP_diff_SPR}.

\subsection{An Illustrative Example} % (fold)
\label{sub:an_illustrative_example}
The 3-bus system in Fig. \ref{fig:3Bus2GeneSystem} serves as an illustrative example in this paper. It was first analyzed using the Multi-Parametric Toolbox 3.0 (MPT 3.0) \cite{Herceg2013}, results are shown in Fig. \ref{fig:SPR_3bus_mpt}. A Monte-Carlo simulation is conducted, with load vectors colored according to their LMPs. The theoretical results are verified by the Monte-Carlo simulation results.
Notice that $P_{D_2}$ and $P_{D_3}$ could be negative. This is for the consideration of renewable resources in the system, which are typically considered as negative loads. 

\begin{figure}[htbp]
  \centering
  \includegraphics[width=0.6\linewidth]{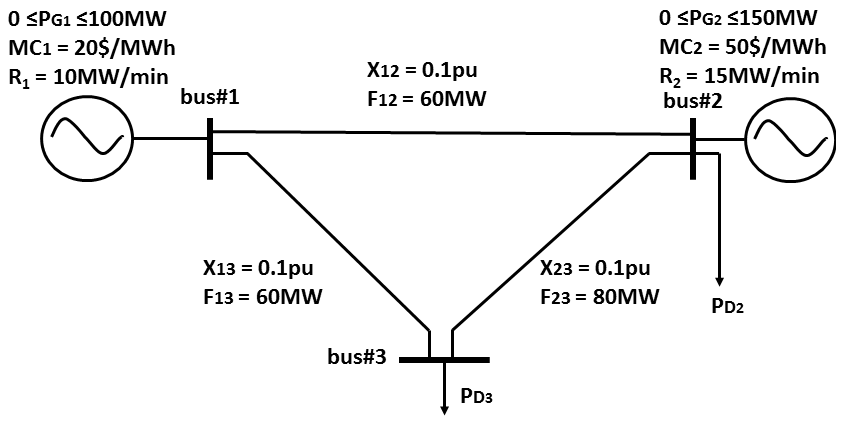}
  \caption{3-bus System}
  \label{fig:3Bus2GeneSystem}
\end{figure}

\begin{figure}[htbp]
  \centering
  \begin{subfigure}[t]{0.49\linewidth}
  \centering
  \includegraphics[width=\linewidth]{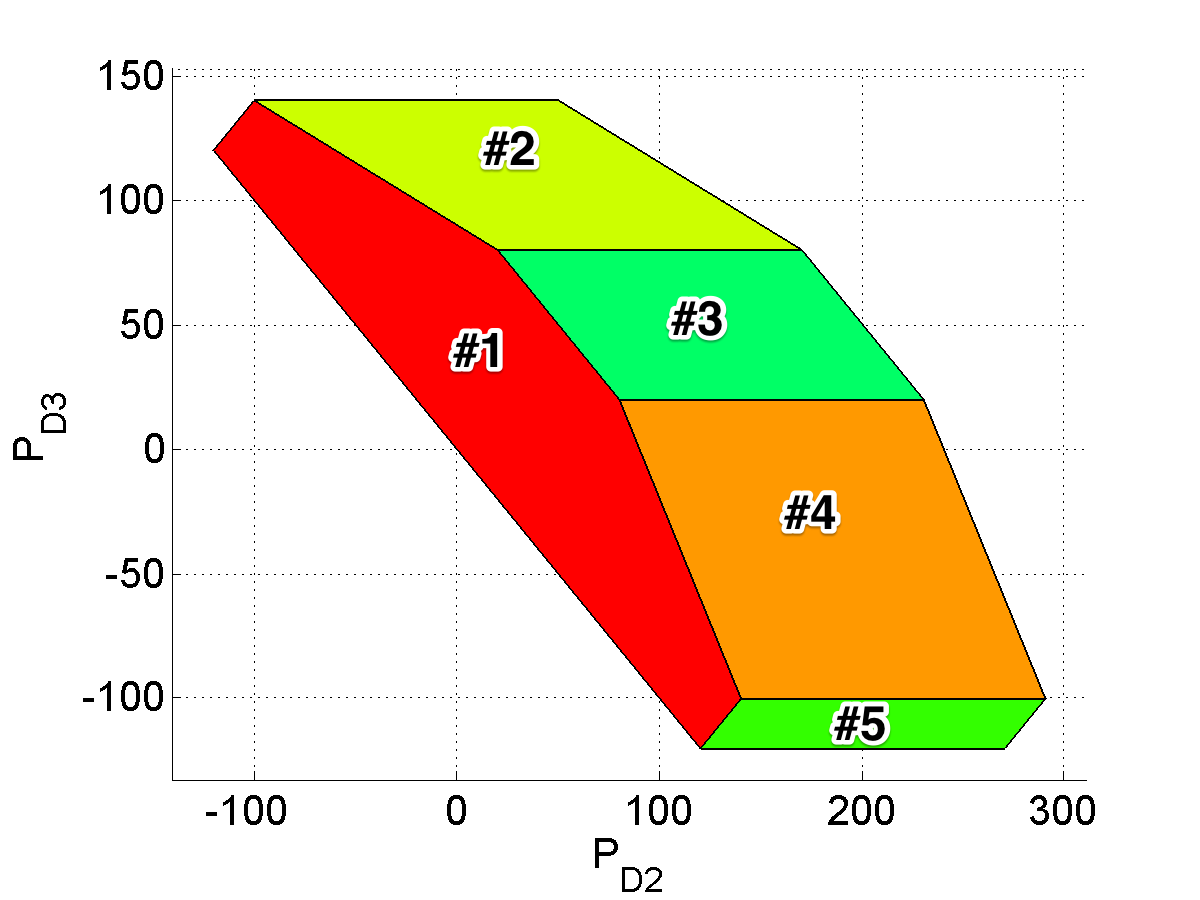}
  \caption{Theoretical Results Using MPT 3.0}
  \label{fig:SPR_3bus_mpt}
  \end{subfigure}
  \begin{subfigure}[t]{0.49\linewidth}
  \centering
  \includegraphics[width=\linewidth]{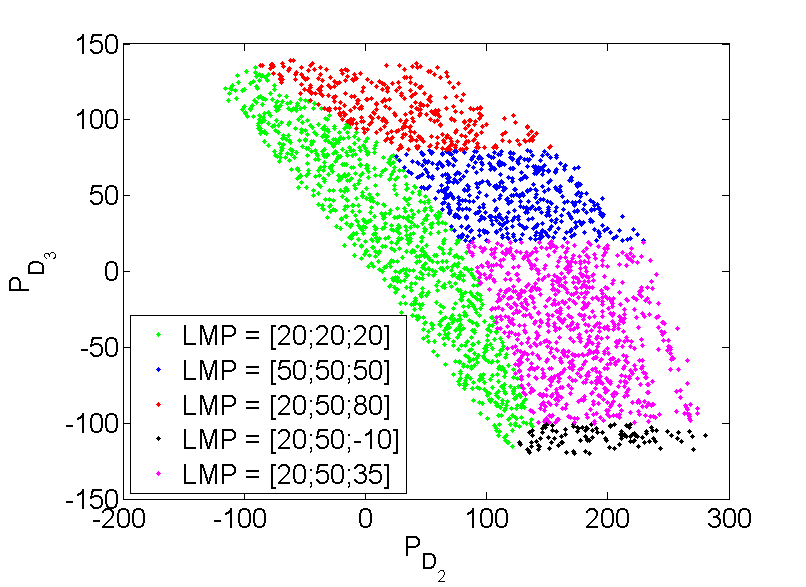} 
  \caption{Monte-Carlo Simulation}
  \label{fig:SPR_3busNO_mc}
  \end{subfigure}
  \caption{SPRs of the 3-bus System (Static SCED)}
\end{figure}  
% subsubsection an_illustrative_example (end)

\section{SPRs with Varying Parameters} % (fold)
\label{sec:sprs_with_varying_parameters}

Section \ref{sub:SCED_analysis_via_MLP} shows construction properties of the load space with fixed parameters of the system (e.g. transmission constraints). However, these parameters might be time-varying due to reasons like dynamic ratings or active ramping constraints. This subsection reveals more features of SPRs with respect to varying factors in the system.
% \begin{remark}
% \label{rmk:SPR_cost}
% For a system pattern $\pi = (\mathcal{B}, \mathcal{N})$, its SPR would remain the same as long as the vector of costs $c$ satisfies:
% \begin{equation}
% \label{eqn:conditions_c}
% A_{\mathcal{B}}^{-1} c < 0
% \end{equation}
% This is a direct conclusion from Eqn. (\ref{eqn:determine_y}).
% \end{remark}

\begin{lem}[Analytical Form of SPRs]
Let $\mathbf{I}_{\mathcal{B}} \cdot (b+WP_D)$ represent the sub-vector $(b+WP_D)_{\mathcal{B}}$, where $\mathbf{I}_{\mathcal{B}}$ is the sub-matrix of the identity matrix  whose row indices are in set $\mathcal{B}$.
Then the analytical form of the SPRs could be solved from Eqn. (\ref{eqn:B_primal}) and Eqn. (\ref{eqn:N_primal}) as follows:
\begin{equation}
\label{eqn:SPR_analytical_form}
(\mathbf{I}_{\mathcal{N}}A\cdot (\mathbf{I}_{\mathcal{B}}A)^{-1}\mathbf{I}_{\mathcal{B}}-\mathbf{I}_{\mathcal{N}})(b+W P_D) < 0 
  % (A_{\mathcal{N}} A_{\mathcal{B}}^{-1}\mathbf{I}_{\mathcal{B}} - \mathbf{I}_{\mathcal{N}})(b+WP_D) < 0
\end{equation}
\end{lem}
We can calculate the analytical expressions of the SPRs using Eqn. (\ref{eqn:SPR_analytical_form}). An illustrative example with complete details is provided in Appendix \ref{sec:analytical_form_of_the_sprs_of_the_3_bus_system_in_fig_}.

\begin{remark}
Eqn. (\ref{eqn:SPR_analytical_form}) could be written as:
\begin{equation}
\label{eqn:SPR_shape}
  % (A_{\mathcal{N}} A_{\mathcal{B}}^{-1}\mathbf{I}_{\mathcal{B}} - \mathbf{I}_{\mathcal{N}}) WP_D < (A_{\mathcal{N}} A_{\mathcal{B}}^{-1}\mathbf{I}_{\mathcal{B}} - \mathbf{I}_{\mathcal{N}})b
  (\mathbf{I}_{\mathcal{N}}A (\mathbf{I}_{\mathcal{B}}A)^{-1}\mathbf{I}_{\mathcal{B}}-\mathbf{I}_{\mathcal{N}})\cdot W P_D < (\mathbf{I}_{\mathcal{N}}-\mathbf{I}_{\mathcal{N}}A (\mathbf{I}_{\mathcal{B}}A)^{-1}\mathbf{I}_{\mathcal{B}})b
\end{equation}
This indicates the shape of the SPR $\mathcal{S}_{\pi}$ only depends on two factors: (1) the corresponding system pattern $\pi = (\mathcal{B},\mathcal{N})$; (2) matrices $A$ and $W$, namely the shift factor matrix $H$ according to Eqn. (\ref{eqn:details_of_the_A_W_b}).
Small changes of vector $b$ only parallel-shift the SPRs' boundaries.
\end{remark}
% subsection more_characteristics_of_sprs (end)

\subsection{Dynamic Line Rating} % (fold)
\label{sub:dynamic_line_rating}
\emph{Dynamic line rating} (DLR), contrary to the \emph{static line rating} (SLR),  refers to the technology that optimizes the transmission capacity based on the real-time conditions such as ambient temperature and wind speed \cite{douglass1996real}.
It is considered to be more adaptive in maximizing the line potential while keeping the secure grid operation.

From dispatch point of view, DLR can be represented by the changes of transmission limits $F^+$ in Eqn. (\ref{eqn:static_sced_transmission}). It changes the vector $b$ in Eqn. (\ref{eqn:details_of_the_A_W_b}) and thus translate the boundaries of SPRs.

The 3-bus system in Fig. \ref{fig:3Bus2GeneSystem} with different transmission limits is analyzed via MPT 3.0. Compared with the standard transmission limits $[60;60;80]$, when we increase the limits by $10\%$ (Fig. \ref{fig:line_11}), SPR \#3 expands but SPR \#1, \#2 and \#4 shrink; when we decrease the limits by $10\%$ (Fig. \ref{fig:line_09}), SPR \#3 shrinks but SPR \#1, \#2 and \#4 expand.
This verifies the claim that dynamic line ratings only shift the boundaries without altering the shapes of SPRs. 
The implication of having DLR is that SPRs in Fig. \ref{fig:spr_3bus_dlr_mc} are overlapping instead of completely separable in Fig. \ref{fig:SPR_3busNO_mc}. Details of the Monte-Carlo simulation are provided in Section \ref{ssub:3_bus_system_dlr}.

\begin{figure}[htbp]
  \centering
  % \subfloat[Primal Solution]{  \includegraphics[width=0.5\linewidth]{./fig/3Bus2GeneSystem_f_[20;50;100]Primal_Solution.png} \label{fig:SPR_3bus_primal} }
  % \subfloat[Optimal Value Function]{  \includegraphics[width=0.5\linewidth]{./fig/3Bus2GeneSystem_f_[20;50;100]objective_value.png} \label{fig:SPR_3bus_obj}  } \\
  % \subfloat[ c= (20, 50, 100) ]{  \includegraphics[width=0.5\linewidth]{./fig/3BusSystem3Gene.png} \label{fig:ssssssss} }
  \begin{subfigure}[t]{0.49\linewidth}
  \centering
  \includegraphics[width=\linewidth]{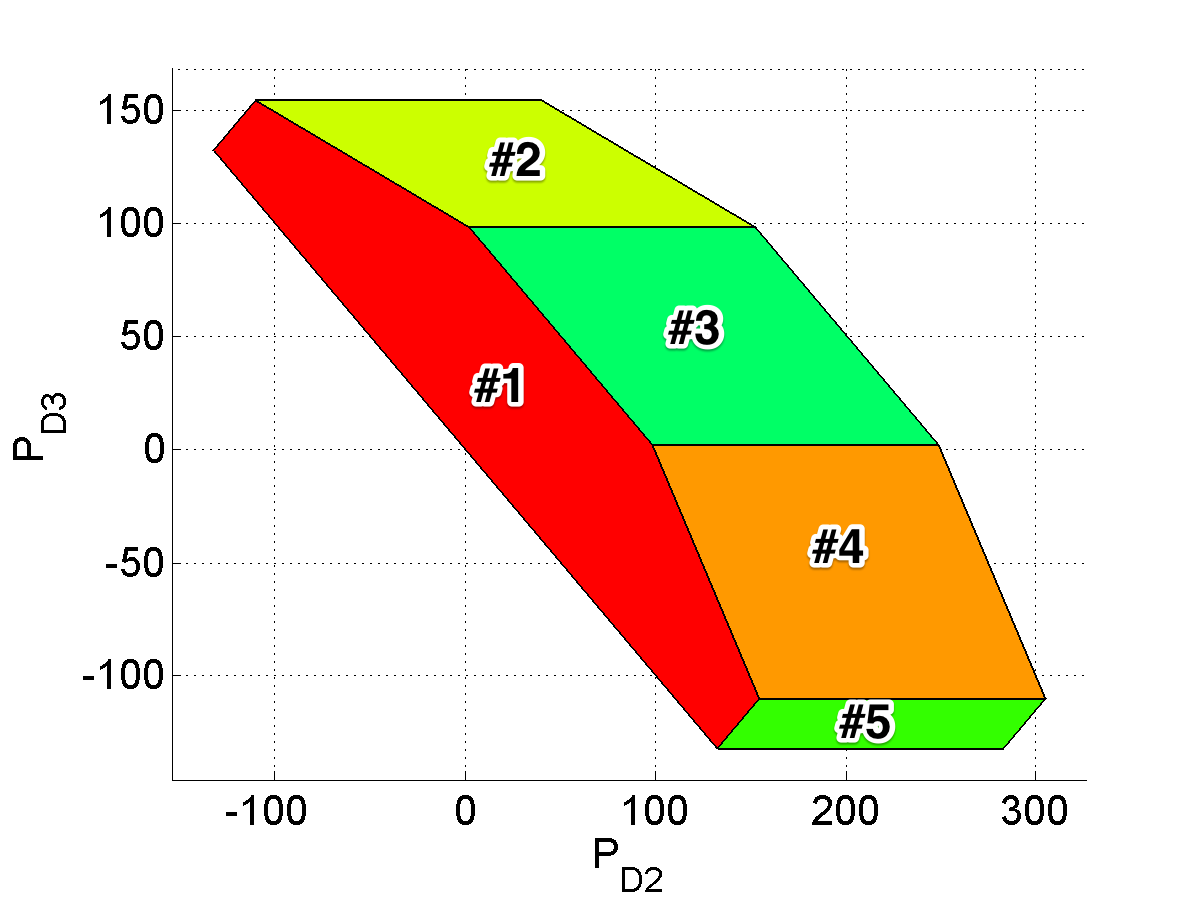} 
  \caption{Line Limits: (66, 66, 88)}
  \label{fig:line_11}
  \end{subfigure}
  \begin{subfigure}[t]{0.49\linewidth}
  \centering
  \includegraphics[width=\linewidth]{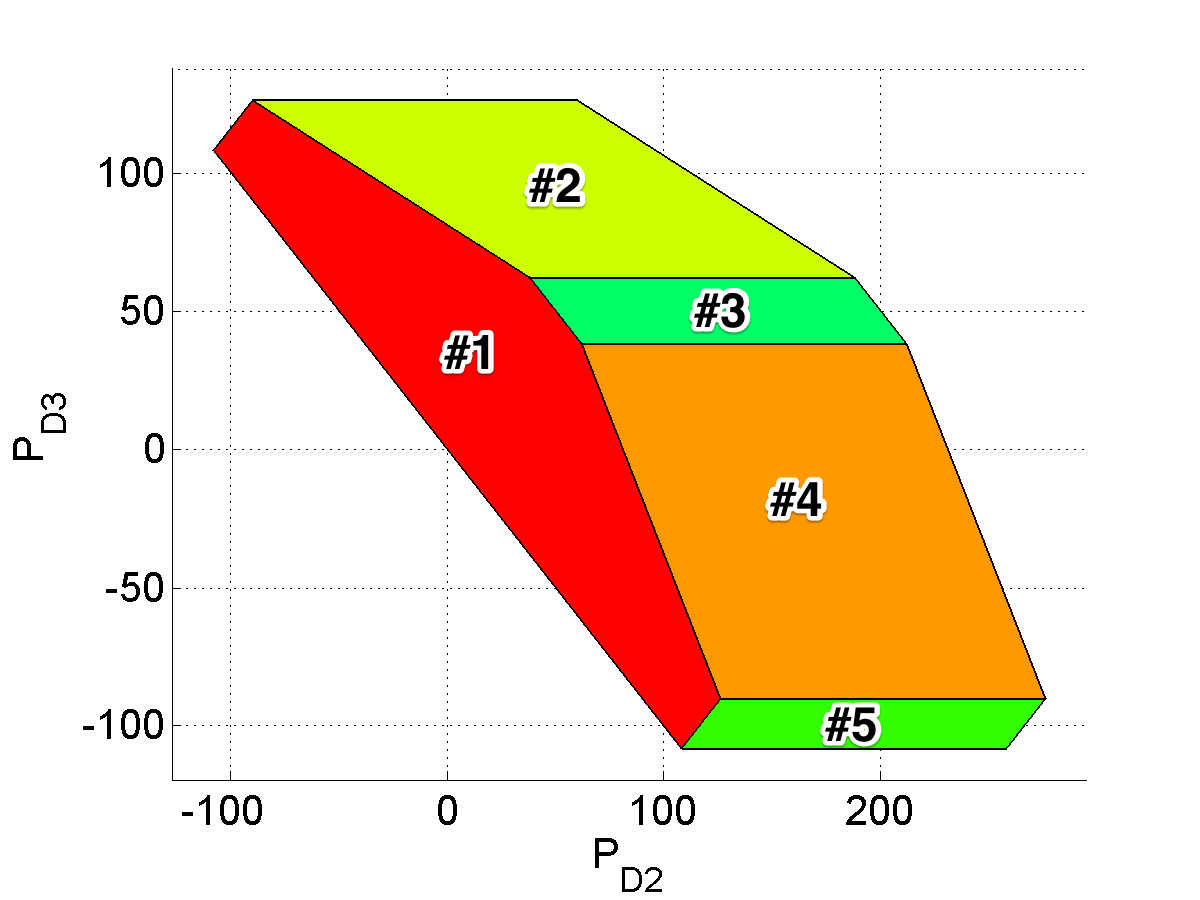} 
  \caption{Line Limits: (54, 54, 72)}
  \label{fig:line_09}
  \end{subfigure}  
  % \subfloat[Data Set]{  \includegraphics[width=0.5\linewidth]{./fig/3BusPlotWithVaryingLineLimits_Edited.png} \label{fig:SPR_3bus_data_plot} }
  % \subfloat[ c= (20, 50, 81) ]{  \includegraphics[width=0.5\linewidth]{./fig/3Bus2GeneSystem_f_[20;50;81]critical_regions.png} \label{fig:SPR_3bus_81} }
  % \subfloat[ c= (20, 50, 79) ]{  \includegraphics[width=0.5\linewidth]{./fig/3Bus2GeneSystem_f_[20;50;79]critical_regions.png} \label{fig:SPR_3bus_79} }
  \caption{SPRs of the 3-bus System (Static SCED with DLRs)}
  \label{fig:SPR_3bus_dlr}
\end{figure}

\begin{figure}[htbp]
  \centering
  \includegraphics[width=0.6\linewidth]{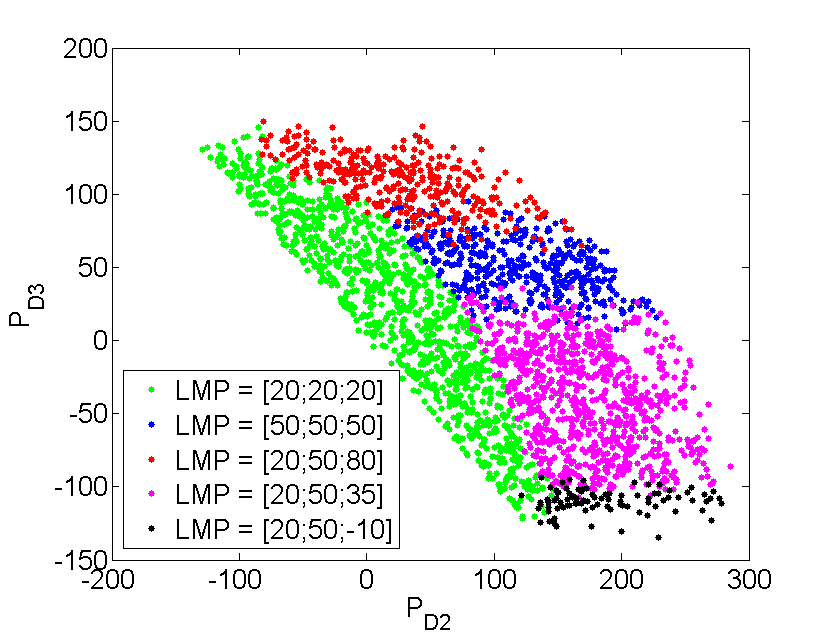}
  \caption{Monte-Carlo Simulation (Static SCED with DLRs)}
  \label{fig:spr_3bus_dlr_mc}
\end{figure}
% \end{remark}

% Assuming the marginal cost $c_3$ of generator\#3 in Fig. \ref{fig:3Bus2GeneSystem} is decreasing (e.g. due to varying price in the natural gas market). When $c_3 = 81 > 80$ (Fig. (\ref{fig:SPR_3bus_81})), the SPRs remain the same as $c_3 = 100$; when $c_3 = 79 < 80$ (Fig. (\ref{fig:SPR_3bus_79})), part of the SPRs change. The threshold $c_3^* = 80$ is calculated from Eqn. (\ref{eqn:conditions_c}). SPR\#1$\sim$\#3 will remain the same as long as $c_3 > c_3^* = 2c_2 - c_1 = 80$.

% subsection dynamic_line_ratings (end)

\subsection{Ramping Constraints} % (fold)
\label{sub:ramping_constraints}
The analysis of SPRs can also be generalized to the dispatch models that include inter-temporal constraints such as ramping:
\begin{equation}
\label{eqn:ramp_const}
  P_G^{k-1} - R^- \Delta t \le P_G^k  \le P_G^{k-1} + R^+ \Delta t
\end{equation}
In Eqn.(\ref{eqn:ramp_const}), $R^+$ and $R^-$ represent the ramp up and down constraints of generators.

Adding ramp constraints to the static SCED problem is equivalent with replacing the generation capacity constraints Eqn. (\ref{eqn:static_sced_generation}) with:
\begin{equation}
\label{eqn:revise_generation_limits}
\max\{P_G^-, P_G^{k-1} - R^- \Delta t\} \le P_G^k  \le \min\{P_G^+, P_G^{k-1} + R^+ \Delta t\}
\end{equation}
When the ramp capacity is not binding, i.e. $P_G^- > P_G^{k-1} - R^- \Delta t$ and $P_G^+ < P_G^{k-1} + R^+ \Delta t$, the SCED problem is the same as the case where no ramp constraints are considered. The SPRs would be exactly the same as in Fig. \ref{fig:SPR_3bus_mpt} and \ref{fig:SPR_3busNO_mc}.
However, active ramp constraints change the actual generation constraints, and therefore change the parameter $b$ in Eqn.(\ref{eqn:details_of_the_A_W_b}). 
This leads to parallel shift of the boundaries of SPRs. The impacts of ramping constraints on SPRs is similar with the case of dynamic line ratings.

The 3-bus system, again, is analyzed via both MPT 3.0 and Monte-Carlo simulation. 
Fig. \ref{fig:ramp_3030} and \ref{fig:ramp_100100} demonstrate the cases where ramp constraints are active. SPRs look similar with parallel changes on the boundaries. 
When analyzing the load and LMP data, we will again see the overlapping SPRs (Fig. \ref{fig:spr_ramp_mc}).

\begin{figure}[htbp]
  \centering
  \begin{subfigure}[t]{0.49\linewidth}
  \centering
  \includegraphics[width=\linewidth]{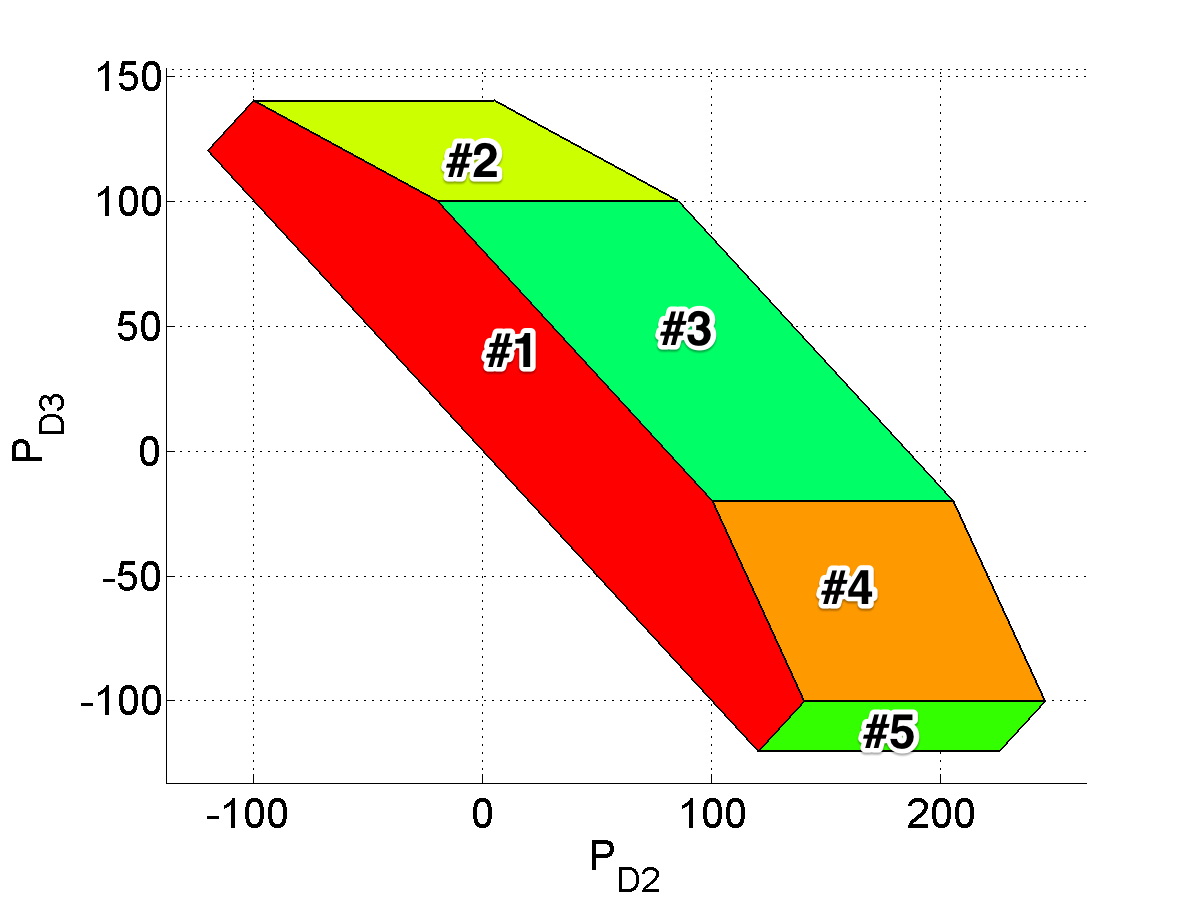}
  \caption{Previous Generation: (30; 30)}
  \label{fig:ramp_3030}
  \end{subfigure}
  \begin{subfigure}[t]{0.49\linewidth}
  \centering
  \includegraphics[width=\linewidth]{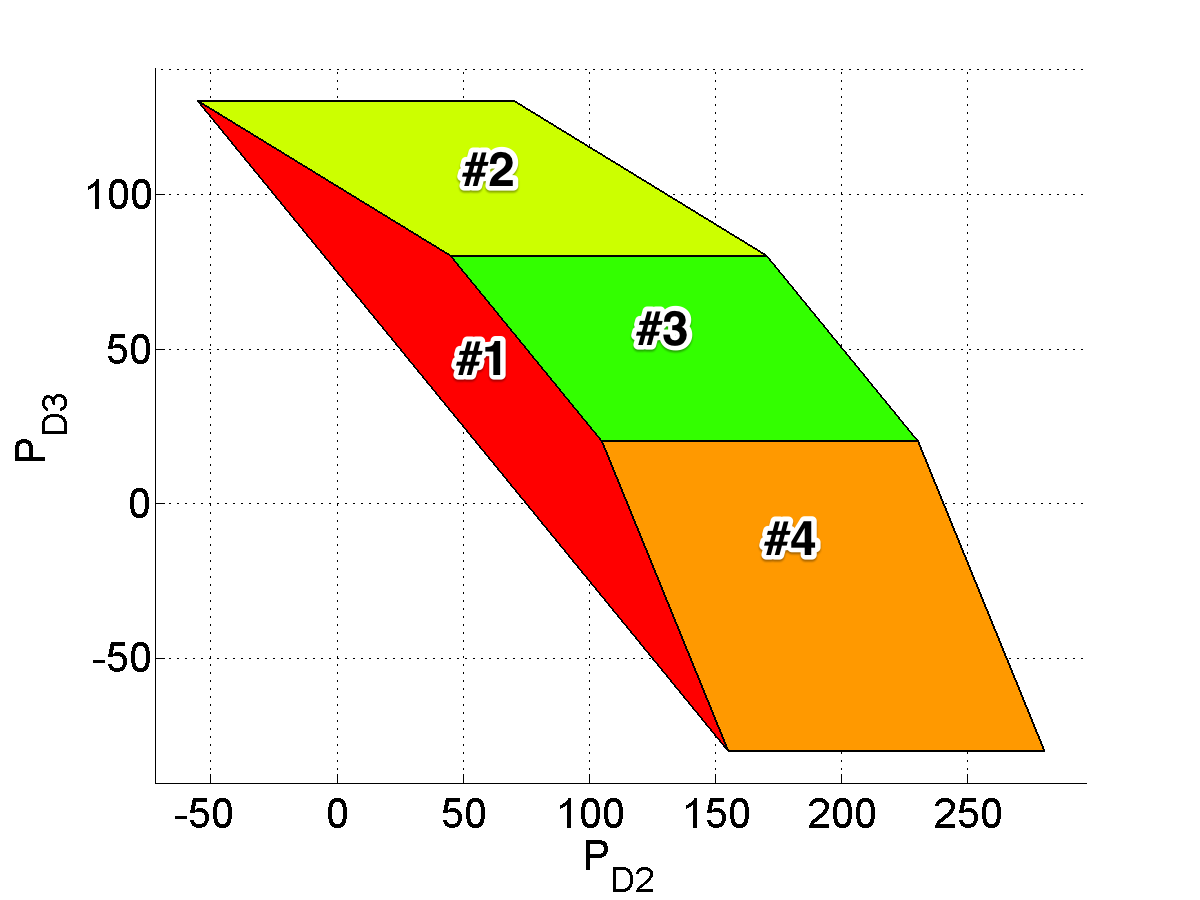}
  \caption{Previous Generation: (100; 100)}
 \label{fig:ramp_100100}
  \end{subfigure}  
  \caption{SPRs of the 3-bus System (SCED with Ramp constraints)}
  \label{fig:spr_ramp_mpt}
\end{figure}

\begin{figure}[htbp]
  \centering
  \includegraphics[width=0.6\linewidth]{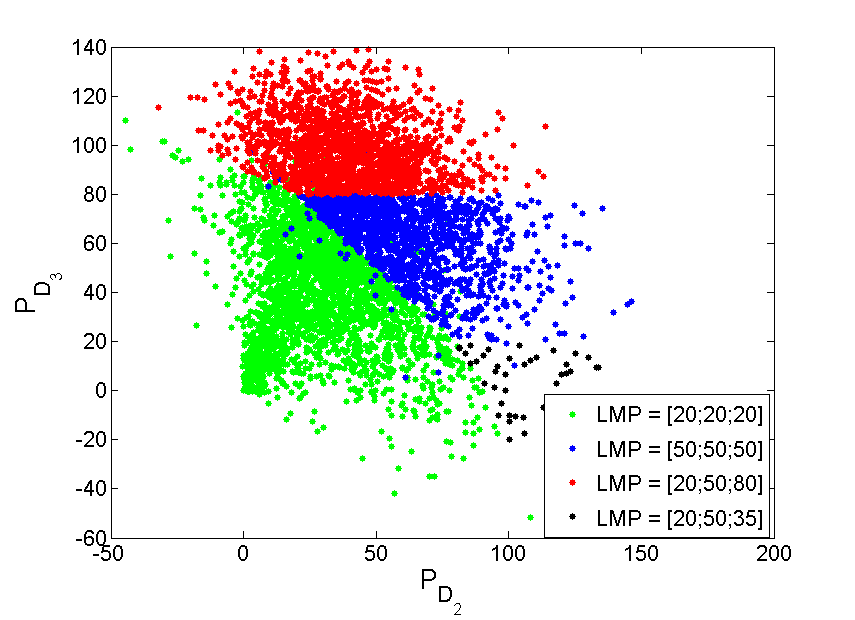}
  \caption{Monte-Carlo Simulation (SCED with Ramp constraints)}
  \label{fig:spr_ramp_mc}
\end{figure}
% subsection ramping_constraints (end)

% \begin{figure}[htbp]
%   \centering
%   \subfloat[Transmission Limits * 0.9]{\includegraphics[width=0.5\linewidth]{./fig/3Bus2GeneSystem_F_[54;54;72]critical_regions.png} \label{fig:line_09}}
%   \subfloat[Transmission Limits * 1.1]{\includegraphics[width=0.5\linewidth]{./fig/3Bus2GeneSystem_F_[66;66;88]critical_regions.png} \label{fig:line11}}
%   \caption{SPRs with Varying Transmission Limits}
%   \label{fig:SPR_transmission_limits}
% \end{figure}

\section{A Data-driven Approach to Identifying SPRs} % (fold)
\label{sec:extended_data_driven_approach}
The SPRs depict the fundamental coupling between loads and LMP vectors. Massive historical data could help market participants estimate SPRs, understand the load-LMP coupling and then forecast LMPs. This section proposes a data-driven method to identify SPRs, which is a significant improvement of the basic method in \cite{Geng2015} by considering varying system parameters and the probabilistic nature of system parameters.

\subsection{The SPR Identification Problem} % (fold)
\label{sub:data_driven_approach_revisited}
\subsubsection{SPR Identification as a Classification Problem} % (fold)
\label{ssub:model_as_a_classification_problem}
A \emph{classifier} is an algorithm to give a \emph{label} $y$ to each \emph{feature} vector $x$. 
The feature vectors sharing the same labels belong to the same \emph{class}. The objective of the classification problem is to find the best classifier which could classify each feature vector accurately. For the parametric classifiers, there is always a training set, i.e. a group of feature vectors whose labels are known. There are two steps in a classification problem: training and classifying. \emph{Training} usually means solving an optimization problem over the training set to find the best parameters of the classifier. And \emph{classifying} is to classify a new feature vector with the trained classifier.

According to Section \ref{sub:SCED_analysis_via_MLP},  the load vectors in an SPR share many common features (e.g. vectors of LMPs). Theorem \ref{thm:diff_LMPs} proved that the LMP vectors are distinct for different SPRs. Therefore, one SPR can be regarded as a \emph{class} and the LMP vector is the \emph{label} of each class. Theorem \ref{thm:convexSPR} proves the existence of the separating hyperplanes. Since each separating hyperplane labels two SPRs with different sides, it turned out that the separating hyperplanes are classifiers and the key of identifying SPRs is to find optimal hyperplanes, which is exactly the objective of Support Vector Machine (SVM).
% subsubsection the_spr_identification_problem_as_a_classification_problem (end)
% subsubsection model_as_a_classification_problem (end)

\subsubsection{SPR Identification with SVM} % (fold)
\label{ssub:spr_identification_with_svm}
% The SPR identification problem is modeled as a classification problem and each SPR is regarded as a class.
Suppose there is a set of labeled load vectors for training and those load vectors belong to only \emph{two} distinct SPRs (labels $y^{(i)} \in \{1,-1\}$).
Then the SPR identification problem with a \emph{binary} SVM classifier (separable case) is stated below:
\begin{subequations}
\begin{align}
  \min_{w,b} & \qquad \frac{1}{2} w^\intercal  w   \label{eqn:separable_SVM_obj} \\ 
  \text{s.t} & \qquad y^{(i)}(w^\intercal  P_D^{(i)}-b) \ge 1, y^{(i)} \in \{-1,1\} \label{eqn:separable_SVM_cons}	
\end{align}
\end{subequations}
The word ``binary'' here specifies only two classes (i.e. SPRs) are being considered. Eqn. (\ref{eqn:separable_SVM_cons}) is feasible only when the two SPRs are not overlapping and there exists at least one hyperplane thoroughly separating them (separable case). For any load vector $P_D$ in the load space, 
$w^\intercal P_D - b = 0$ represents the separating hyperplane where $w$ is the norm vector to the hyperplane. Two lines satisfying $w^\intercal P_D - b = \pm 1$ separate all the training data and formulate an area with no points inside. This empty area is called \emph{margin}. The width of the margin is $2/||w||$, which is the distance between those two lines. The optimal solution refers to the separating hyperplane which maximizes the width of the margin $2/||w||$, therefore the objective of the binary SVM classifier is to minimize the norm of vector $w$.
\begin{figure}[htbp]
  \centering
  \includegraphics[width=0.5\linewidth]{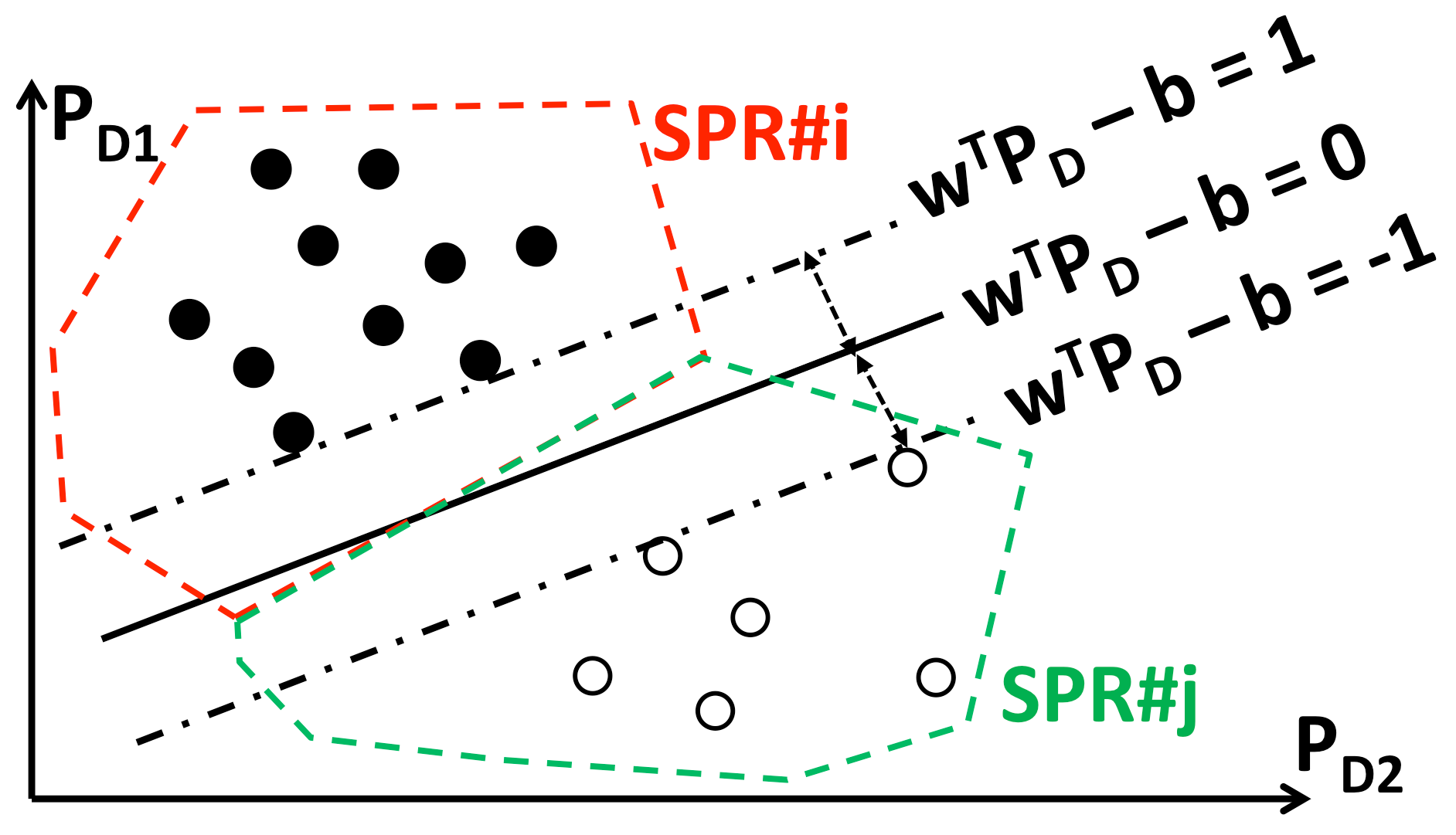}
  \caption{SPR Identification Problem with SVM (Separable Case)}
  \label{fig:SVM_SPR}
\end{figure}

% The features of SPRs are exploited and lead to a ``one-vs-one'' multi-class SVM classifier, which is the foundation of the data-driven approach proposed in \cite{Geng2015}.

% subsubsection spr_identification_with_svm (end)

Due to the existence of multiple SPRs, multi-class classifiers are needed. Since Theorem \ref{thm:convexSPR} guarantees the existence of separating hyperplanes between every pair of SPRs, the ``one-vs-one'' multi-class SVM classifier is incorporated in the data-driven approach to identifying SPRs. Detailed procedures are summarized in Section \ref{sub:extended_data_driven_approach}.

% subsubsection svm_classifier (end)

% each SPR is regarded as a \emph{class}, the system pattern or vector of LMPs is the \emph{label} of the class and the load vector is the \emph{feature vector}. 

% \paragraph{Step\#1 Training} % (fold)
% \label{par:step1_training}
% Suppose there are $N$ different SPRs in the training data set. Each time two SPRs are selected, trained and we get a \emph{binary} SVM classifier. This pairwise training procedure is repeated $C_{N}^2 = N(N-1)/2$ times and we collect $N(N-1)/2$ binary classifiers, namely the $N(N-1)/2$ separating hyperplanes between any two out of $N$ SPRs.
% % paragraph step#1_training (end)
% \paragraph{Step\#2: Classifying} % (fold)
% \label{par:step2_classifying}
% Given a \emph{new} load vector $P_D$ to be classified, each binary classifier provides a classification result (vote), the SPR collects the most votes will be the final classification result. The new feature vector $P_D$ is therefore pinpointed to an SPR, whose label vector (LMP vector) is the forecast of LMPs.
% This algorithm is called ``max-vote-wins'' algorithm in the following sections.
% paragraph step#2_classifying (end)
% subsubsection data_driven_approach (end)
% subsection data_driven_approach_revisited (end)
\subsection{A Data-driven Approach} % (fold)
\label{sub:extended_data_driven_approach}
% \subsubsection{Overlapping Data Sets} % (fold)
% \label{ssub:overlapping_data_sets}
\subsubsection{SPR Identification with Varying System Parameters} % (fold)
\label{ssub:spr_identification_with_varying_system_parameters}
When the system parameters are varying (e.g. dynamic line ratings), two SPRs may overlap with each other. The SPR identification problem is no longer a separable case as in Section \ref{ssub:spr_identification_with_svm}. The SVM classifier needs to incorporate \emph{soft margins} to allow some tolerance of classification error. The slack variable $s_i$ is added to Eqn. (\ref{eqn:separable_SVM_obj}) and penalties of violation $C\sum_i s_i$ are added in the objective function. Large $C$ indicates low extent of tolerance.
 % \footnote{Usually $C = 10$ is a good choice, but a better numerical value of $C$ can be found by training the data set with $C = 0.1,1,10,100,\cdots$.}.
\begin{subequations}
\begin{align}
  \min_{w,b,s} & \qquad { \frac{1}{2} w^\intercal  w + C\sum_i s^{(i)} } \label{eqn:non_sep_SVM_obj} \\ 
  \text{s.t} & \qquad y^{(i)}(w^\intercal  P_D^{(i)} -b) \ge 1-s^{(i)} \label{eqn:non_sep_SVM_cons} \\
 			 & \qquad s^{(i)} \ge 0, y^{(i)} \in \{-1,1\} \nonumber	
\end{align}
\end{subequations}
\begin{figure}[htbp]
  \centering
  \includegraphics[width=0.5\linewidth]{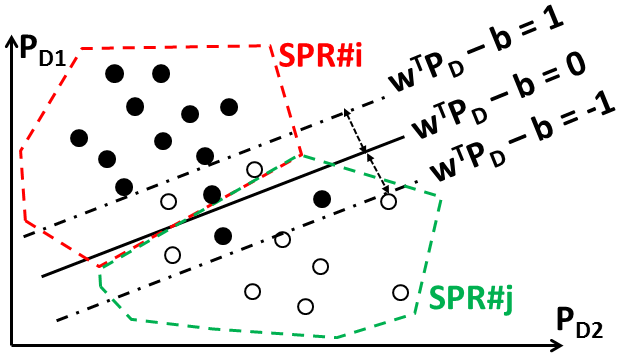}
  \caption{SPR Identification Problem with SVM (Non-Separable Case)}
  \label{fig:SVM_SPR_non}
\end{figure}

% subsubsection spr_identification_with_varying_system_parameters (end)

% The data-driven approach proposed in \cite{Geng2015} is only suitable for separable data sets. In this subsection, varying system parameters in SCED are taken into consideration and the SPRs are therefore non-separable. The data-driven approach is extended and the confidence (posterior probability) of classifying a new load vector is calculated. A probabilistic forecast of LMPs is built upon the posterior probabilities.
%% subsubsection overlapping_data_sets (end)

\subsubsection{Fitting Posterior Probabilities} % (fold)
\label{ssub:fitting_posterior_probabilities}

The posterior probability is the probability that the hypothesis is true given relevant data or observations.
In the classification problem, the posterior probability can be stated as: $\mathbb{P}(\text{class}|\text{input})$.

Estimating the posterior probability is very helpful in practical problems \cite{Platt1999}. 
When identifying SPRs, knowing the posterior probability $\mathbb{P}( y = i | P_D \text{ and } y \in \{1,2,\cdots,n\})$ is not only about knowing the classification result $y = i$ ($P_D$ belongs to SPR\#i), but also understanding the confidence or possible risk. The market participants could accordingly adjust their bidding strategy and reduce possible loss.

Although the posterior probabilities are desired, the standard SVM algorithm provides an uncalibrated value which is not a probability as output \cite{Platt1999}. Modifications are needed to calculate the \emph{binary} posterior probabilities $\mathbb{P}( y = i | P_D \text{ and } y \in \{i,j\})$. 
Common practice is to add a link function to the binary SVM classifier and train the data to fit the link function. Some typical link functions include sigmoid functions \cite{Platt1999} and Gaussian approximations \cite{hastie1998classification}. In this paper, the sigmoid link function is selected due to its general better performance than other choices \cite{Platt1999}. 

In general, there are more than two SPRs. What we really want to know is the \emph{multi-class} posterior probabilities $\mathbb{P}( y=i | P_D \text{ and } y \in \{1,2,\cdots,n\} )$. For short, we will use $\mathbb{P}(y=i|P_D)$ to represent multi-class posterior probabilities. \cite{hastie1998classification} proposed a well-accepted algorithm to calculate multi-class posterior probabilities from pairwise binary posterior probabilities. This algorithm is incorporated in our approach and briefly summarized in Appendix \ref{sec:hastie_tibshirani_s_algorithm}.
% subsubsection fitting_posterior_probabilities (end)

\subsubsection{The Data-driven Approach} % (fold)
\label{ssub:extended_data_driven_approach}
There are three steps in the proposed data-driven approach (Fig. \ref{fig:flowchart}):
\begin{figure}[htbp]
  \centering
  \includegraphics[width=\linewidth]{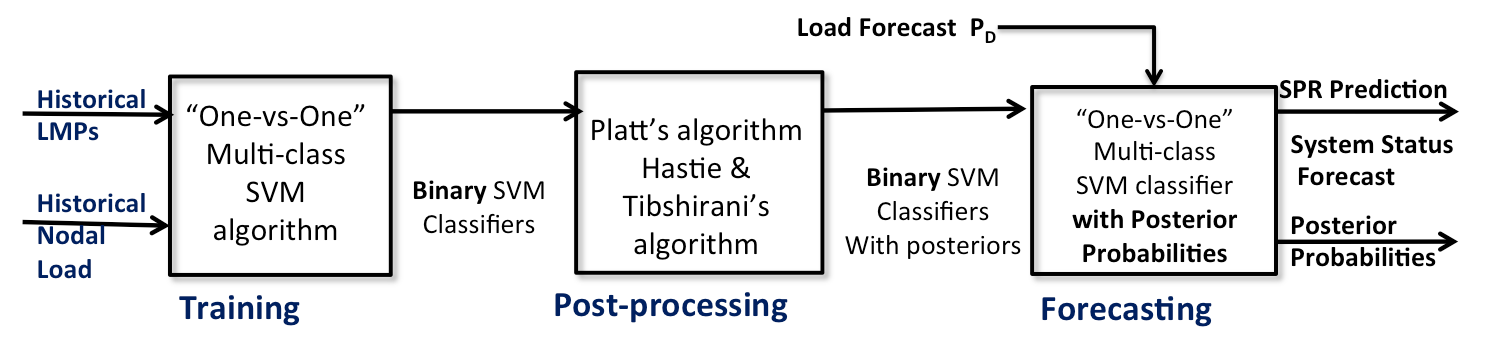}
  \caption{The Data-driven Approach}
  \label{fig:flowchart}
\end{figure}

\paragraph{Training} % (fold)
\label{par:step_1}
Suppose there are $n$ different SPRs in the training data set. Each time two SPRs are selected, trained and we get a \emph{binary} SVM classifier. This pairwise training procedure is repeated $C_{n}^2 = n(n-1)/2$ times and we collect $n(n-1)/2$ binary classifiers, namely the $n(n-1)/2$ separating hyperplanes between any two out of $n$ SPRs.
% paragraph step_1_ (end)

\paragraph{Classifying/Predicting} % (fold)
\label{par:step_3}
Given load forecast $P_D$, we could use the ``\emph{max-vote-wins}'' algorithm to get the classification results:
each binary classifier provides a classification result (vote) for the load forecast $P_D$, the SPR which collects the most votes will be the final classification result. The load forecast $P_D$ is therefore pinpointed to an SPR.
The LMP forecast: 
$\hat{\lambda}(P_D) = \lambda^{(i^*)}$ where $i^*$ is the index of the SPR winning the most votes. This step is independent of the data post-processing procedure.
\paragraph{Data Post-processing} % (fold)
\label{par:step_2}
Calculate posterior probabilities $\mathbb{P}(y = i | P_D)$ for $i=1,2,\cdots, n$ by applying Platt's algorithm and then Hastie and Tibshirani's algorithm\footnote{Details of these two algorithms are summarized in Appendix \ref{sec:platt_s_algorithm_} and \ref{sec:hastie_tibshirani_s_algorithm}.}.
% paragraph step_2_ (end)
% paragraph step_3 (end)
It is worth noting that the proposed approach is generalizable to many other scenarios with overlapping SPRs in the data, possible extensions are discussed in Section \ref{sub:on_posterior_probabilities}.
% subsubsection extended_data_driven_approach (end)

% subsection extended_data_driven_approach (end)

% section data_driven_approach (end)

\section{Case Study} % (fold)
\label{sec:case_study}
% In this section, the data-driven approach is demonstrated in details on a 3-bus system and then tested on the IEEE 118 bus system with more realistic settings.
In this section, we illustrate the proposed data-driven approaches on two systems.
\subsection{Performance Metrics} % (fold)
\label{sub:performance_measurements}
We first introduce the performance metrics.
\subsubsection{5-fold Cross Validation} % (fold)
\label{ssub:cross_validation}
To evaluate the performance of the model to an independent data set and avoid overfitting, the \emph{$k$-fold cross validation} technique is being used.
In $k$-fold cross-validation, the overall data set is randomly and evenly partitioned into $k$ subsets. Every time a subset is chosen as validation data set, and the remaining $k-1$ subsets are used for training. This cross-validation process is repeated $k$ times ($k$ folds), and each subset serves as the validation data set once.
The 5-fold cross validation is being used in this paper.

% subsubsection cross_validation (end)

\subsubsection{Classification Accuracy} % (fold)
\label{ssub:classification_accuracy_eval}
\emph{Classification accuracy} is the most common criteria to evaluate the performance of classifiers. The classification accuracy $\alpha$ is the ratio of the correctly classified points in the validation data set.
When incorporating $5$-fold cross validation, the classification accuracy of each fold ($\alpha_1, \alpha_2, \cdots, \alpha_5$) is calculated first, then the overall performance of the method is evaluated by the average classification accuracy 
$\bar{\alpha} = (\sum_{i=1}^5 \alpha_i)/5$.
% If there are $N$ points in the data set and the \emph{k-fold cross validation} method is being used, then the \emph{classification accuracy} of the $i$th fold $\alpha_i^C$ is the portion of correctly classified points in the testing data set, and the overall classification accuracy $\alpha^C$ is the average of $\alpha_i^C$.
% \begin{equation}
% \label{eqn:classification_accuracy}
% \alpha_i^C = \frac{N_{correct}}{N/k}, \alpha^C = \frac{\sum_{i=1}^k \alpha_i}{k}
% \end{equation}

\subsubsection{LMP Forecast Accuracy} % (fold)
\label{ssub:lmp_forecast_accuracy}
The proposed approach forecasts the LMP at every bus. The performance of LMP forecast at bus $i$ is evaluated by the nodal LMP forecast accuracy $\beta_i$, which is the average forecast accuracy of all the validation data points ($j=1,2,\cdots,n_v$)
\begin{equation}
\label{eqn: LMP_accu_nodal}
\beta_i =  \frac{1}{n_v} \sum_{j=1}^{n_v}{ \frac{ | \hat{\lambda_i}[j] - \lambda_i[j]| }{ \lambda_i[j] } }
% \alpha^{\lambda} &=&\alpha_i^\lambda 
\end{equation}
The overall LMP forecast accuracy $\beta$ evaluates the performance of LMP forecast for the whole system. It is the average of all the nodal LMP forecast accuracy $\beta_i$ ($i=1,2,\cdots,n_b$):
\begin{equation}
\label{eqn: LMP_accu_overall}
\beta =  \frac{1}{n_b} \sum_{i=1}^{n_b} \beta_i
% \alpha^{\lambda} &=&\alpha_i^\lambda 
\end{equation}

% subsubsection lmp_forecast_accuracy (end)

% subsection performance_measurements (end)

\subsection{Static SCED with Static Line Ratings} % (fold)
\label{sub:case_studies_of_static_sced_with_static_line_ratings}
This section explores the simplest case: static SCED with SLRs. Since \cite{Geng2015} discusses the 3-bus system as well as the IEEE 24-bus system, we only examine the data-driven approach on an 118-bus system. The same dataset generated in this section is used in Section \ref{sub:how_the_nodal_load_levels_help_us_} as well.

\paragraph{System Configuration} % (fold)
\label{par:system_configuration}
Most of the system settings follow the \emph{IEEE 118-bus, 54-unit, 24-hour system} in \cite{Technology} but with the following changes: (1) the lower bounds of generations are set to zero, but the upper bounds of generators remain the same as in \cite{Technology}; (2) generation costs are linear. Details of the parameters are summarized in \cite{Geng2015c}.
% (3) the transmission capacity of the line from bus 87 to 86 is set to be 600 instead of 500  (static line ratings) are the same as \cite{Technology} except the 

\paragraph{Load} % (fold)
\label{par:Load}
\cite{Technology} also provides an hourly system load profile and a bus load distribution profile. With linear interpolation, the hourly system load profile is modified to be 5-min based. To account for the variability of loads, we assume the load at each bus follows normal distribution $\mathcal{N}(\mu,\sigma)$. The expectation $\mu$ of each nodal load is calculated from the system load profile and bus load distribution profile, the standard deviation $\sigma$ is set to be $10\%$ of the expectation. 1440 (5 days, 5-min based) load vectors are generated, then Matpower \cite{zimmerman2011matpower} solves these 1440 SCED problems and records 1440 LMP vectors. These 1440 load vectors and LMP vectors are the training and validation data.
% These settings of load focus on a small but practically meaningful portion of the load space, and the calculation burden is significantly reduced.

\paragraph{Simulation Results} % (fold)
\label{par:simulation_results}
Results are summarized in Table \ref{tab:classification_accuracy_118bus}.
The classification accuracy is around 67\% but the LMP forecast is satisfying. When the classification result of a load vector is correct, the LMP forecast is correct for every bus, i.e. $\beta = 100\%$. 
It is worth noting that even if the classification fails, the overall LMP forecast still has accuracy about 90\%. This is because the classification errors happen between one SPR and it neighbors. LMPs of adjacent SPRs are similar due to the fact that only one active constraint is different (Lemma \ref{lem:adjacent_SPRs}). Therefore, the LMP forecast result is much more accurate than classification.

\begin{table}[htbp]
  \caption{Results of the 118-bus System (Static SCED with SLR)}
  \label{tab:classification_accuracy_118bus}
  \centering

  \begin{tabular}{l|cc}
  \hline

  \hline
  \textbf{Fold} & \textbf{Classification} & \textbf{LMP Forecast} \\
  \hline
    1 & 64.24\% & 96.82\% \\
    2 & 67.36\% & 96.71\% \\
    3 & 64.93\% & 96.95\% \\
    4 & 71.18\% & 97.34\% \\
    5 & 65.63\% & 96.84\% \\
  \hline
  \textbf{avg} & 66.67\% & 96.93\% \\
  \hline

  \hline
  \end{tabular}
\end{table}

% subsection case_studies_of_static_sced_with_static_line_ratings (end)

\subsection{Static SCED with Dynamic Line Ratings} % (fold)
\label{sub:case_study_dlr}
\subsubsection{3-bus System} % (fold)
\label{ssub:3_bus_system_dlr}
We start with an illustrative 3-bus system example. This succinct example provides key insights and visualization of the proposed method.

\paragraph{Data} % (fold)
\label{par:3bus_data}
The parameters of the 3-bus system are presented in Fig. \ref{fig:3Bus2GeneSystem}.
The data set is generated using Matpower with the following assumptions: 
(1) the load vector is evenly distributed in the load space;
and (2) the transmission limits $F$ is time-varying: for simplicity, we utilize the following model to calculate the real-time transmission limits $F$:
\begin{equation}
\label{eqn:model_dlr}
F = (1+\xi) F_0
\end{equation}
$F_0$ is the ``standard'' transmission limits and $F_0 = [60;60;80]$. It is the same as the case of static line ratings.  
$\xi \sim N(0,0.1)$ represents the major factor (e.g. ambient temperature or wind speed) that impacts the transmission capacities. All the data generated is visualized in Fig. \ref{fig:spr_3bus_dlr_mc}.

% The diagonal standard deviation matrix $\Sigma$ indicates the varying line limits are independent of each other. 

% Each color represents an SPR. As shown in Fig. (\ref{fig:SPR_3bus_data_plot}), some SPRs (e.g. blue and red) are overlapping.
% subsubsection data_generation (end)
% \begin{figure}[htbp]
%   \centering
%   \includegraphics[width=0.7\linewidth]{./fig/3BusPlotWithVaryingLineLimits_Edited.png}  
%   \caption{Generated Dataset}
%   \label{fig:SPR_3bus_data_plot}
% \end{figure}

% \begin{figure}[htbp]
%   \centering
%   % \subfloat[Data Set]{  \includegraphics[width=0.45\linewidth]{./fig/3BusPlotWithVaryingLineLimits_Edited.png} \label{fig:SPR_3bus_data_plot} }
%   \subfloat[Posterior Probabilities of A Given Load Vector]{  \includegraphics[width=0.9\linewidth]{./fig/visualize_posteriors.png} \label{fig:visualize_posteriors}  } \\
%   \subfloat[Posterior Probability Surfaces]{  \includegraphics[width=\linewidth]{./fig/post_prop_3bus_fine.png} \label{fig:post_prop_3bus_fine} }
%   % \subfloat[ c= (20, 50, 81) ]{  \includegraphics[width=0.5\linewidth]{./fig/3Bus2GeneSystem_f_[20;50;81]critical_regions.png} \label{fig:SPR_3bus_81} }
%   % \subfloat[ c= (20, 50, 79) ]{  \includegraphics[width=0.5\linewidth]{./fig/3Bus2GeneSystem_f_[20;50;79]critical_regions.png} \label{fig:SPR_3bus_79} }
%   \caption{Results of the 3 Bus System}
%   \label{fig:results_3bus}
% \end{figure}

\paragraph{Simulation Results} % (fold)
\label{par:classification}
Table \ref{tab:classification_accuracy_3bus} summarizes the classification and LMP forecast accuracies. The accuracies are around 95\% because of the overlapping SPRs.

\begin{table}[htbp]
  \caption{Results of 3-bus System (5-fold Validation)}
  \label{tab:classification_accuracy_3bus}
  \centering

  \begin{tabular}{l|cccc}
  \hline

  \hline
  \textbf{Fold} & \textbf{Classification} & \textbf{LMP Forecast}  \\
  \hline
  1 & 93.967\% & 96.218\% \\
  2 & 93.236\% & 96.054\% \\
  3 & 94.150\% & 95.767\%\\
  4 & 95.612\% & 96.700\%\\
  5 & 94.150\% & 96.405\%\\
  \hline
  \textbf{avg} & 94.23\% & 96.23\%\\
  \hline

  \hline
  \end{tabular}
\end{table}

% subsubsection classification (end)

\paragraph{Posterior Probabilities} % (fold)
\label{par:posterior_probability_estimation}
The posterior probabilities are visualized. 
The posterior probabilities of an SPR compose a surface (Fig. \ref{fig:3bus_dlr_post_lmp_205030} and Fig. \ref{fig:3bus_dlr_post_lmp_205030}). When putting all the 5 surfaces of 5 SPRs together (shown in Fig. \ref{fig:3bus_dlr_post_all}), the five surfaces intersect with each other and formulate some ``mountains'' and ``valleys''. The ``mountains'' correspond to the inner parts of SPRs, where the overlapping of SPRs is almost impossible to happen. And the  ``valleys'' always locate at the boundaries among SPRs.

\begin{figure}[htbp]
  \centering
  \begin{subfigure}[t]{0.49\linewidth}
  \centering
  \includegraphics[width=\linewidth]{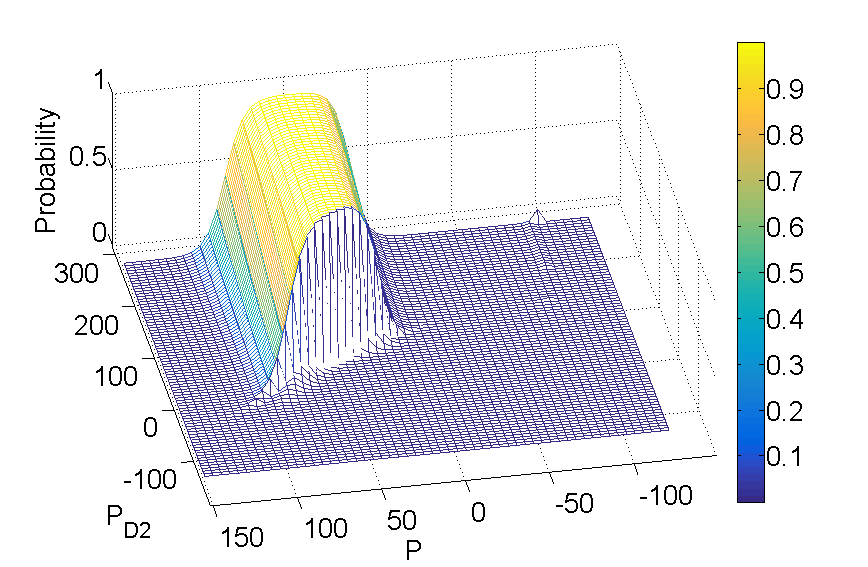}
  \caption{SPR\#3: LMP = (50,50,50)}
  \label{fig:3bus_dlr_post_lmp_505050}
  \end{subfigure}
  \begin{subfigure}[t]{0.49\linewidth}
  \centering
  \includegraphics[width=\linewidth]{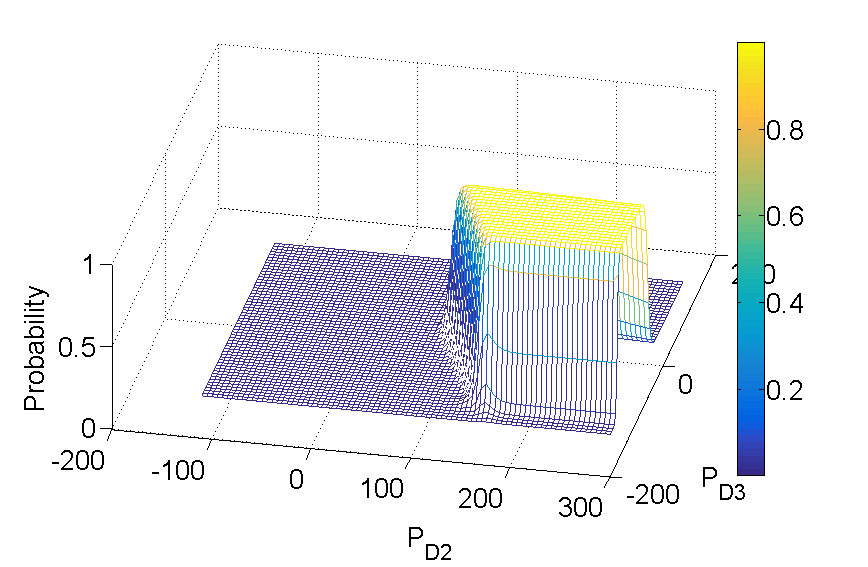}
  \caption{SPR\#4: LMP = (20,50,35)}
  \label{fig:3bus_dlr_post_lmp_205030}
  \end{subfigure} 
  \caption{Posterior Probabilities of Two SPRs}  
  \label{fig:visualize_posteriors} 
\end{figure}

\begin{figure}[htbp]
  \centering
  \includegraphics[width=0.8\linewidth]{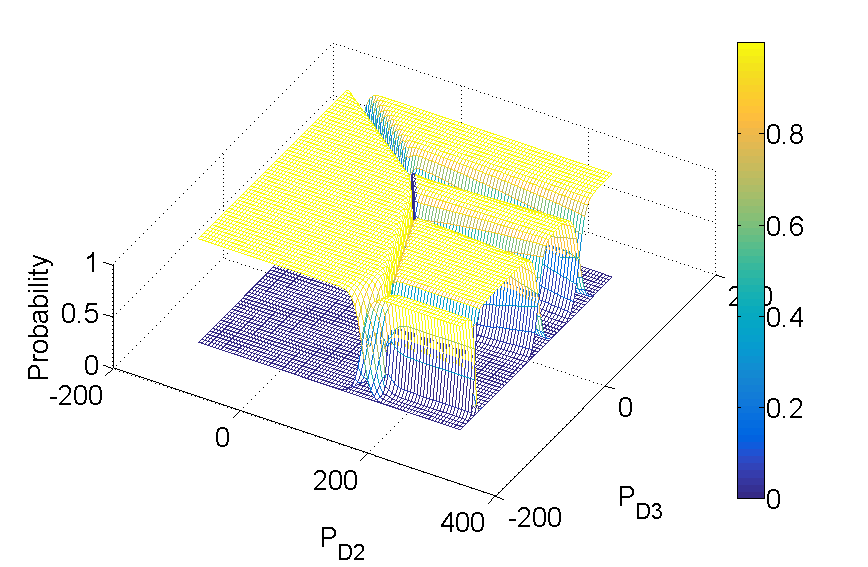} 
  \caption{Posterior Probability Surfaces}
  \label{fig:3bus_dlr_post_all} 
\end{figure}

% If we take supremum over all the eight planes, we could obtain the classification results for any $P_D\in\mathbb{D}$. In Fig. (\ref{fig:top_post_prob_3bus_fine}), darker area corresponds to the boundaries between two neighbor SPRs. 

% Two examples are provided in Fig. (\ref{fig:visualize_posteriors}). When $P_{D_2} = P_{D_3}=100$MW, the load vector is in the middle of SPR\#3, $\mathbb{P}(y=3|P_D = [100;100])\approx 1$ and other posterior probabilities are close to zero. However, when $P_{D_2} = 175$MW and $P_{D_3} = 74$MW, the load vector is close to the boundaries of SPR\#2,\#3,\#4 and \#6. In this case, $\mathbb{P}(y=2|P_D) = 0.439$, $\mathbb{P}(y=2|P_D) = 0.276$, $\mathbb{P}(y=2|P_D) = 0.156$, $\mathbb{P}(y=2|P_D) = 0.104$ and other probabilities are close to zero. With the estimated posterior probabilities, market participants could estimate the risks associated with LMPs given different combinations of load levels without running numerous simulations. 

% subsubsection posterior_probability_estimation (end)
% subsection an_illustrative_example (end)
% subsubsection 3_bus_system (end)

\subsubsection{118-bus System} % (fold)
\label{ssub:118_bus_system_dlr}
A more comprehensive case study is conducted on the 118-bus system to evaluate the performance and computational burden of the data-driven approach on a complex system with realistic settings.
% (3) to explore the data requirement on the proposed approach.

\paragraph{System Configuration} % (fold)
\label{par:system_configuration}
The only difference from the system configuration in Section \ref{sub:case_studies_of_static_sced_with_static_line_ratings} is about transmission limits. To consider DLR, we use the same model as Eqn. (\ref{eqn:model_dlr}). $F_0$ is the same as the transmission limits in \cite{Technology} and $\xi \sim N(0,0.1)$.
% paragraph system_configuration (end)

% all the parameters except nodal loads remain the same as the . We assume all the 118 buses follow the same \emph{normalized} nodal load daily cycles (Fig. \ref{fig:nodal_load_cycle_3bus}), the actual load curve of the $i$th bus is the load $P_{D_i}$ times the normalized load cycle. Nominal values of $P_{D_i}$ are the real power of each bus in the file ``case118.m'' of MATPOWER. We assume the load level of each bus satisfies independent Gaussian distribution, the expectations are provided by the nodal load curve, and the standard deviation equals to $5\%$ of the expectation.
% subsubsection system_configuration (end)

% \begin{figure}[htbp]
% \centering
% \subfloat[System Load (5min-based)]{\includegraphics[width=0.5\linewidth]{./fig/system_load.png} \label{fig:system_load}} 
% \subfloat[Load Profile at Bus\#1]{\includegraphics[width=0.5\linewidth]{./fig/load_at_bus1.png} \label{fig:load_at_bus_1} } 
% \caption{Assumptions on the Load (3 Bus Example)}
% \label{fig:load_path}
% \end{figure}

% subsubsection system_load_path (end)

\paragraph{Performance} % (fold)
\label{par:classification_accuracy}
The algorithm is implemented using the Statistics and Machine Learning Toolbox of Matlab.
Table \ref{tab:avg_cal_time_per_step} summarizes the computation time of each step in the data-driven approach on a PC with Intel i7-2600 8-core CPU@3.40GHz and 16GB RAM memory. There are 181 SPRs found in 1152 points for training, $C_{181}^2 = 16290$ SVM classifiers are trained in 58.72 seconds. On average, one SVM classifier is trained within $0.004$ seconds. This is because most of the SPRs are completely separable, these cases will be solved in an extremely short time. Those adjacent SPRs are overlapping and are the major source of the computational burden.

\begin{table}[htbp]
  \caption{Average Computation Time (in seconds)}
  \label{tab:avg_cal_time_per_step}
  \centering

  \begin{tabular}{l|cccc}
  \hline

  \hline
  \textbf{Steps} & \textbf{Computation Time (s)}\\
  \hline
  \textbf{(a) training}  & 58.73 \\
  \textbf{(b) predicting} (288 points)  & 26.8504  \\
  \textbf{(c) data post-processing} & 701.22 \\
  \hline

  \hline
  \end{tabular}
\end{table}

\begin{table}[htbp]
  \caption{Results of the 118-bus System (Dynamic Line Rating)}
  \label{tab:classification_accuracy_118bus}
  \centering

  \begin{tabular}{l|cc}
  \hline

  \hline
  \textbf{Fold} & \textbf{Classification} & \textbf{LMP Forecast} \\
  \hline
    1 & 61.11\% & 95.11\% \\
    2 & 59.38\% & 94.53\% \\
    3 & 60.76\% & 95.24\% \\
    4 & 51.39\% & 93.34\% \\
    5 & 55.90\% & 94.22\% \\
  \hline
  \textbf{avg} & 57.71\% & 94.49\% \\
  \hline

  \hline
  \end{tabular}
\end{table}

% \begin{table}[htbp]
%   \caption{Overall LMP Forecast Accuracy (118 Bus System)}
%   \label{tab:LMP_accuracy_118bus}
%   \centering

%   \begin{tabular}{l|cc}
%   \hline

%   \hline
%   \textbf{Fold} & \textbf{Max-vote-wins} & \textbf{Max-posterior-wins} \\
%   \hline
%     1 & \% & \% \\
%   2 & \% & \% \\
%   3 & \% & \% \\
%   4 & \% & \% \\
%   5 & \% & \% \\
%   \hline
%   \textbf{avg} & \% & \% \\
%   \hline

%   \hline
%   \end{tabular}
% \end{table}

% subsubsection classification_accuracy (end)

% subsubsection discussions (end)
% subsubsection price_spikes (end)

\subsection{Case Studies with Ramp Constraints} % (fold)
\label{sub:ramp_constraints}
\paragraph{Settings} % (fold)
\label{par:settings}
The parameters of the 118-bus system are the same as in Section \ref{sub:case_studies_of_static_sced_with_static_line_ratings}. And the ramp capacities of generators follows the simplified assumption below: each generator could ramp up (down) to its generation limits in 15 min. For example, a generator with $G^+ = 200$MW and $G^- = 125$MW, its ramp capacity is: $R^+ = R^- = (200-125)/15 = 5$MW/min. This setting is called $R_0$ in Table \ref{tab:118_bus_diff_ramp}.
Due to the temporal coupling of SCED with ramp constraints, a daily load curve is necessary.
The settings of loads are the same as in Section \ref{sub:case_studies_of_static_sced_with_static_line_ratings}. 1440 SCED problems are solved consecutively with Matpower, and 1440 load vectors and LMP vectors are recorded.
% paragraph settings (end)
\paragraph{Simulation Results} % (fold)
\label{par:simulation_results}
The classification and LMP forecast accuracies are summarized in Table. \ref{tab:118_bus_diff_ramp}.
With the ramp settings above, the classification and LMP forecast are satisfying. 
However, different ramp settings would change the results dramatically. As shown in Table. \ref{tab:118_bus_diff_ramp}: when generators ramp up/down 2 times faster ($R/R_0 = 2$), the ramp constraints would rarely be active, then it is the same as static SCED; when generators ramp up/down 2 times slower ($R/R_0 = 0.5$), the actual generation upper/lower bounds are determined by previous dispatch results and ramp constraints. Generation limits become time-varying and the SPRs are overlapping. This explains the unsatisfying results when the system is lack of ramp capacities. Furthermore, varying SPRs could also explain the price spikes during ramping up hours in the morning and ramping down hours in the early evening.

\begin{table}[htbp]
  \caption{Results on 118-bus System with Different Ramp Settings}
  \label{tab:118_bus_diff_ramp}
  \centering

  \begin{tabular}{l|ccc}
  \hline

  \hline
  $R/R_0$ & 0.5 & 1 & 2 \\
  \hline
  \textbf{LMP Forecast} & 44.57\% & 85.10\% & 96.33\% \\
  \hline

  \hline
  \end{tabular}
\end{table}
% paragraph simulation_results (end)
% subsubsection 118_bus_system (end)
% subsection ramp_constraints (end)

\section{The Impact of Nodal Load Information} % (fold)
\label{sec:impact_of_nodal_load_information}
We would like to point out that one possible contribution of this paper is to consider the LMP changes due to nodal
load variations. This section dedicates to a detailed discussion about the impact of nodal load information on the understanding of LMP changes. We first demonstrate the benefits of having nodal load information in Section \ref{sub:how_the_nodal_load_levels_help_us_}; then Section \ref{sub:partial_load_information} illustrates the effects of incomplete load information and the attempts to solve the issue.

To concentrate on the effects that incomplete load information brings, we make the following assumptions: (1) transmission limits are constant, no dynamic line ratings are being considered; (2) ramp constraints are not taken into account.

\subsection{On Nodal Load Levels} % (fold)
\label{sub:how_the_nodal_load_levels_help_us_}
Previous literature such as \cite{Li2007} studied the impact of system load levels on the LMPs. An important concept ``critical load level'' is defined as the system load level where the step changes of LMPs happen. Many LMP forecast methods were proposed based on identifying CLLs. But the definition of CLL assumes that the nodal load levels of all the buses change \emph{proportionally}. This assumption constrains the load vectors in the load space to be on a straight line, and the CLLs are indeed the intersection points of the straight line with the boundaries of SPRs.

We would like to point out that one possible contribution of this paper is to consider the LMP changes due to \emph{nodal-level} load variations.
Contrary to CLL-based methods, which solve a one-dimensional problem, the proposed SVM-based method could explore all the dimensions of the load space and is indeed a generalization of the CLL-based method.

Consider the SPR identification problem with only one feature vector: the total demand of the system. 
Fig. \ref{fig:cll_ID} illustrates the problem formulation. Since only the total demand $P_D = P_{D_1} + P_{D_2}$ is available, the load vectors in the original SPRs are projected to the axis of total demand. Because this is a one-dimension problem, the SVM classifier degenerates to the case that there is only one decision variable $b$, the direction of the separating hyperplane $w$ is represented by the positivity of $b$. The objective becomes finding the optimal value $b$ which has the least overlapping points of different LMPs.
\begin{subequations}
\begin{align}
  \min_{b,s} & \qquad { \sum_i s^{(i)} } \label{eqn:SVM2CLL_obj} \\ 
  \text{s.t} & \qquad y^{(i)}( P_D^{(i)} -b) \ge 1-s^{(i)} \label{eqn:SVM2CLL_cons} \\
  & \qquad  s^{(i)} \ge 0, y^{(i)} \in \{-1,1\} \nonumber 
\end{align}
\end{subequations}

\begin{figure}[htbp]
  \centering
  \includegraphics[width=0.5\linewidth]{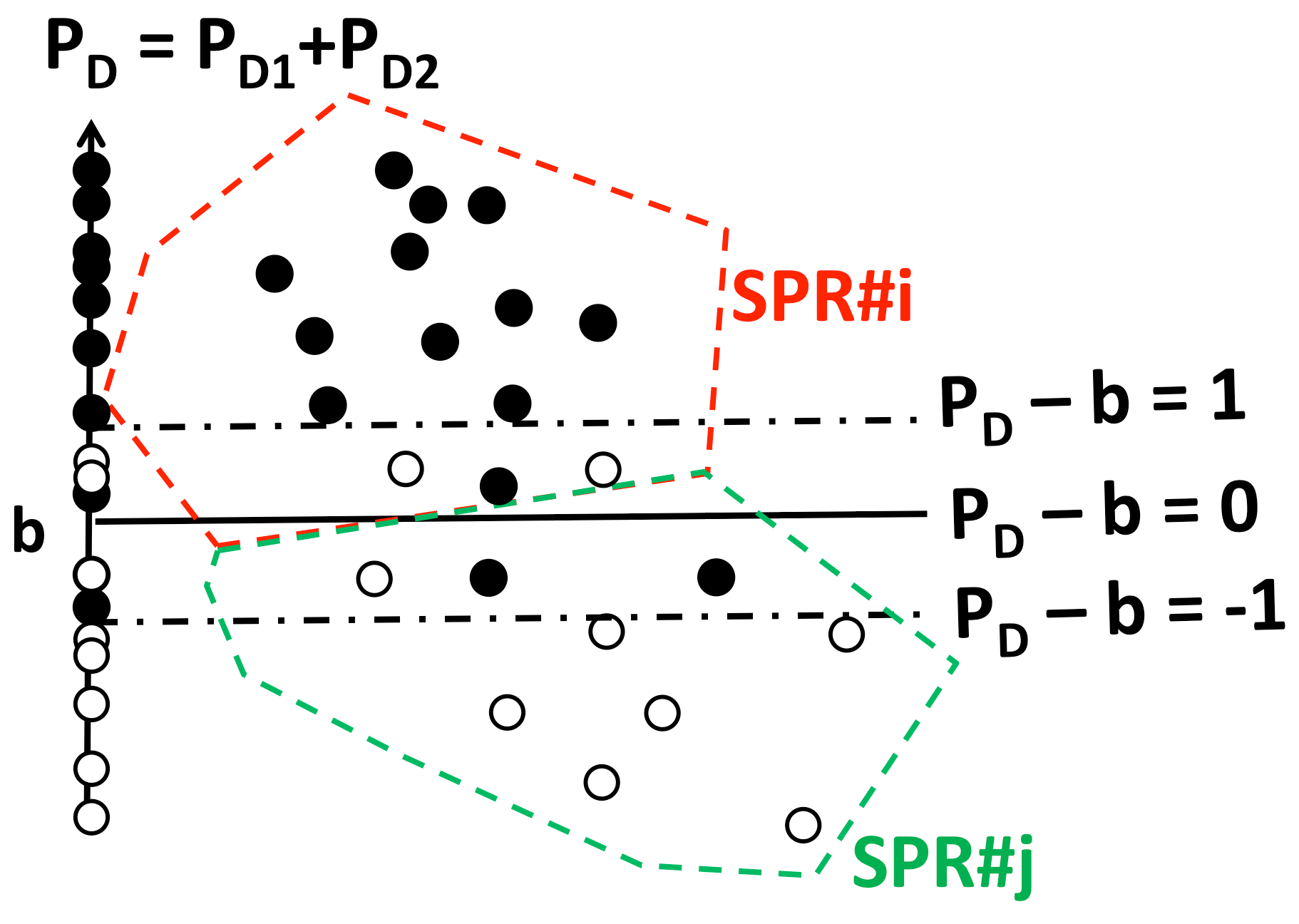}
  \caption{Identifying Critical Load Levels}
  \label{fig:cll_ID}
\end{figure}

We compare this CLL-based method and SVM-based method on the 3-bus system and 118-bus system. Results are demonstrated in Table. \ref{tab:LMP_forecast_lse_3bus_alllmp}, \ref{tab:cll_svm_118bus} and Fig. \ref{fig:svm_cll_nodal_compare}. The performance of both methods are close for the nodal LMP forecast of many buses, 
but the CLL-based method failed to provide correct forecast of some specific buses (e.g. bus 23 in Fig \ref{fig:svm_cll_nodal_compare}), while the SVM-based method provides much better results. The SVM-based method is also better on forecasting high prices.

\begin{table}[htbp]
  \caption{Comparison of CLL and SVM (118-bus system)}
  \label{tab:cll_svm_118bus}
  \centering

  \begin{tabular}{l|cc}
  \hline

  \hline
  \textbf{LMP Forecast} & CLL & SVM  \\ 
  \hline
  \textbf{Overall} & 94.82\% & 95.95\%\\
  \textbf{Price $> 45$ \$/MWh} & 88.86\%  & 96.32\%\\
  \textbf{Worst Forecast (Bus No.)} & 73.92\% (23) & 88.17\% (23)\\
  \hline

  \hline
  \end{tabular}
\end{table}

\begin{figure}[htbp]
  \centering
  \includegraphics[width=0.6\linewidth]{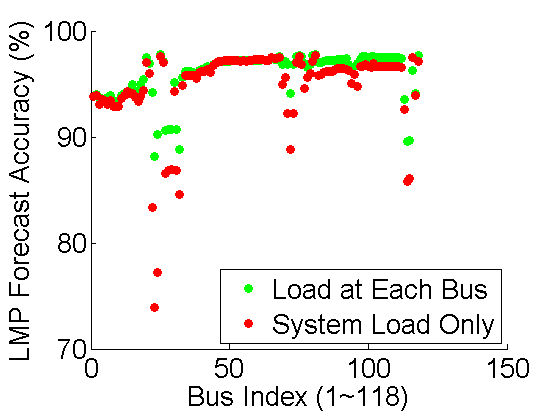}
  \caption{Nodal LMP Forecast Accuracy}
  \label{fig:svm_cll_nodal_compare}
\end{figure}

% Because of this, the methods based on identifying CLLs are not capable to explore the impacts of congestions on the changes of LMPs. The SVM-based method, not restricting on an one-dimension line, incorporates nodal load levels and explores around that straight line in a much higher dimension space. By doing this, our method could forecast the nodal LMP and the LMP differences more accurately than using system load only.

% subsection how_the_nodal_load_levels_help_us_ (end)

\subsection{Incomplete Load Information} % (fold)
\label{sub:partial_load_information}
In practice, LSEs or other market participants may not have the complete information about load levels at all buses. We investigate the performance of the algorithm by assuming LSEs have access to only: (1) the total system-level load; and (2) the nodal load levels in its own area.

To better illustrate the problem formulation. We add a load $P_{D_1}$ at bus 1 to the 3-bus system in Fig.\ref{fig:3Bus2GeneSystem}\footnote{If there are still two loads in the system, knowing system-level load $P_{D_2}+P_{D_3}$ and $P_{D_2}$ is equivalent with knowing $P_{D_2}$ and $P_{D_3}$.}. Modified system is shown in Fig.\ref{fig:3Bus2Gen3Load}. 
Assume there are three LSEs in the system. LSE \#$i$ at bus $i$ has access to the following information: (1) load at bus $i$: $P_{D_i}$; and (2) system-level load: $\sum_{i=1}^3 P_{D_i}$.
\begin{figure}[htbp]
  \centering
  \includegraphics[width=0.7\linewidth]{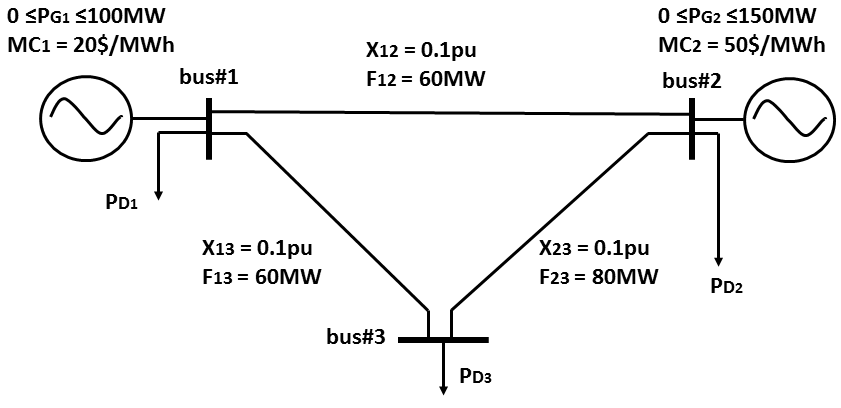}
  \caption{3-bus system with 3 loads}
  \label{fig:3Bus2Gen3Load}
\end{figure}

% \begin{figure}[htbp]
%   \centering
%   % \includegraphics[]{}
%   \subfloat[LMP Vector]{\includegraphics[width=0.5\linewidth]{./fig/LSE1.png} \label{fig:LSE1_3bus}}
%   \subfloat[LMP at Bus 1]{\includegraphics[width=0.5\linewidth]{./fig/LSE1_own_LMP.png} \label{fig:LSE1_own_LMP_3bus}}
%   \caption{LSE 1}
%   \label{fig:LMP_LSE1}
% \end{figure}
With incomplete load information, the SPR identification problem becomes more difficult. 
For example, LSE 2 observes two SPRs which almost completely overlap with each other (blue and red in Fig. (\ref{fig:LSE2_3bus})). 
Since the one-to-one mapping of the LMP vectors and SPRs is not effected by the incomplete load information, this is still a classification problem. The data-driven approach can still be applied but the feature vectors are the system load and a subset of nodal load levels, instead of load levels at every bus in Section \ref{sec:extended_data_driven_approach}.
\begin{figure}[htbp]
  \centering
  \begin{subfigure}[t]{0.49\linewidth}
  \centering
  \includegraphics[width=\linewidth]{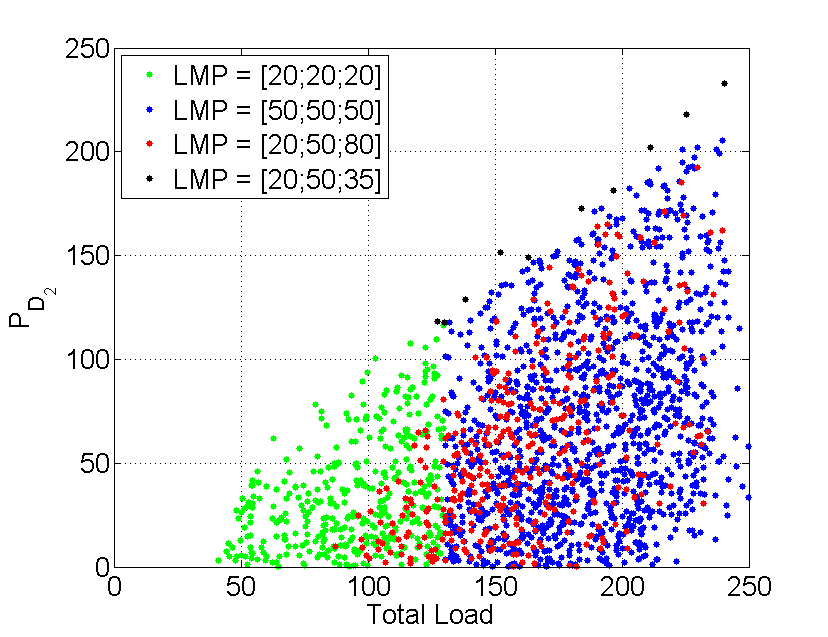}
  \caption{LSE 2}
  \label{fig:LSE2_3bus}
  \end{subfigure}
  \begin{subfigure}[t]{0.49\linewidth}
  \centering
  \includegraphics[width=\linewidth]{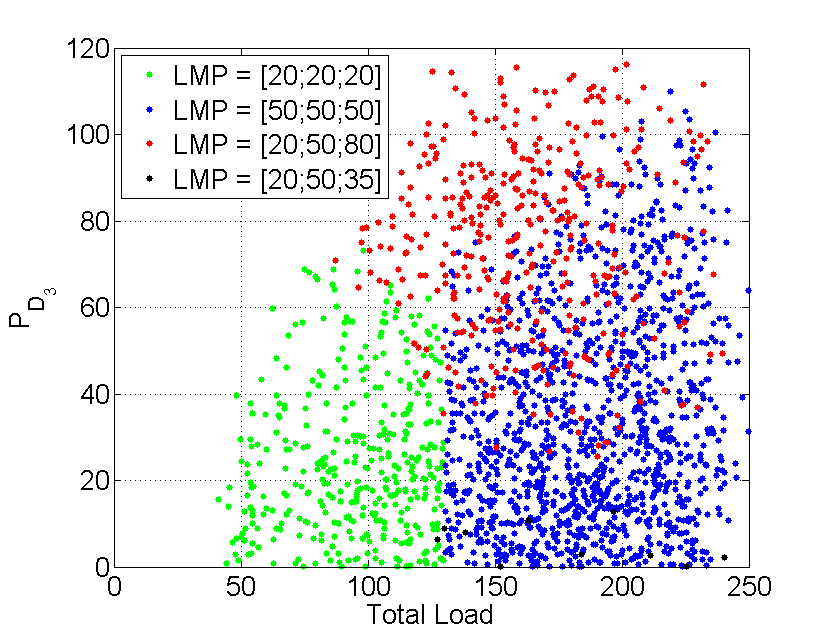}
  \caption{LSE 3}
  \label{fig:LSE3_3bus}
  \end{subfigure}
  \caption{LSEs}
  \label{fig:LMP_LSE}
\end{figure}

Simulation results are summarize in Table.\ref{tab:LMP_forecast_lse_3bus_alllmp}.  
The results indicate that classification accuracy goes down to around around 50\% while the LMP forecast accuracy is still satisfactory. 
This could be explained by the following observations: (1) Fig. \ref{fig:LSE2_3bus} and \ref{fig:LSE3_3bus} are obtained by projecting the 3D SPRs to a lower dimension space. Since the projection is a linear transformation, although the SPRs are overlapping, their boundaries remain linear; (2) the LSEs may care more about their own LMPs. For example, Fig. \ref{fig:LSE2_3bus} could be re-colored by the LMPs at bus 2 (Fig. \ref{fig:LSE2_own_LMP.png}). Since there are only two possibilities of LMPs at bus 2 (20 and 50), there are only two colored regions in Fig. \ref{fig:LSE2_own_LMP.png}. Even with relatively low accuracy of the overall classification, the forecast of LMPs at bus 2 is still accurate.
\begin{table}[htbp]
  \caption{Results 3-bus System}
  \label{tab:LMP_forecast_lse_3bus_alllmp}
  \centering

  \begin{tabular}{l|cccc}
  \hline

  \hline
  \textbf{LSE} & \textbf{LMP@Bus 1} & \textbf{LMP@Bus 2} & \textbf{LMP@Bus 3} & \textbf{Overall}\\
  \hline
  1  & 86.08\%  &97.45\%  & 88.53\%  & 90.69\%\\
  2  &70.69\% & 96.13\% & 89.31\%   &  85.38\%\\
  3  &87.91\% & 98.65\% & 93.53\%  & 93.53\%\\
  \hline
  CLL & 69.48\% & 97.15\% & 89.24\% & 85.29\% \\
  \hline

  \hline
  \end{tabular}
\end{table}

When forecasting a subset of nodal LMPs becomes the major concern, it might be more computationally efficient to formulate the problem in a way as Fig. \ref{fig:LSE2_own_LMP.png} shows. The number of classes decreases significantly and so does the computational burden. But the new colored regions might be the union of SPRs. Though the colored regions in Fig. \ref{fig:LSE2_own_LMP.png} are convex, the union of convex sets are usually non-convex. Because of this, the SVM with linear kernel may not be the best choice. Choosing the best classifier would depend upon the feature of the regions, and will be part of the future work.

\begin{figure}[htbp]
  \centering
  \begin{subfigure}[t]{0.49\linewidth}
  \centering
  \includegraphics[width=\linewidth]{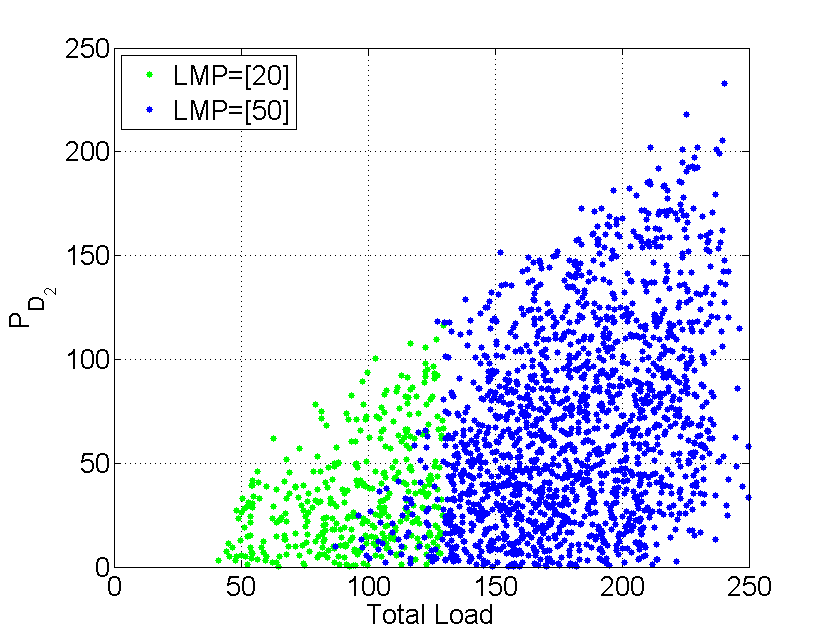}
  \caption{$P_{D_2}$ and system load}
  \label{fig:LSE2_own_LMP.png}
  \end{subfigure}
  \begin{subfigure}[t]{0.49\linewidth}
  \centering
  \includegraphics[width=\linewidth]{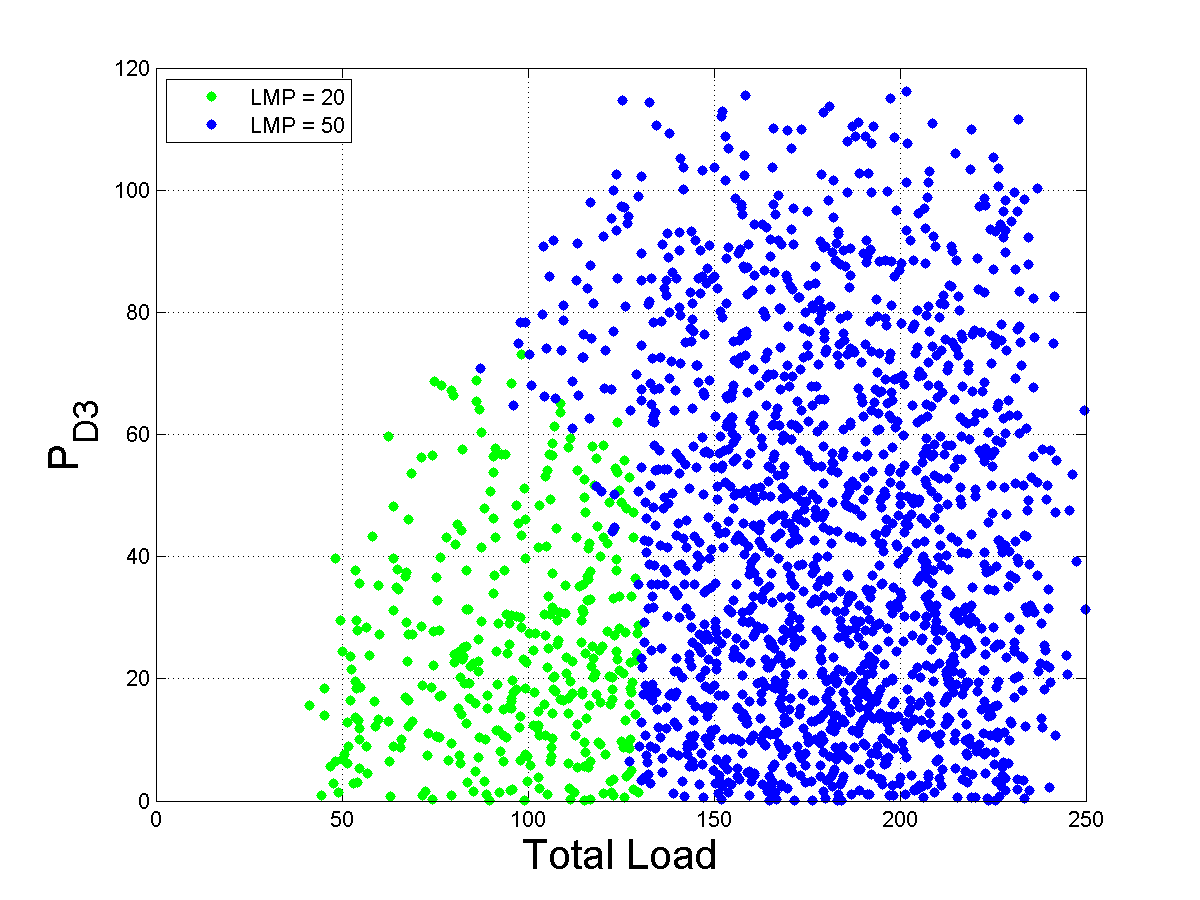}
  \caption{$P_{D_3}$ and system load}
  \label{fig:Load3LMP2}
  \end{subfigure}
  \caption{LMP at bus 2}
\end{figure}

Similar with the case of DLRs or ramp constraints, overlapping SPRs implies uncertainties and the posterior probabilities are necessary. The posterior probabilities for LSE \#2 and \#3 are visualized, respectively.
\begin{figure}[htbp]
  \centering
  \begin{subfigure}[t]{0.49\linewidth}
  \centering
  \includegraphics[width=\linewidth]{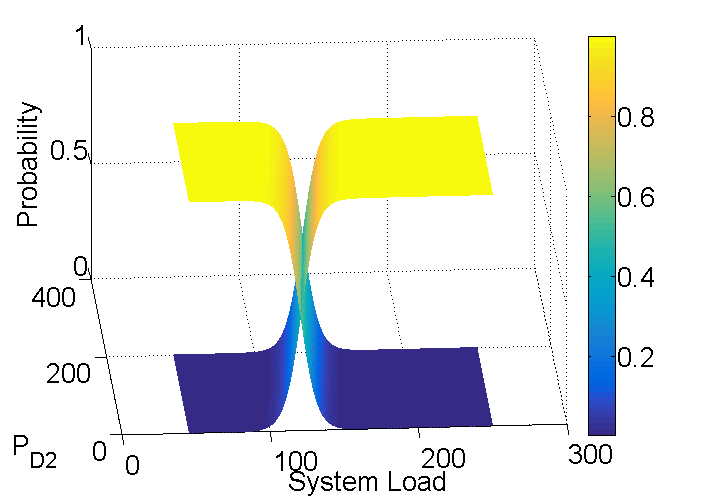}
  \caption{LSE 2}
  \end{subfigure}
  \begin{subfigure}[t]{0.49\linewidth}
  \centering
  \includegraphics[width=\linewidth]{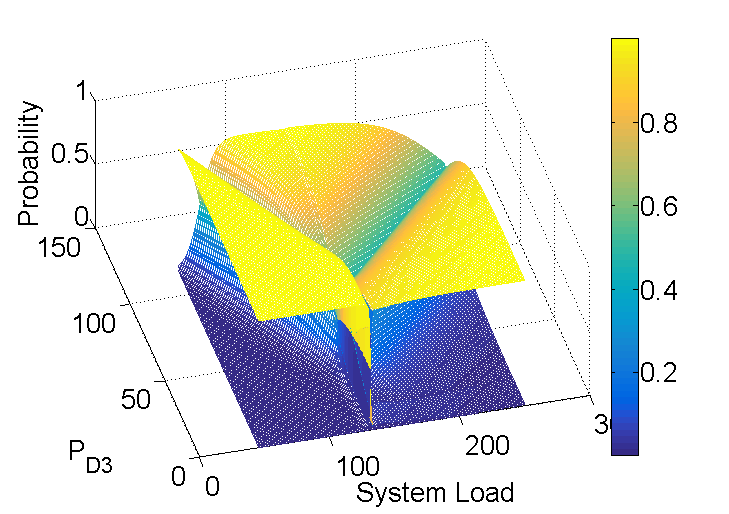} 
  \caption{LSE 3}
  \end{subfigure}
  \caption{Posterior Probability Surfaces}
  \label{fig:posterior_prob_3lses}
\end{figure}

% subsection 3_bus_system (end)

Because of the relatively small resistances of transmission lines, the loss components of LMPs are usually small compared with the other two components. 
Geometrically speaking, each LMP vector is a point in the LMP space and the LMPs of the same SPR form a cluster. The center of the cluster contains the energy, congestion component and the average loss component, 
the deviations from the center represent varying loss components due to different line flows.
We could run a clustering algorithm (e.g. K-means) on the LMP data in order to find out the centers of those clusters. Then the LMP vectors of the same cluster are regarded as the LMPs of the same SPR. By doing so, the SPR identification problem is modeled as a classification problem. The LMP forecast is the forecast of energy components, congestion components and the average loss components.

% subsection lmps_with_loss_component (end)

\section{Discussions} % (fold)
\label{sec:discussions}

\subsection{On Posterior Probabilities} % (fold)
\label{sub:on_posterior_probabilities}
When dealing with uncertainties, it is natural to analyze the data in a probabilistic manner. The calculation of posterior probabilities is essential and provides the quantification of possible risks. We only propose the method to calculate posterior probabilities in this paper, but quantification of the posterior probabilities could yield many interesting applications. For example, LSEs could consider demand response mechanisms to partially change the load vector and thus shift from high price SPRs. Market participants could also estimate the price volatilities due to renewables in a system. Further discussions on how to utilize the posterior probabilities for specific applications are our future work.

% subsubsection about_posterior_probabilities (end)

\subsection{On the Computational Cost} % (fold)
\label{sub:data_requirement_of_the_data_driven_approach}

% \begin{figure}[htbp]
%   \centering
%   % \includegraphics[width=0.5\linewidth]{./fig/todo.jpg}
%   \caption{Data Requirement of the Data-driven Approach}
%   \label{fig:data_requirement}
% \end{figure}
The theoretical analysis reveals that the load space could be partitioned into many SPRs. This overall structure of the load space could help solve the SCED problem and shift part of the online computational burden to offline \cite{Jia}. The total number of SPRs could help evaluate computational burden to some extent.
% Lemma \ref{lem:finiteSPR} provides a (loose) upper bound of the total number of SPRs.
% \begin{remark}
% \label{rmk:finiteSPR}

With MPT 3.0, the exact number of SPRs of some IEEE benchmark systems are calculated. Though the total number of SPRs is finite\footnote{A loose upper bound is $2^{n_g-1}\times C_{n_g+n_l}^{n_g-1}$.}, it grows extremely fast with the scale of the system. However, with the Monte-Carlo simulation, we found much less SPRs than the theoretical results. \cite{Zhou2011} points out that because of the regular patterns of loads, only some subsets of the complete theoretical load space could be achievable thus helpful in practice. Therefore, only a small subset of the SPRs is meaningful to be analyzed. This suggests the great potential of reducing the computational burden.
The proposed approach is also parallel computation-friendly, which could be very useful when dealing with large-scale simulations.

\begin{table}[htbp]
  \caption{Number of SPRs of Some Benchmark Systems}
  \label{tab:number_spr}
  \centering

  \begin{tabular}{l|c|c}
  \hline

  \hline
  \textbf{System Info} & \textbf{MPT 3.0} & \textbf{Simulation (8640 points)} \\
  \hline
  3 Bus System (Fig.\ref{fig:3Bus2GeneSystem}) & 5 & 4 \\
  IEEE 6 Bus System & 20 & 7 \\
  IEEE 9 Bus System & 15 &  7\\
  IEEE 14 Bus System & 1470 & 50 \\
  IEEE 24 Bus System & $\sim 10^6$ & 445 \\
  IEEE 118 Bus System & - & 971 \\
  % IEEE 24 Bus System & $2.7\times 10^6$ & about 10 days \\
  \hline

  \hline
  \end{tabular}
\end{table}

% subsection data_requirement_of_the_data_driven_approach (end)

\subsection{On Generation Offer Prices} % (fold)
\label{sub:on_generation_offer_prices}
The marginal costs of generators are fluctuating due to many factors such as oil prices. This leads to the changes of generation
offer prices c in the SCED formulation. Intuitively, the SPRs would change with respect to large offer price variations. Eqn. \ref{eqn:determine_y} in Lemma \ref{lem:complementary_slackness} quantifies the variation of offer prices: for a system pattern $\pi = (\mathcal{B}, \mathcal{N})$, the corresponding SPR $\mathcal{S}_{\pi}$  would remain the same as long as the generation cost vector $c$ satisfies Eqn. (\ref{eqn:determine_y}).

An illustrative example is provided below. A diesel turbine is added at bus 3 in Fig. \ref{fig:3Bus2GeneSystem}, the new 3-bus 3-generator 2-load system is shown in Fig. \ref{fig:3BusSystem3Gene}. Suppose the offer price of the diesel turbine is varying due to the fluctuations of oil prices. Fig. \ref{fig:c_205065} shows the SPRs when the offer price of the new
generator is 65; when the offer price increases from 65 to 100, three SPRs are different while the others remain the same\footnote{More specifically, we can calculate the condition from Eqn. (\ref{eqn:determine_y}): if the offer price of the new generator satisfies $c_3 < 2c_2 - c_1 = 80$, then the SPRs in Fig. \ref{fig:c_205065} would remain the same}.
This shows that the SPRs have some extent of robustness to the varying generation offer prices.

\begin{figure}[htbp]
  \centering
  \includegraphics[width=0.6\linewidth]{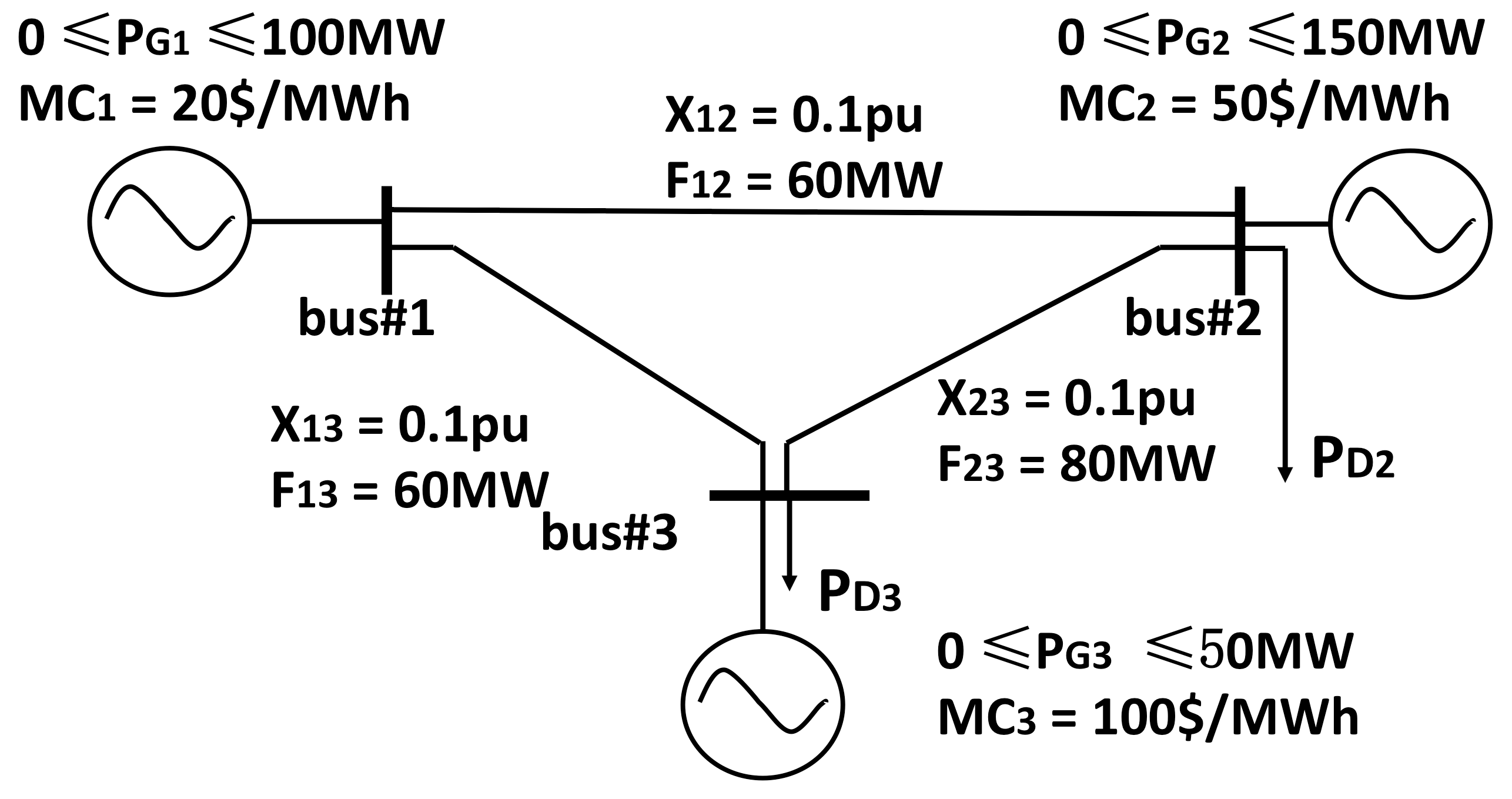}
  \caption{The 3-bus 3-generator 2-load System}
  \label{fig:3BusSystem3Gene}
\end{figure}

\begin{figure}[htbp]
  \centering
  \begin{subfigure}[t]{0.49\linewidth}
  \centering
  \includegraphics[width=\linewidth]{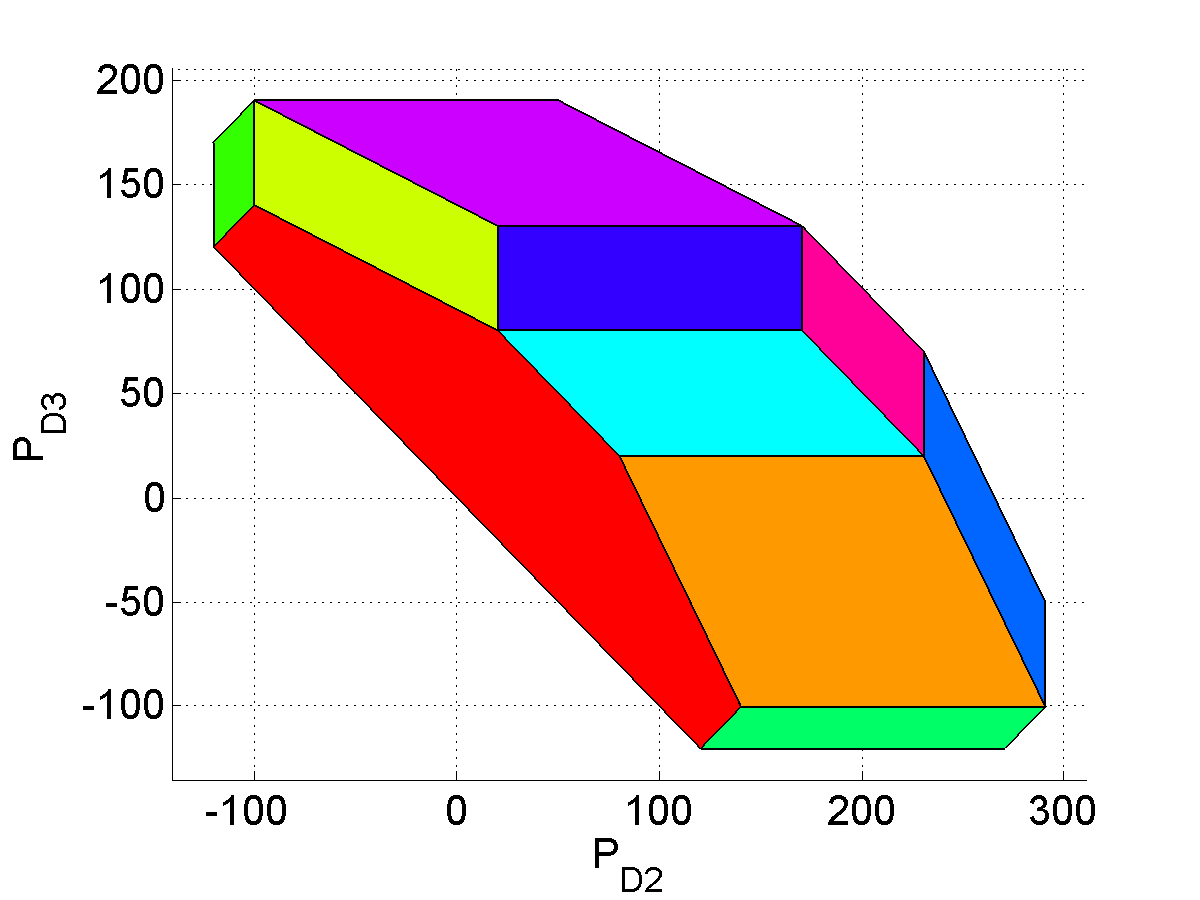}
  \caption{$c = (20,50,65)$}
  \label{fig:c_205065}
  \end{subfigure}
  \begin{subfigure}[t]{0.49\linewidth}
  \centering
  \includegraphics[width=\linewidth]{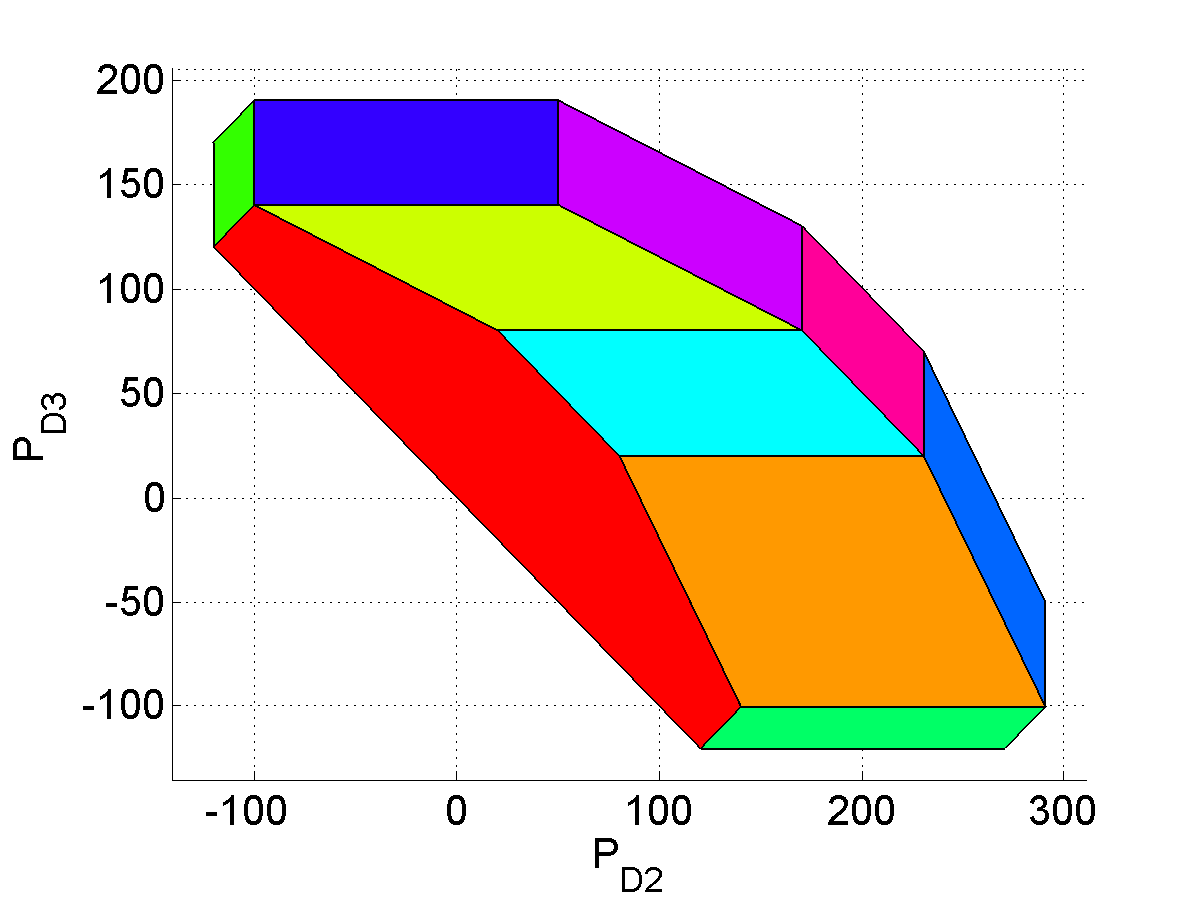}
  \caption{$c = (20,50,100)$}
  \label{fig:c_2050100}
  \end{subfigure}
  \caption{System Pattern Regions with Different Generation Offer Prices}
  \label{fig:SPRs_with_diff_offer_price}
\end{figure}

% subsection on_generation_offer_prices (end)

\subsection{LMPs with Loss Components} % (fold)
\label{sub:lmps_with_loss_components}
Since the line losses are not explicitly modeled in the SCED formulation, all the theoretical analysis is conducted on the ``lossless''LMP vectors. The LMP forecast discussed above is the forecast of the energy components and congestion components. In the reality, the proposed method could be applied directly on the markets not considering line losses (e.g. ERCOT), and the electricity markets providing the energy component, congestion component and loss component separately (e.g. MISO). There are many possible methods to forecast the loss components, but that is a separate story.

There are Economic Dispatch models with line losses explicitly modeled (e.g. \cite{Litvinov2004}), similar analysis using MLP theory could be conducted but it is beyond the scope of this paper.
% subsection lmps_with_loss_components (end)

% The x-axis of Fig. \ref{fig:N_SPR} is the number of buses of test systems, and the y-axis is the log of the total number of SPRs. The black circles are calculated numbers of SPRs, the red line is the upper bound provided in Lemma \ref{lem:finiteSPR}. Although this upper bound is relatively loose, it satisfyingly describes the trend of number of SPRs. The green line indicates an attempt to fit the actual number of SPRs using the upper bound. The equation of the green line is $2^{n_g-1}\times C_{n_g+n_l}^{n_g-1}/1.12^{n_g+n_l}$.

% subsection practical_issues_about_sprs (end)

\section{Conclusions} % (fold)
\label{sec:conclusion}
In this paper, we examine the fundamental coupling between nodal load levels and LMPs in real-time SCED.
It is shown that the load space can be partitioned into convex system pattern regions, which are one-to-one mapped with distinct LMP vectors. 
Based on the theoretical results, we propose a data-driven learning algorithm for market participants to identify SPRs. 
Identifying SPRs is modeled as a classification problem, and the proposed data-driven approach is built upon a ``one-vs-one'' multi-class SVM classifier. 
The proposed algorithm is shown to be capable of estimating SPRs solely from historical data without knowing confidential system information such as network topology and bidding curves.
The approach is shown to be extensible towards considering dynamic line ratings, line losses and partial load information. 
Simulation results based on the IEEE 118-bus system demonstrates that the proposed algorithm is effective in understanding the past and predicting the future.

This paper is a first step towards developing theoretically
rigorous and computationally feasible algorithms to analyzing
the market prices as a result of varying loading levels. Future
work should investigate: (1) the system pattern regions with
different unit commitment results and system topologies; (2)
the impacts of multi-interval temporal constraints on the
system pattern regions. Another important avenue of research
is to develop efficient learning algorithm to process a large
amount of historical data in near real-time market operations.

\appendices
\section{Platt's Algorithm} % (fold)
\label{sec:platt_s_algorithm_}
\emph{Platt's algorithm} \cite{Platt1999} is used to calculate the posterior probabilities of a \emph{binary} SVM classifier. According to \cite{Platt1999}, training data is fitted to a sigmoid function by minimizing the negative log likelihood function.
\begin{subequations}
\begin{align}
  \min_{A,B} & \qquad -\sum_k t_k \log(r_k) + (1-t_k)\log(1-r_k) \label{eqn:log_likelihood} \\
  \text{where} & \qquad r_k = \mathbb{P}(y=i | P_D^{(k)} \text{ and } y\in \{i,j\} ) \nonumber \\
    & \qquad = (1+e^{A P_D^{(k)}+B})^{-1} \label{eqn:sigmoid}   
\end{align}
\end{subequations}
$P_D^{(k)}$ is the $k$th load vector in the training set, and the parameter $t_k$ associated with $P_D^{(k)}$ is calculated by Eqn. (\ref{eqn:calculate_t}), where $N_{+}$ ($N_{-}$) is the number of positive (negative) examples.
\begin{equation}
\label{eqn:calculate_t}
t_k = \left\{
\begin{array}{ll}
    \frac{N_+ +1}{N_+ + 2} & \text{if } y_k = +1 \\
    \frac{1}{N_- + 2} & \text{if } y_k = -1
\end{array}
\right.
\end{equation}
% section platt_s_algorithm_ (end)
\section{Hastie \& Tibshirani's Algorithm} % (fold)
\label{sec:hastie_tibshirani_s_algorithm}
\emph{Hastie \& Tibshirani's Algorithm} \cite{hastie1998classification} is widely accepted to calculate the \emph{multi-class} posterior probabilities. 

The basic idea behind Hastie \& Tibshirani's Algorithm is pretty straightforward: with the \emph{multi-class} posterior probabilities (i.e. $\hat{p_i} = \mathbb{P}(y=i | P_D \text{ and } i \in \{1,2,\cdots,n\} )$), we could estimate the \emph{binary} posterior probabilities:
\begin{equation}
  \hat{\mu_{ij}} = \frac{\hat{p_i}}{\hat{p_i}+\hat{p_j} }
\end{equation}
If $\hat{p_i}$s are correctly estimated, then the estimation of \emph{binary} posterior probabilities $\hat{\mu_{ij}}$ should be identical to the calculated binary posterior probabilities $r_{ij}$s from Platt's algorithm. Therefore the objective of Hastie \& Tibshirani's Algorithm is to minimize the Kullback-Leibler distance between $\hat{\mu_{ij}}$ and $r_{ij}$. Details are summarized below:

The \emph{first step} of Hastie \& Tibshirani's Algorithm is to calculate the \emph{binary} posterior probabilities $r_{ij}$ based on Platt's algorithm. 

The \emph{second step} of Hastie \& Tibshirani's Algorithm is to run the following algorithm until convergence:
\begin{enumerate}
  \item Start with the initial guess for the $\hat{p_i}$ and $\hat{\mu_{ij}} = \hat{p_i}/(\hat{p_i}+\hat{p_j} ) $.
  \item Repeat this ($i=1,2,\cdots,n,1,2,\cdots$) until convergence:
  % \begin{equation}
  %   \hat{p_i} \leftarrow \hat{p_i}\sum_{j\ne i} n_{ij}r_{ij}/\sum_{j\ne i} n_{ij}\mu_{ij}
  % \end{equation}
  \begin{equation}
  \label{eqn:iteration_posterior}
    \hat{p_i} \leftarrow \hat{p_i}\frac{\sum_{j\ne i} n_{ij}r_{ij}}{\sum_{j\ne i} n_{ij}\mu_{ij}}
  \end{equation}
  Then renormalize $\hat{p_i} \leftarrow \hat{p_i}/\sum_{j=1}^n {\hat{p_j}}$ and recompute $\hat{\mu_{ij}} = \hat{p_i}/(\hat{p_i}+\hat{p_j} ) $.
  \item If $\mathbf{\hat{p}}/\sum{\hat{p_i}}$ converges to the same $\mathbf{\hat{p}}$, then the algorithm stops, the vector $\mathbf{\hat{p}}$ will be the estimated multi-class posterior probabilities.
\end{enumerate}

% section hastie_&_tibshirani_s_algorithm (end)

\section{Proof of Theorem \ref{thm:diff_LMPs}} % (fold)
\label{sec:proof_diff_LMP_diff_SPR}
This section provides the details of the proof of the theorem ``different system pattern regions (SPRs) have different LMP vectors''.

\subsection{Basics of the SCED} % (fold)
\label{sub:basics_of_the_sced}
Consider the SCED problem in the form of Eqn. (\ref{eqn:ED_primal}). Its Lagrangian $L: \mathbb{R}^{n_b}\times \mathbb{R} \times \mathbb{R}^{n_l}\times \mathbb{R}^{n_l}\times \mathbb{R}^{n_g}\times \mathbb{R}^{n_g} \rightarrow \mathbb{R}$ is:
\begin{eqnarray}
  & & L(P_G,\lambda_1, \mu_+, \mu_-, \eta_+, \eta_-) \nonumber \\
  &=&  c^\intercal P_G + \lambda_1 (\mathbf{1}^\intercal P_G - \mathbf{1}^\intercal P_D) \nonumber \\
  &+& \mu_+^\intercal (H P_G - H P_D - F_+) - \mu_-^\intercal (H P_G - H P_D - F_-) \nonumber \\
  &+& \eta_+^\intercal (P_G - G_+) - \eta_-^\intercal (P_G - G_-)
\end{eqnarray}
According to KKT conditions, $\mu_+, \mu_-, \eta_+, \eta_- \ge 0$ and 
\begin{equation}
  \nabla_{P_G}^\intercal L = 0 \Rightarrow c + \lambda_1 \mathbf{1} + H^\intercal \mu + \eta = 0 \label{eqn:kkt_nabla_is_0}
\end{equation}
where $\mu = \mu_+ - \mu_-$ and $\eta = \eta_+ - \eta_-$.
The LMP vector $\lambda$ can be calculated:
\begin{equation}
  \lambda = -\nabla_{P_D}^\intercal L = \lambda_1 \mathbf{1} + H^\intercal \mu
\end{equation}
This is consistent with \cite{Wu1996}.

% subsection basics_of_the_sced (end)

\subsection{Preparation} % (fold)
\label{sub:preparation}
\subsubsection{On General Mathematics} % (fold)
\label{ssub:two_lemmas}
The lemma and definition below lie the foundation of the proof of Theorem \ref{thm:diff_LMPs}.
\begin{lem}[Convex Piecewise Linear Functions With Parallel Segments]
\label{lem:convex_piecewise_linear_function_parallel_segments}
Assume the piecewise linear function $f: \mathbb{R}^{n}\rightarrow \mathbb{R}$ is composed of $m$ linear functions $f_k = c_k^\intercal x$ where $k = 1,2,\cdots, m$. Let $D = \text{dom} f$, $D_k = \text{dom } f_k$\footnote{It is obvious that $D = \cup_k D_k$.  }, and assume $D$ and $D_i$ are closed convex sets \footnote{The word ``closed'' indicates for $i\ne j$: $D_i \cap D_j \ne \emptyset$, but $\text{relint } D_i \cap \text{relint } D_j = \emptyset$}. 
If $f$ is convex and has two parallel segments: $f_i$ and $f_j$ ($i\ne j)$ with $\nabla^\intercal f_i = \nabla^\intercal f_j$, then $f_i$ and $f_j$ have to be on the same hyperplane. Namely:
\begin{enumerate}
\item $\forall x_i \in \text{relint } D_i, \forall x_j \in \text{relint } D_j$, $f(x_j) = f(x_i) + \nabla^\intercal f(x_i)\cdot (x_j - x_i)$.
\item $\forall x_i \in \text{relint } D_i, \forall x_j \in \text{relint } D_j$, if the convex combination of $x_i$ and $x_j$ belongs to $D_k$ instead of $D_i$ or $D_j$ ($i\ne j \ne k$), then $\nabla^\intercal f_k = \nabla^\intercal f_i = \nabla^\intercal f_j$ and $f(x_k) = f(x_i) + \nabla^\intercal f(x_i)\cdot (x_k - x_i)$.
\end{enumerate}
\end{lem}
Lemma \ref{lem:convex_piecewise_linear_function_parallel_segments} could be easily proved by applying the first-order conditions and the definition of convex functions.

\begin{defn}[Adjacent Sets]
\label{defn:adjacent_sets}
Given two closed set $D_i$ and $D_j$, and $\text{dim}(D_i) = \text{dim}(D_j) = d \ge 2$. 
We say $D_i$ and $D_j$ are \emph{adjacent} if $D_i \cap D_j \ne \emptyset$ and $\text{dim}(D_i \cap D_j) = d-1$.
\end{defn}

% This definition is used for defining \emph{adjacent} SPRs in the next section.
% subsubsection two_lemmas (end)

\subsubsection{On the Features of SPRs} % (fold)
\label{ssub:on_the_features_of_sprs}
According to the literatures on Multi-parametric Linear Programming theory:
\begin{lem}
If the problem is not degenerate, then the partition of the load space is unique, and $S_{\pi}$  is an open polyhedron of the same dimension as $\mathbb{D}$ \cite{filippi1997geometry}. This indicates the dimensions of all the SPRs are the same.
\end{lem}
\begin{lem}
\label{lem:affine_solution_MLP}
The optimal value function $f^*(P_D) = c^\intercal  P_G^*(P_D)$ is convex and piecewise affine over $\mathcal{D}$, and affine in each SPR. The optimal solution $P_G^*$ within an SPR is an affine function of the load vector $P_D$ \cite{gal1995postoptimal}.
\end{lem}
The 3-bus 3-generator 2-load system in Fig. \ref{fig:3BusSystem3Gene} is analyzed via MPT 3.0. The optimal value function and primal solutions are demonstrated in Fig. \ref{fig:affine_solution_MLP}. This verifies Lemma \ref{lem:affine_solution_MLP}.

\begin{figure}[htbp]
  \centering
\begin{subfigure}[t]{0.49\linewidth}
  \includegraphics[width=\linewidth]{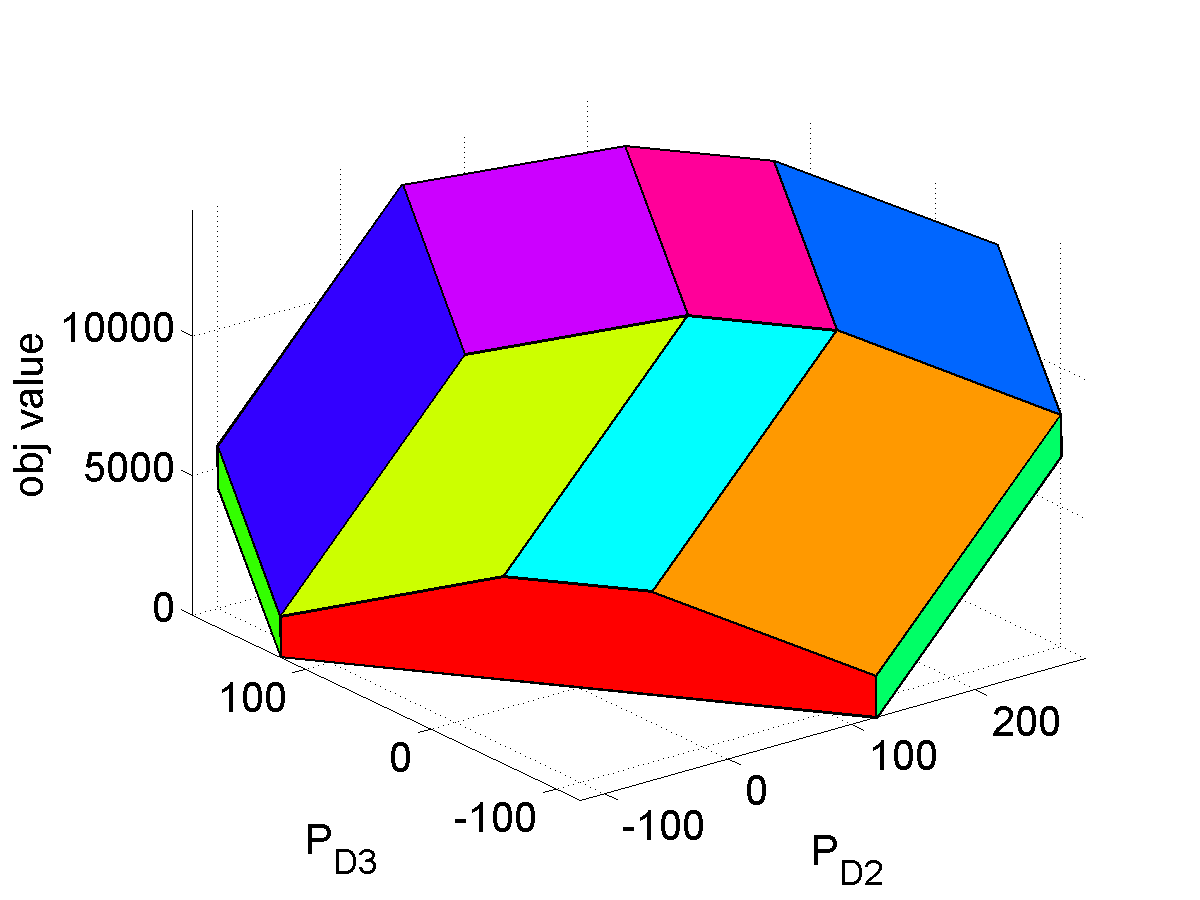}
  \caption{Optimal Value Function}
\end{subfigure}
\begin{subfigure}[t]{0.49\linewidth}
  \includegraphics[width=\linewidth]{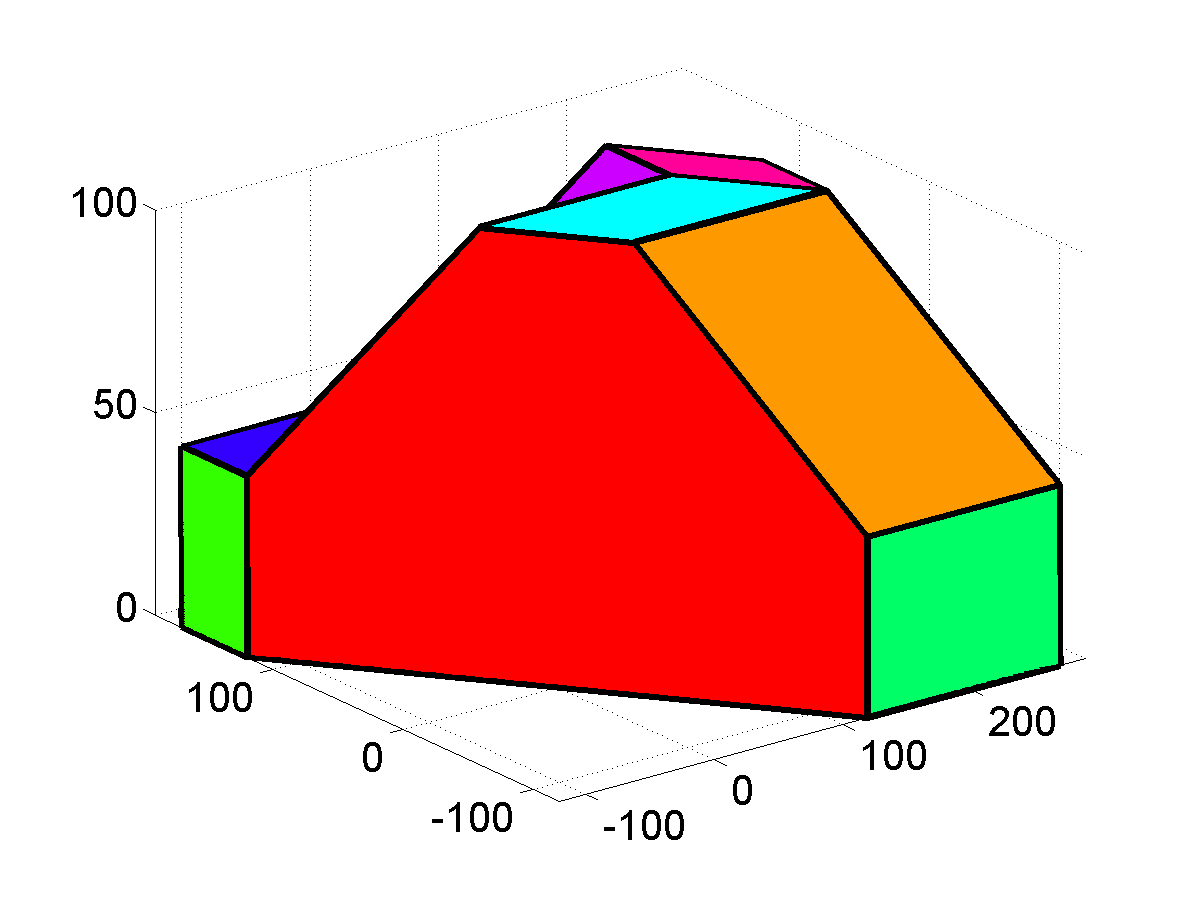}
  \caption{Primal Solution}
\end{subfigure}
  \caption{Piecewise Affine Optimal Value Function and Primal Solution}
  \label{fig:affine_solution_MLP}
\end{figure}

\begin{lem}[System Patterns of Adjacent SPRs]
\label{lem:adjacent_SPRs}
Given two system pattern regions (SPRs) $S_i$ and $S_j$ and their system patterns $\pi_i = (\mathcal{B}_i, \mathcal{N}_i ) $ and $\pi_j = (\mathcal{B}_j, \mathcal{N}_j )$. If $S_i$ and $S_j$ are \emph{adjacent}\footnote{See Definition \ref{defn:adjacent_sets}.}, then $\mathcal{B}_i$ and $\mathcal{B}_j$ only differ in one entry.
\end{lem}
\begin{proof}
Lemma \ref{lem:adjacent_SPRs} is a direct conclusion from Definition \ref{defn:adjacent_sets}.
If $S_i$ and $S_j$ are \emph{adjacent} but $\mathcal{B}_i$ and $\mathcal{B}_j$ differ in $k\ge 2$ entries. Then $S_i \cap S_j$ is depicted by $k$ linear constraints. If the constraints are linear independent\footnote{If they are linear dependent, we can always eliminate the redundant constraints, which will not make any difference.}, then $\text{dim}  S_i \cap S_j = \text{dim}  S_i - k < \text{dim}  S_i - 1$. This is contradictory with the definition of \emph{adjacent} sets, which requires $\text{dim}  S_i \cap S_j = \text{dim}  S_i -1$.
\end{proof}

% subsection preparation (end)

\subsection{Sketch of the Proof} % (fold)
\label{sub:sketch_of_the_proof}
The optimal value function of SCED problem $f^*$ is convex and piecewise affine over $\mathcal{D}$ and affine in each SPR (Lemma \ref{lem:affine_solution_MLP}). According to the definition of LMPs, the LMP vector is the norm vector of the affine segment of $f^*$. 
If two different SPRs have the same LMP, then their norm vectors are the same, thus $f^*$ has two parallel linear segments. Lemma \ref{lem:convex_piecewise_linear_function_parallel_segments} claims the existence of two adjacent SPRs \footnote{May not be exactly same as previous two.} with the same LMP vectors. The system patterns of those two adjacent SPRs, according to Lemma \ref{lem:adjacent_SPRs}, differ in only one entry. 
There are only two possibilities: (1) those two SPRs have one different binding generation constraints, but the binding transmission constraints are the same; (2) those two SPRs have one different binding transmission constraints, but the binding generation constraints are the same.

\subsection{Proof of Theorem \ref{thm:diff_LMPs}} % (fold)
\label{sub:proof_of_theorem_ref_thm_diff_lmps}
A complete proof is provided in this section.

Assume there exist two SPRs ($i,j$) which have the same LMP vector $\lambda^{(i)} = \lambda^{(j)}$.
Notice that this equality $\lambda^{(i)} = \lambda^{(j)}$ is true for each entry. The LMPs of the slack bus (i.e. bus \#1) are the same: $\lambda_1^{(i)} = \lambda_1^{(j)}$. In other words, the \emph{energy components} of the LMP vectors are the same. According to $H^\intercal \mu = \lambda - \lambda_1 \mathbf{1}$, the \emph{congestion components} are also the same:
\begin{equation}
H^\intercal \mu^{(i)} = H^\intercal \mu^{(j)}
\end{equation}

According to Eqn. (\ref{eqn:kkt_nabla_is_0}):
\begin{equation}
\label{eqn:multiplier_gen} 
  \eta^{(i)} =  -c - \lambda^{(i)} = -c - \lambda^{(j)} = \eta^{(j)}
\end{equation}
$\eta^{(i)} = \eta^{(j)} $ means the marginal generators (which are ON) of the two SPRs ($i,j$) are also exactly the same.

Also, according to the analysis in Appendix \ref{sub:sketch_of_the_proof}, there are two \emph{adjacent} SPRs ($i$ and $k$\footnote{$k\ne j$ is possible.}) with the same LMP vectors. And the system pattern of adjacent SPRs only differ in one entry. This indicates that there is only one different binding constraints: either one different congested line or one different marginal generator.

We will discuss these two cases separately:

% subsubsection two_system_pattern_regions_having_the_same_lmp_vector (end)

\subsubsection{Case 1: same generation shadow prices represent different system patterns} % (fold)
\label{ssub:case1_same_trans_diff_gen}
For SPR $i$ and $k$, similar arguments would show that $\lambda^{(i)} = \lambda^{(k)}$ indicates $\eta^{(i)} = \eta^{(k)}$. If SPR $i$ and $k$ are different, then the same vector $\eta = \eta^{(i)} = \eta^{(k)}$ represents two different sets of marginal generators. This is possible only when there are some ``equivalent'' generators. For those ``equivalent'' generators, if we increase the output of one and decrease the other one by the same amount, the total generation cost remains the same and optimal. That means both generation outputs are optimal, the SCED problem has infinite optimal solutions. This is the case that SCED is degenerate and the Lagrange multipliers cannot be uniquely determined.

% In this case, we call those generators hits the maximum (minimum) capacity constraint as ``max (min) generators''. 
% Since their generation costs are the same, increasing the output of one min generator but decreasing the same amount of another one will not change the total system generation cost. It is worth noting that this would not change the congestion pattern, otherwise the $\eta$ vector would be different due to the cost of congestions.
% For all those ``equivalent'' generators, their Lagrange multipliers are the same (generation cost). This means they could be adjusted and therefore not hit either upper or lower bounds. Their generation capacity constraints are not binding. 

% The system pattern region (SPR) relates with this case is still convex. The reason is as follow: the SPR is constrained by several hyperplanes (linear constraints in the load space) regardless how many inequalities there are in the problem. The halfspace is convex, and the intersection of halfspaces is also convex. Therefore the SPR is still convex. 

% We could therefore argue that this case could be further reduced by regarding the ``equivalent generators'' as one generator but with larger capacity.

% subsubsection case_1_congestion_lines_are_the_same_but_there_is_one_different_marginal_generator_ (end)

\subsubsection{Case 2: same congestion component represent different congested lines} % (fold)
\label{ssub:case_2_marginal_generators_are_the_same_but_there_is_one_different_congested_line_}
We will show this case is not possible.

% We first rearrange the order of constraints. In Eqn. (\ref{eqn:details_of_the_A_W_b}), the order (of rows) is: (1) supply-demand balance constraint; (2) transmission limit constraints and (3) generation capacity constraints. We rearrange that to be: (1) supply-demand balance constraint; (2) generation capacity constraints and (3) transmission limit constraints. 

There is one different congested line between SPR $i$ and SPR $j$. Since we can label the congested lines with any non-repetitive numbers, assume line $1$ is congested in SPR $i$ but not congested in SPR $j$. Similarly, line $2$ is congested in SPR $j$ but not congested in SPR $i$. And the index set of all the other lines congested in both SPR $i$ and SPR$j$ is denoted by $\mathcal{C}$.

In our previous settings, the matrix $A_{\mathcal{B}}$, which relates with binding constraints, has the structure:
\begin{equation}
A_{\mathcal{B}} = 
  \begin{bmatrix}
    \text{supply-demand balance: } \mathbf{1}_{n_b}^\intercal  \\
    \text{shift factor matrix related with congested lines} \\
    \text{matrix related with generation constraints}
  \end{bmatrix}
\end{equation}
We rearrange the structure of matrix $A_{\mathcal{B}}$ as follows\footnote{This step will not make any difference to the theoretical results, but will significantly simplify the notations.}:
\begin{equation}
A_{\mathcal{B}} = 
  \begin{bmatrix}
    \text{supply-demand balance: } \mathbf{1}_{n_b}^\intercal  \\
    \text{matrix about generation constraints: }  G\\
    \text{shift factor matrix of commonly congested lines: } H_{\mathcal{C}} \\
    \text{shift factor matrix of uniquely congested lines: } H_U\\
  \end{bmatrix}
\end{equation}
For SPR $i$ and SPR $k$:
\begin{equation}
A_{\mathcal{B}_i} = 
  \begin{bmatrix}
    \mathbf{1}_{n_b}^\intercal \\
    G \\
    H_{\mathcal{C}} \\
    h_1 \\
  \end{bmatrix} 
  = \begin{bmatrix}
    E \\
    H_{\mathcal{C}} \\
    h_1 \\
  \end{bmatrix}, 
A_{\mathcal{B}_k} = 
  \begin{bmatrix}
    \mathbf{1}_{n_b}^\intercal \\
    G \\
    H_{\mathcal{C}} \\
    h_2 \\
  \end{bmatrix}
\begin{bmatrix}
    E \\
    H_{\mathcal{C}} \\
    h_2 \\
  \end{bmatrix} 
\end{equation}
where $[\mathbf{1}_{n_b}^\intercal; G]$ is common for both SPR $i$ and SPR $k$, we use matrix $E = [\mathbf{1}_{n_b}^\intercal; G]$ to represent it. $H_{\mathcal{C}}$ is the shift factor matrix related with lines congested in both SPR $i$ and SPR $k$. $h_1$ is the row of line $1$ in the shift factor matrix $H$, $h_2$ is the row of line $2$ in the shift factor matrix $H$. $h_1, h_2 \in \mathbb{R}^{1\times n_b} $ are row vectors, where $n_b$ is the number of buses.

The structure of $A_{\mathcal{B}_i}^{-1}$ and $A_{\mathcal{B}_k}^{-1}$ is as follows:
\begin{eqnarray}
  A_{\mathcal{B}_i}^{-1} &=& 
  \begin{bmatrix}
    B_1^{(i)} & B_2^{(i)} & \beta_3^{(i)}
  \end{bmatrix}\\
  A_{\mathcal{B}_k}^{-1} &=& 
  \begin{bmatrix}
    B_1^{(k)} & B_2^{(k)} & \beta_3^{(k)}
  \end{bmatrix}
\end{eqnarray}
where $\beta_3^{(k)}, \beta_3^{(i)} \in \mathbb{R}^{n_b \times 1}$.

Since $A_{\mathcal{B}_k} A_{\mathcal{B}_k}^{-1} = \mathbf{I}$:
\begin{equation}
\label{eqn:inverse_AB_k}
  A_{\mathcal{B}_k} \cdot A_{\mathcal{B}_k}^{-1} = 
  \begin{bmatrix}
    E B_1^{(k)} & E B_2^{(k)} & E \beta_3^{(k)} \\
    H_{\mathcal{C}} B_1^{(k)} & H_{\mathcal{C}} B_2^{(k)} & H_{\mathcal{C}} \beta_3^{(k)} \\
    h_2 B_1^{(k)} & h_2 B_2^{(k)} & h_2 \beta_3^{(k)}
  \end{bmatrix}
  = \begin{bmatrix}
    \mathbf{I} & \mathbf{0} & \mathbf{0} \\
    \mathbf{0} & \mathbf{I} & \mathbf{0} \\
    \mathbf{0} & \mathbf{0} & 1
  \end{bmatrix}
\end{equation}

An interesting observation:
\begin{eqnarray}
  & & (A_{\mathcal{B}_k}^\intercal)^{-1} A_{\mathcal{B}_i}^\intercal = (A_{\mathcal{B}_k}^{-1})^\intercal A_{\mathcal{B}_i}^\intercal = (A_{\mathcal{B}_i}A_{\mathcal{B}_k}^{-1})^\intercal  \nonumber \\
  &= &
  \begin{bmatrix}
    E B_1^{(k)} & E B_2^{(k)} & E \beta_3^{(k)} \\
    H_{\mathcal{C}} B_1^{(k)} & H_{\mathcal{C}} B_2^{(k)} & H_{\mathcal{C}} \beta_3^{(k)} \\
    h_1 B_1^{(k)} & h_1 B_2^{(k)} & h_1 \beta_3^{(k)}
  \end{bmatrix}^\intercal 
  = 
  \begin{bmatrix}
    \mathbf{I} & \mathbf{0} & \mathbf{0} \\
    \mathbf{0} & \mathbf{I} & \mathbf{0} \\
    h_1 B_1^{(k)} & h_1 B_2^{(k)} & h_1 \beta_3^{(k)}
  \end{bmatrix}^\intercal
\end{eqnarray}
Multiply $y_{\mathcal{B}_i}$ on both sides:
\begin{eqnarray}
& & (A_{\mathcal{B}_k}^\intercal)^{-1}A_{\mathcal{B}_i}^\intercal \times y_{\mathcal{B}_i} 
= ((A_{\mathcal{B}_k}^\intercal)^{-1}A_{\mathcal{B}_i}^\intercal) \times y_{\mathcal{B}_i} \nonumber \\
&=& \begin{bmatrix}
    \mathbf{I} & \mathbf{0} & (B_1^{(k)})^\intercal h_1^\intercal  \\
    \mathbf{0} & \mathbf{I} & (B_2^{(k)})^\intercal h_1^\intercal  \\
    \mathbf{0} & \mathbf{0} & h_1 \beta_3^{(k)}
  \end{bmatrix}
  \begin{bmatrix}
    \alpha \\
    \mu_{\mathcal{C}}^{(i)}  \\
    \mu_1
  \end{bmatrix} 
  = \begin{bmatrix}
    \alpha \\
    \mu_{\mathcal{C}}^{(i)} \\
    \mathbf{0}
  \end{bmatrix} + \mu_1
  \begin{bmatrix}
    (B_1^{(k)})^\intercal h_1^\intercal  \\
    (B_2^{(k)})^\intercal h_1^\intercal  \\
    h_1 \beta_3^{(k)}
  \end{bmatrix}
\end{eqnarray}
Also:
\begin{equation}
  (A_{\mathcal{B}_k}^\intercal)^{-1}A_{\mathcal{B}_i}^\intercal \times y_{\mathcal{B}_i} = (A_{\mathcal{B}_k}^\intercal)^{-1}(A_{\mathcal{B}_i}^\intercal \times y_{\mathcal{B}_i})
  =(A_{\mathcal{B}_k}^\intercal)^{-1}\times (-c) = y_{\mathcal{B}_k} = 
  \begin{bmatrix}
    \alpha \\
    \mu_{\mathcal{C}}^{(k)}  \\
    \mu_2
  \end{bmatrix}
\end{equation}
Therefore:
\begin{equation}
  \begin{bmatrix}
    \alpha \\
    \mu_{\mathcal{C}}^{(k)}  \\
    \mu_2
  \end{bmatrix} = 
  \begin{bmatrix}
    \alpha \\
    \mu_{\mathcal{C}}^{(i)} \\
    \mathbf{0}
  \end{bmatrix} + \mu_1
  \begin{bmatrix}
    (B_1^{(k)})^\intercal h_1^\intercal  \\
    (B_2^{(k)})^\intercal h_1^\intercal  \\
    h_1 \beta_3^{(k)}
  \end{bmatrix}
\end{equation}

% \begin{equation}
%   \begin{bmatrix}
%     \mathbf{0} \\
%     \mu_{\mathcal{C}}^{(k)} - \mu_{\mathcal{C}}^{(i)} \\
%     \mu_2
%   \end{bmatrix} = 
%   \mu_1
%   \begin{bmatrix}
%     (B_1^{(k)})^\intercal h_1^\intercal  \\
%     (B_2^{(k)})^\intercal h_1^\intercal  \\
%     h_1 \beta_3^{(k)}
%   \end{bmatrix}
% \end{equation}

We get the following equations:
\begin{eqnarray}
  \mu_{\mathcal{C}}^{(k)} - \mu_{\mathcal{C}}^{(i)} &=& \mu_1 (B_2^{(k)})^\intercal h_1^\intercal \label{eqn:congested_multipliers_ik} \\
  \mu_2 &=& \mu_1 h_1 \beta_3^{(k)} \label{eqn:unqiue_congested_multipliers_ik} 
\end{eqnarray}
From the assumption $\lambda^{(i)} = \lambda^{(k)}$ we get $H^\intercal (\mu^{(i)} - \mu^{(k)}) = \mathbf{0}$. Since the shadow prices of the non-congested lines are zero:
\begin{equation}
\label{eqn:zero_congestion_diff}
  \mathbf{0} = H^\intercal (\mu^{(i)} - \mu^{(k)})
  = H_{\mathcal{C}}^\intercal (\mu_{\mathcal{C}}^{(i)} - \mu_{\mathcal{C}}^{(k)}) 
  + \mu_1 h_1^\intercal - \mu_2 h_2^\intercal 
\end{equation}
Using Eqn.(\ref{eqn:congested_multipliers_ik}):
\begin{equation}
\label{eqn:null_space_H_transpose}
  \mathbf{0}  = -\mu_1 H_{\mathcal{C}}^\intercal (B_2^{(k)})^\intercal h_1^\intercal + \mu_1 h_1^\intercal - \mu_2 h_2^\intercal
\end{equation}
From $(A_{\mathcal{B}_k})^{-1}A_{\mathcal{B}_k} = \mathbf{I}$:
\begin{equation}
\label{eqn:inverse_ABk}
  B_1^{(k)}E + B_2^{(k)}H_{\mathcal{C}} + \beta_3^{(k)}h_2 = \mathbf{I}
\end{equation}
Using Eqn. (\ref{eqn:inverse_ABk}), (\ref{eqn:unqiue_congested_multipliers_ik}) and $\mu_1 > 0$, Eqn. (\ref{eqn:null_space_H_transpose}) becomes:
\begin{equation}
  E^\intercal (B_1^{(k)})^\intercal h_1^\intercal = \mathbf{0}
\end{equation}
Since $E B_1^{(k)} = \mathbf{I}$, $(B_1^{(k)})^\intercal = (E E^\intercal)^{-1} E$:
\begin{equation}
  \mathbf{0} = E E^\intercal (B_1^{(k)})^\intercal h_1^\intercal = E h_1^\intercal = 0
\end{equation}
Given the structure of matrix $E$, we get the following equation:
\begin{equation}
  h_1 \mathbf{1}_{n_b} = 0
\end{equation}
This is not possible given the feature of the shift factor matrix $H$.

Therefore it is not possible that different congestion patterns have the same LMP vector.

\section{More Details About the SPRs of the 3-bus System in Fig. \ref{fig:3BusSystem3Gene} } % (fold)
\label{sec:analytical_form_of_the_sprs_of_the_3_bus_system_in_fig_}
This section provide complete details about the SPRs of the 3-bus 3-generator 2-load system in Fig. \ref{fig:3BusSystem3Gene}.
10 SPRs are visualized in Fig. \ref{fig:3Bus3GeneSystem_SPR}, and their corresponding system pattern, analytical form and LMP vectors are summarized in Table \ref{tab:details_spr}.

To better illustrate the concept \emph{system pattern} in Table \ref{tab:details_spr}, we provide the detailed formulation of the SCED problem of the 3-bus system below. According to definition \ref{defn:optimal_partition}, system pattern partitions all the constraints into two sets: binding constraints $\mathcal{B}$ and non-binding constraints $\mathcal{N}$. In Table \ref{tab:details_spr}, we only use the binding constraints to represent the system pattern\footnote{Since the non-binding constraints are just the complement of the index set $\{1,2,\cdots,14\}$.}, and each binding constraint is represented by its index. Also it is worth noticing that the supply-demand balance constraint is rewritten to be two inequality constraints, therefore it has indices $1$ and $2$\footnote{Since this constraint is always binding, therefore $1$ and $2$ are actually equivalent. One of them is redundant.}

\begin{equation*}
\begin{aligned}
& \underset{P_{G_1},P_{G_2},P_{G_3}}{\text{minimize}} & & 20P_{G_1}+50P_{G_2}+100P_{G_3} &\\
& \text{subject to} & & P_{G_1}, P_{G_2}+P_{G_3} = P_{D_2}+P_{D_2} & :1,2\\
& & & -\frac{2}{3}(P_{G_2}-P_{D_2}) -\frac{1}{3}(P_{G_3}-P_{D_3}) \le 60 &:3 \\
& & & -\frac{1}{3}(P_{G_2}-P_{D_2}) -\frac{2}{3}(P_{G_3}-P_{D_3}) \le 60 &:4 \\
& & & \frac{1}{3}(P_{G_2}-P_{D_2}) -\frac{1}{3}(P_{G_3}-P_{D_3}) \le 80 &:5 \\
& & & \frac{2}{3}(P_{G_2}-P_{D_2}) +\frac{1}{3}(P_{G_3}-P_{D_3}) \le 60 &:6 \\
& & & \frac{1}{3}(P_{G_2}-P_{D_2}) +\frac{2}{3}(P_{G_3}-P_{D_3}) \le 60 &:7 \\
& & & -\frac{1}{3}(P_{G_2}-P_{D_2}) +\frac{1}{3}(P_{G_3}-P_{D_3}) \le 80 &:8 \\
& & & P_{G_1} \le P_{G_1}^+ & :9\\
& & & P_{G_2} \le P_{G_2}^+ & :10\\
& & & P_{G_3} \le P_{G_3}^+ & :11\\
& & & P_{G_1}^- \le P_{G_1} & :12\\
& & & P_{G_2}^- \le P_{G_2} & :13\\
& & & P_{G_3}^- \le P_{G_3} & :14\\
\end{aligned}
\end{equation*}

\begin{figure}[htbp]
  \centering
  \includegraphics[width=0.6\linewidth]{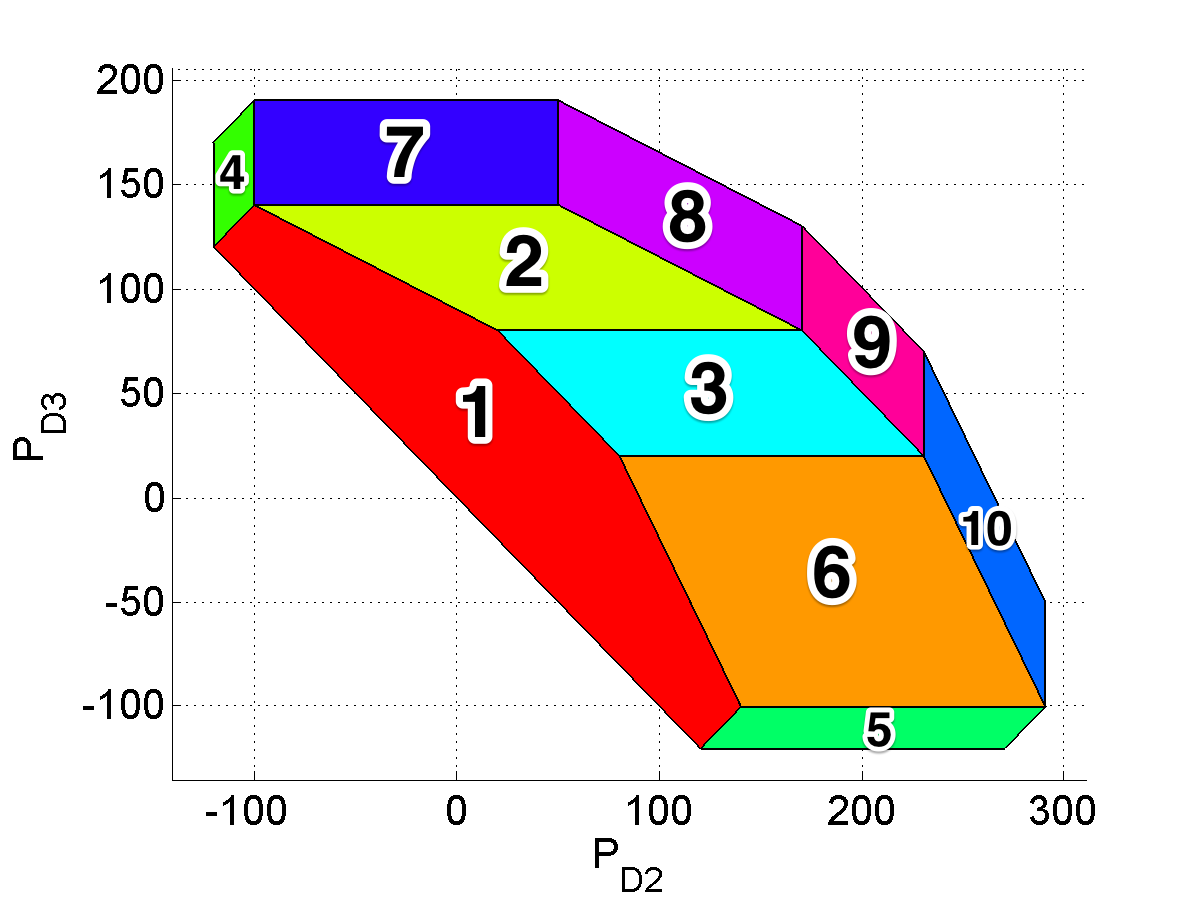}
  \caption{SPRs of the 3-bus System in Fig. \ref{fig:3BusSystem3Gene}}
  \label{fig:3Bus3GeneSystem_SPR}
\end{figure}

\begin{longtable}{l|ccc}
  \caption{Details of the SPRs in Fig. \ref{fig:3Bus3GeneSystem_SPR}}
  \label{tab:details_spr}   \\

  \hline

  \hline
  \textbf{SPR No.} & \textbf{System Pattern} & \textbf{Analytical Form of The SPRs} & \textbf{LMPs} \\
  \hline
  & & & \\
  1& 
$\begin{bmatrix}
    1\\
     2\\
    13\\
    14
\end{bmatrix}$
  & $\begin{bmatrix}  
    0.8944  &  0.4472\\
    0.4472  &  0.8944\\
   -0.7071 &   0.7071\\
    0.7071  & -0.7071\\
    0.7071 &   0.7071\\
   -0.7071  & -0.7071 
  \end{bmatrix}\begin{bmatrix}
    P_{D_2} \\ P_{D_3}
  \end{bmatrix} \le \begin{bmatrix}
       80.4984\\
   80.4984\\
  169.7056\\
  169.7056\\
   70.7107\\
         0
  \end{bmatrix}$ & 
  $\begin{bmatrix}
    20\\
    20\\
    20
  \end{bmatrix}$ \\
  & & & \\
  2& 
$\begin{bmatrix}
  1\\
  2\\
  4\\
  14
\end{bmatrix}$
   & $\begin{bmatrix}  
         0   & 1.0000\\
         0  & -1.0000\\
    0.4472  &  0.8944\\
   -0.4472  & -0.8944
  \end{bmatrix}\begin{bmatrix}
    P_{D_2} \\ P_{D_3}
  \end{bmatrix} \le \begin{bmatrix}
  140.0000\\
  -80.0000\\
  147.5805\\
  -80.4984
  \end{bmatrix}$ &   $\begin{bmatrix}
    20\\
    50\\
    80
  \end{bmatrix}$ \\ 
  & & & \\
  3& 
  $\begin{bmatrix}
  1\\
  2\\
  9\\
  14
\end{bmatrix}$
& $\begin{bmatrix}  
         0  & -1.0000\\
         0  &  1.0000\\
    0.7071  &  0.7071\\
   -0.7071  & -0.7071
  \end{bmatrix}\begin{bmatrix}
    P_{D_2} \\ P_{D_3}
  \end{bmatrix} \le \begin{bmatrix}
  -20.0000\\
   80.0000\\
  176.7767\\
  -70.7107
  \end{bmatrix}$ &   $\begin{bmatrix}
    50\\
    50\\
    50
  \end{bmatrix}$ \\
  & & & \\
  4& 
  $\begin{bmatrix}
  1\\
  2\\
  3\\
  14
\end{bmatrix}$
& $\begin{bmatrix}  
    1.0000 &        0\\
    0.7071 & -0.7071\\
   -0.7071 &   0.7071\\
   -1.0000  &       0
  \end{bmatrix}\begin{bmatrix}
    P_{D_2} \\ P_{D_3}
  \end{bmatrix} \le \begin{bmatrix}
 -100.0000\\
 -169.7056\\
  205.0610\\
  120.0000
  \end{bmatrix}$ &   $\begin{bmatrix}
    20\\
    50\\
    35
  \end{bmatrix}$ \\
  & & & \\
  5&
  $\begin{bmatrix}
  1\\
  2\\
  8\\
  14
\end{bmatrix}$
 &$\begin{bmatrix}  
         0  &  1.0000\\
    0.7071 &  -0.7071\\
         0 &  -1.0000\\
   -0.7071  &  0.7071
  \end{bmatrix}\begin{bmatrix}
    P_{D_2} \\ P_{D_3}
  \end{bmatrix} \le \begin{bmatrix}
 -100.0000\\
  275.7716\\
  120.0000\\
 -169.7056
  \end{bmatrix}$ &   $\begin{bmatrix}
    20\\
    50\\
    -10
  \end{bmatrix}$ \\
  & & & \\
  6&
  $\begin{bmatrix}
  1\\
  2\\
  5\\
  13
\end{bmatrix}$
 &$\begin{bmatrix}  
    1.0000   &      0\\
    0.7071  & -0.7071\\
   -0.7071  &  0.7071\\
   -1.0000  &       0
  \end{bmatrix}\begin{bmatrix}
    P_{D_2} \\ P_{D_3}
  \end{bmatrix} \le \begin{bmatrix}
 -100.0000\\
 -169.7056\\
  205.0610\\
  120.0000
  \end{bmatrix}$ &  $\begin{bmatrix}
    20\\
    -60\\
    100
  \end{bmatrix}$ \\
  & & & \\
  7&
  $\begin{bmatrix}
  1\\
  2\\
  4\\
  5
\end{bmatrix}$
 &$\begin{bmatrix}  
     0   &-1\\
     1  &  0\\
     0  &   1\\
    -1  &   0
  \end{bmatrix}\begin{bmatrix}
    P_{D_2} \\ P_{D_3}
  \end{bmatrix} \le \begin{bmatrix}
 -140.0000\\
   50.0000\\
  190.0000\\
  100.0000
  \end{bmatrix}$ &   $\begin{bmatrix}
    20\\
    50\\
    100
  \end{bmatrix}$ \\
  & & & \\
  8&
  $\begin{bmatrix}
  1\\
  2\\
  4\\
  10
\end{bmatrix}$
 & $\begin{bmatrix}  
   -1.0000  &       0\\
   -0.4472  & -0.8944\\
    1.0000 &        0\\
    0.4472 &   0.8944
  \end{bmatrix}\begin{bmatrix}
    P_{D_2} \\ P_{D_3}
  \end{bmatrix} \le \begin{bmatrix}
  -50.0000\\
 -147.5805\\
  170.0000\\
  192.3018
  \end{bmatrix}$&   $\begin{bmatrix}
    20\\
    60\\
    100
  \end{bmatrix}$ \\
  & & & \\
  9&
  $\begin{bmatrix}
  1\\
  2\\
  9\\
  10
\end{bmatrix}$
 & $\begin{bmatrix}  
    1.0000 &        0\\
   -1.0000 &        0\\
   -0.7071 &  -0.7071\\
    0.7071 &   0.7071
  \end{bmatrix}\begin{bmatrix}
    P_{D_2} \\ P_{D_3}
  \end{bmatrix} \le \begin{bmatrix}
  230.0000\\
 -170.0000\\
 -176.7767\\
  212.1320
  \end{bmatrix}$ &   $\begin{bmatrix}
    100\\
    100\\
    100
  \end{bmatrix}$  \\   
  & & & \\
10 &
$\begin{bmatrix}
  1\\
  2\\
  3\\
  10
\end{bmatrix}$
 &$\begin{bmatrix}  
    1.0000  &       0\\
   -0.8944  & -0.4472\\
   -1.0000  &       0\\
    0.8944  &  0.4472
  \end{bmatrix}\begin{bmatrix}
    P_{D_2} \\ P_{D_3}
  \end{bmatrix} \le \begin{bmatrix}
  290.0000\\
 -214.6625\\
 -230.0000\\
  237.0232
  \end{bmatrix}$ &   $\begin{bmatrix}
    20\\
    180\\
    100
  \end{bmatrix}$ \\
  & & & \\
  \hline

  \hline
  % \end{tabular}
\end{longtable}
% We can calculate the conditions on the generation offer prices for each SPR using Eqn. (\ref{eqn:determine_y}) in Lemma \ref{lem:complementary_slackness}. The results are demonstrated in Fig. \ref{fig:SPR_Offer_Prices}\footnote{The index of each SPR here is different from the indices in Fig. \ref{fig:3Bus3GeneSystem_SPR}, please use the system pattern to compare the results in Fig. \ref{fig:3Bus3GeneSystem_SPR}, Table \ref{tab:details_spr} and Fig. \ref{fig:SPR_Offer_Prices} }.

% \begin{figure}[htbp]
%   \centering
%   \includegraphics[width=0.45\linewidth]{fig/all_spr/IN_SPR_6.png}
%   \includegraphics[width=0.45\linewidth]{fig/all_spr/IN_SPR_10.png}\\
%   \includegraphics[width=0.45\linewidth]{fig/all_spr/IN_SPR_11.png}
%   \includegraphics[width=0.45\linewidth]{fig/all_spr/IN_SPR_15.png}
%   \caption{SPRs and Corresponding Requirements on the Generation Offer Prices}
%   \label{fig:SPR_Offer_Prices}
% \end{figure}

\section{System Pattern Regions in 3D Space} % (fold)
\label{sec:system_pattern_regions_in_3d_space}
For better illustration, we only visualized the 2-dimension SPRs in previous sections. But the SPRs are usually polyhedrons in high-dimension space. The visualization of the 3-dimension SPRs could help readers get more intuition on the high-dimension SPR/polyhedron.
One load is added to the 3-bus system in Fig. \ref{fig:3Bus2GeneSystem}, and the new 3-bus 2-generator 3-load system is shown in Fig. \ref{fig:3Bus2Gen3Load}. Since there are 3 loads in the system, the SPRs locate in the 3D space. Similar with the 2D case, the SPRs are polyhedrons and there exists a separating hyperplane between any two SPRs.

\begin{figure}[htbp]
  \centering
  \includegraphics[width=0.6\linewidth]{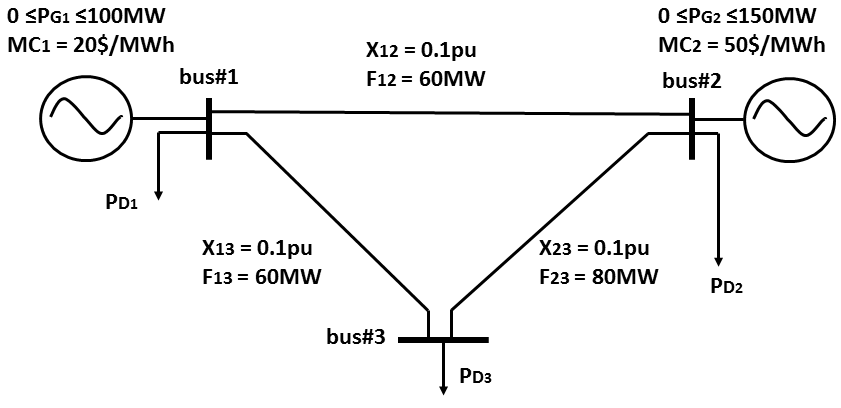}
  \caption{3-bus 2-generator 3-load System}
  \label{fig:3Bus2Gen3Load}
\end{figure}

\begin{figure}[htbp]
  \centering
  \includegraphics[width=\linewidth]{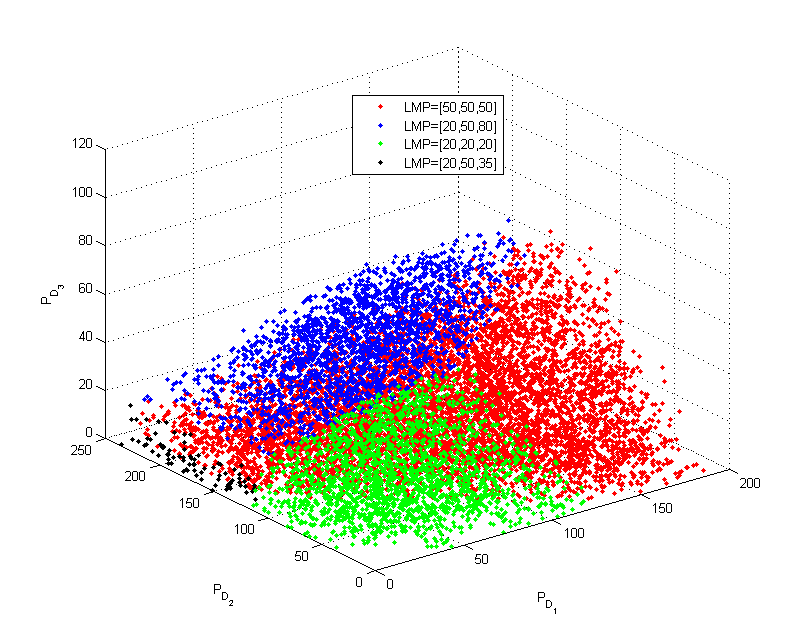}
  \caption{3D System Pattern Regions}
\end{figure}

\begin{figure}[htbp]
  \centering
  \includegraphics[width=0.45\linewidth]{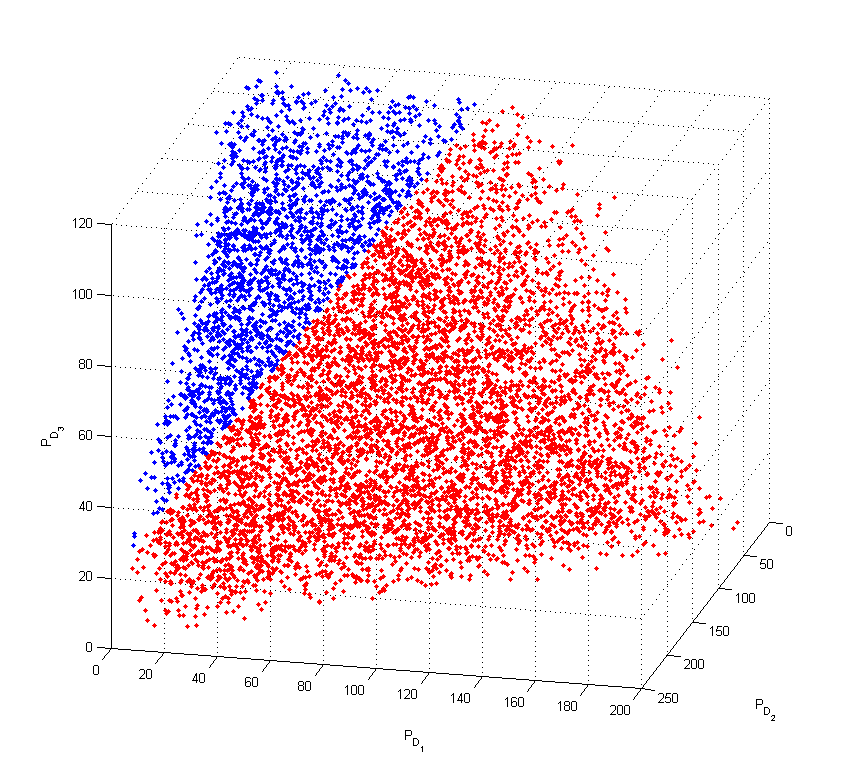}
  \includegraphics[width=0.45\linewidth]{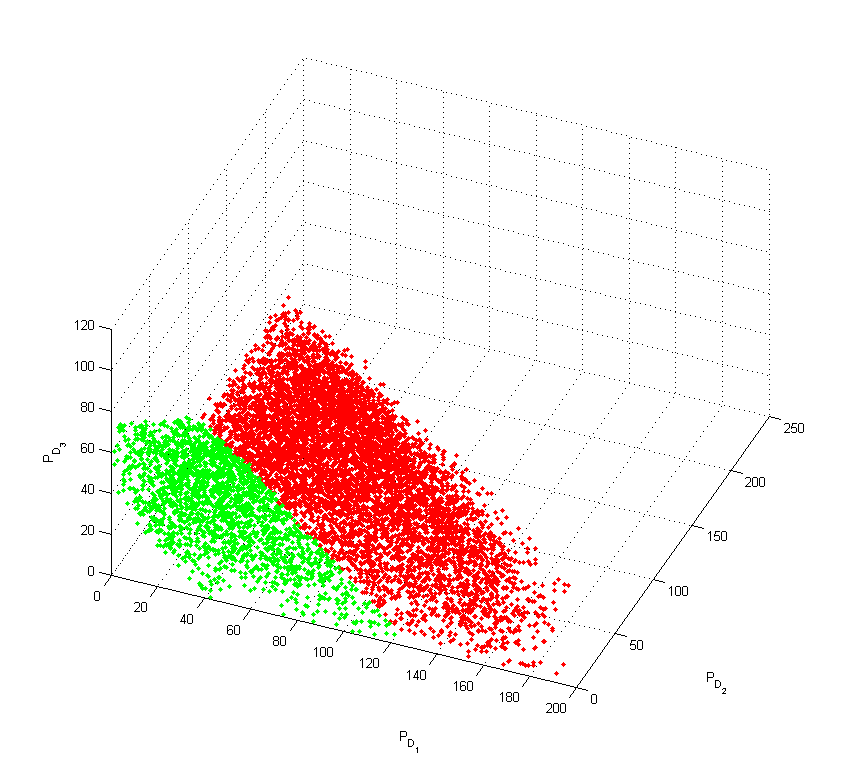} \\
  \includegraphics[width=0.45\linewidth]{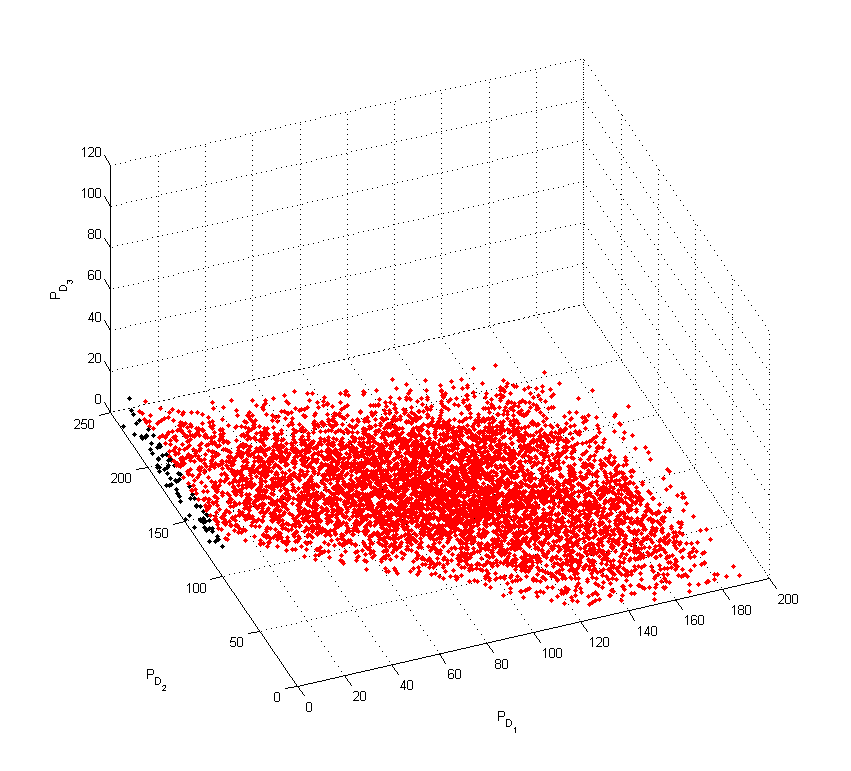} 
  \includegraphics[width=0.45\linewidth]{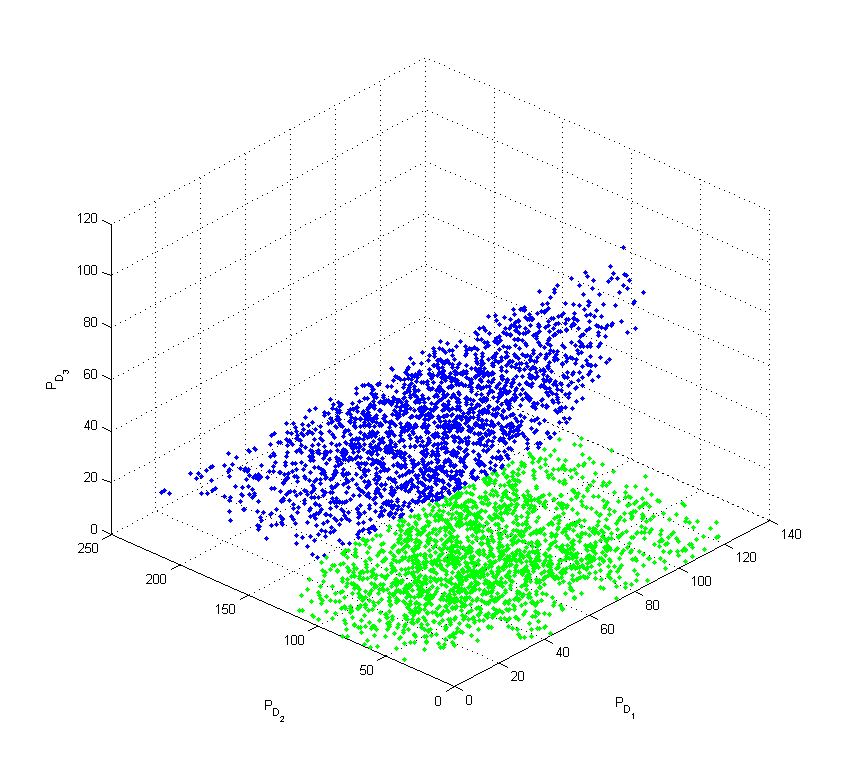} \\
  \includegraphics[width=0.45\linewidth]{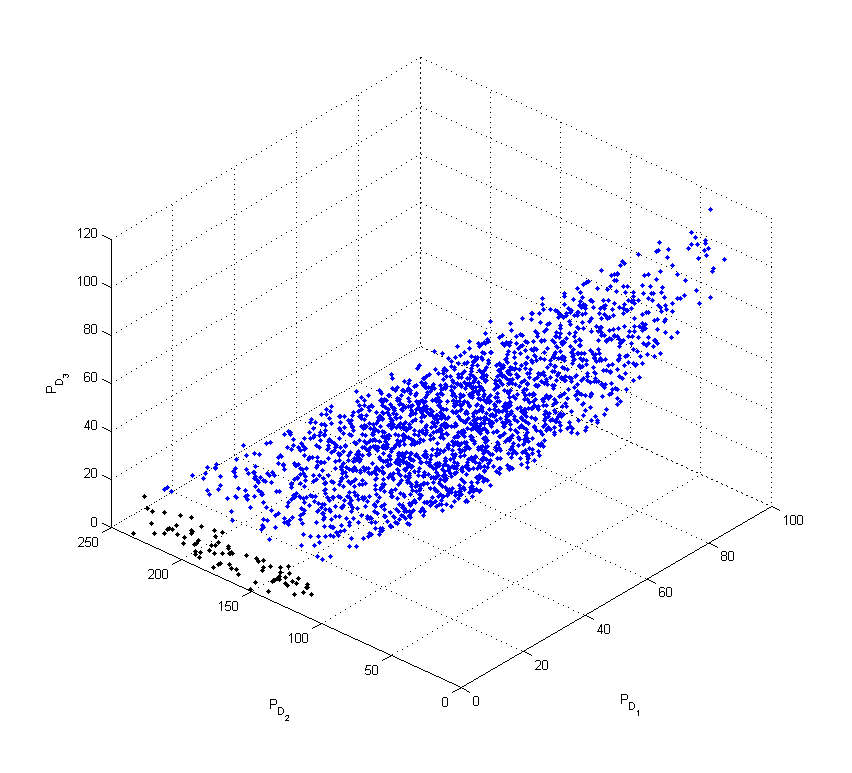}
  \includegraphics[width=0.45\linewidth]{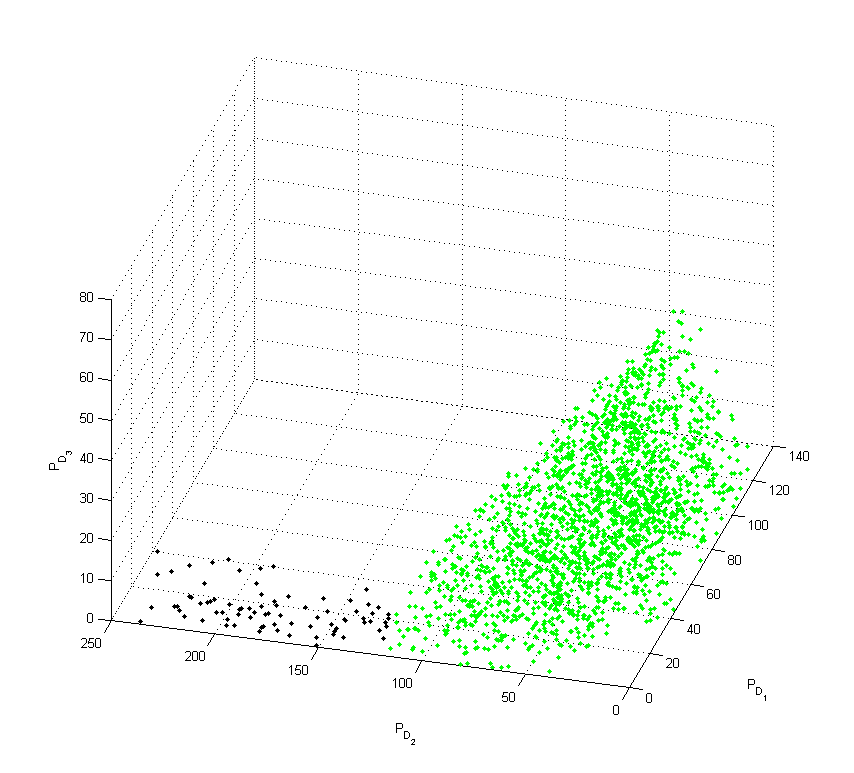}
  \caption{Pairwise Visualization of System Pattern Regions}
\end{figure}
% section system_pattern_regions_in_3d_space (end)

% biography section
\newpage
\bibliographystyle{IEEEtran}
\bibliography{MyCollection}

\end{document}